\documentclass[journal]{IEEEtran}

 \usepackage{pdfsync}

\usepackage{cite}

\ifCLASSINFOpdf
\else
\fi
\usepackage{amsmath}
\usepackage{algorithm}
\usepackage{algorithmic}

\usepackage{array}
\usepackage{amsmath}
\usepackage{amsthm}
\usepackage{amsfonts}
\usepackage{amssymb}
\usepackage{mathrsfs}
\usepackage{bm}
\usepackage{nicefrac}
\usepackage{color}
\usepackage{cuted}
\usepackage{mathtools,lipsum}

\usepackage[pdftex]{graphicx}
\usepackage{comment}
\usepackage{subcaption}

\newtheoremstyle{mystyle}{1pt}{1pt}{\normalfont}{\parindent}{\bfseries}{:}{1em}{}
\theoremstyle{mystyle}
\newtheorem{Thm}{Theorem}
\newtheorem{Asu}{Assumption}
\newtheorem{Pro}{Proposition}

\newtheorem{Lem}{Lemma}
\newtheorem{Col}{Corollary}
\newtheorem{Rem}{Remark}

\allowdisplaybreaks[4]

\begin{document}
%
\title{Distributed Policy Gradient with Variance Reduction in Multi-Agent Reinforcement Learning}
%
%
%

\author{Xiaoxiao Zhao, ~
        Jinlong Lei,~\IEEEmembership{Member,~IEEE,}
        ~Li Li,~\IEEEmembership{Member,~IEEE}
        and ~Jie Chen, ~\IEEEmembership{Fellow,~IEEE}
\thanks{This work was supported by Shanghai Municipal Science
	and Technology Major Project under Grant 2021SHZDZX0100, National Key Research
	and Development Program of Science and Technology of China under
	Grant 2018YFB1305304, Shanghai Science and Technology Pilot
	Project under Grant 19511132100,  National Natural Science Foundation of China under Grant 72171172.}
\thanks{X. Zhao is with the College
of Electronic and Information Engineering, Tongji University, Shanghai,
201804 China. (e-mail: zh\_xiaoxiao@tongji.edu.cn).
J. Lei, L. Li, and J. Chen are with the College of Electronic and Information Engineering, Tongji University, Shanghai, 201804 China, and Institute of Intelligent Science and Technology, Tongji University, Shanghai 201203, China. (e-mail: leijinlong@tongji.edu.cn,~lili@tongji.edu.cn,~chenjie@bit.edu.cn).}
\thanks{Manuscript received xxxx xx}}

%
%

\markboth{Manuscript }%
{X. Zhao \MakeLowercase{\textit{et al.}}: Distributed Policy Gradient with Variance Reduction in Multi-Agent Reinforcement Learning}

%


\maketitle

\begin{abstract}
This paper studies a distributed policy gradient in collaborative multi-agent reinforcement learning (MARL), where agents over a communication network aim to find the optimal policy to maximize the average of all agents' local returns. Due to the non-concave performance function of policy gradient, the existing distributed stochastic optimization methods for convex problems cannot be directly used for policy gradient in MARL. This paper proposes a distributed policy gradient with variance reduction and gradient tracking to address the high variances of policy gradient, and utilizes importance weight to solve the {distribution shift} problem in the sampling process. We then provide an upper bound on the mean-squared stationary gap, which depends on the number of iterations, the mini-batch size, the epoch size, the problem parameters, and the network topology. We further establish the sample and communication complexity to obtain an $\epsilon$-approximate stationary point.
Numerical experiments are performed to validate the effectiveness of the proposed algorithm.
\end{abstract}


%
\IEEEpeerreviewmaketitle

\section{Introduction}
%
%
%
%

\IEEEPARstart{R}{einforcement} learning (RL) searches for the optimal policy through the dynamic interaction between an agent and the environment \cite{sutton2018reinforcement, bertsekas2019reinforcement}. Multi-agent reinforcement learning (MARL) incorporates the idea of RL into multi-agent systems \cite{yan2022optimal, duan2022optimal}, where a common environment is influenced by the joint actions of multiple agents \cite{Busoniu2008survey}. MARL achieves significant success in many complex decision-making tasks, such as intelligent traffic systems \cite{lin2018efficient}, resource allocation \cite{chen2019iraf}, and networked system control \cite{li2012optimal}, etc. Since agents interact not only with the environment but also with other agents, MARL suffers from several challenges including the non-stationarity, partial observability, scalability issues, and various information structures\cite{zhang2019multi}.

For collaborative MARL, the agents share a common aim of maximizing the globally averaged return of all agents. Especially, the reward functions are private to each agent and might vary for heterogeneous agents. One straightforward choice is to have a central controller that gathers the rewards of all agents and decides the actions for each agent. However, a central controller may be too expensive to deploy, and may also degrade the robustness to malicious attacks. Therefore, we focus on the fully distributed scheme, where agents are connected through a communication network and each agent takes an individual action based on local/neighboring information. Particularly, we study the policy gradient, one of  the most popular approaches to search for the optimal policy in high dimensional continuous action space \cite{sutton2000policy}.

Policy gradient methods parameterize the policy with an unknown parameter $\boldsymbol{\theta} \in \mathbb{R}^d$ and find the optimal parameter $\boldsymbol{\theta}^{*}$ to directly optimize the policy. The objective function $J\left(\boldsymbol{\theta}\right)$ of policy gradient is the expected return under a given policy and is usually non-concave. The goal is to find a stationary point $\boldsymbol{\theta}^{*}$ using gradient-based algorithms such that $\left\|\nabla J(\boldsymbol{ \theta }^{*})\right\|^2 = 0 $.
Since $J\left(\boldsymbol{\theta}\right)$ is expectation-valued, it is impracticable to calculate the exact gradients. Stochastic gradient estimators such as REINFORCE \cite{williams1992simple} and G(PO)MDP \cite{baxter2001infinite} have been applied to approximate the gradient via sampled trajectories. However, such approximation introduces high variances and slows down the convergence. To reduce the high variance of policy gradient approaches, we reformulate the problem of multi-agent policy gradient as a distributed stochastic optimization problem, and propose a fully distributed policy gradient algorithm with variance reduction and gradient tracking in MARL. At the nucleus of the proposed algorithm is the local gradient trackers to track the average of the variance-reduced gradient estimators across agents. Moreover, we provide the theoretical guarantees that the proposed algorithm finds an $\epsilon$-approximate stationary point of the non-concave performance function.

\subsection{Related Work}
\subsubsection{Multi-agent reinforcement learning}
Existing researches on MARL are mainly based on the framework of Markov games proposed by Littman \cite{littman1994markov}. However, most of these early works only consider the tabular setting, which suffers from the curse of dimensionality with large action-state space. To solve this issue, multi-agent deep reinforcement learning has received increasing attention \cite{hernandez2019survey, nguyen2020deep, foerster2016learning, lowe2017multi, foerster2019bayesian, vinyals2019grandmaster}, where deep neural networks are trained to approximate the policy or value function. But most of these efforts focus on empirical performance while are lack of theoretical guarantees. In addition, they mainly use the centralized scheme or ``centralized training and decentralized execution", and ignore the importance of information exchange across agents.

Recently, distributed MARL approaches with theoretical guarantees have been studied \cite{lee2020optimization}. Macua et al. \cite{macua2014distributed} applied diffusion strategies to develop a fully distributed gradient temporal-difference (GTD) algorithm, then provided a mean-square-error performance analysis and established the convergence under constant step size updates. Besides, Lee et al. \cite{lee2018primal} studied a new class of distributed GTD algorithm based on primal-dual iterations, and proved that it almost surely converged to a set of stationary points using ODE-based methods. In addition, Wai et al. \cite{wai2018multi} proposed a decentralized primal-dual optimization algorithm with a double averaging update scheme to solve the policy evaluation problem in MARL, and established the global geometric rate of convergence. Doan et al. \cite{doan2019finite} proposed a distributed TD(0) algorithm for solving the policy evaluation problem in MARL, and provided a finite-time convergence analysis over time-varying networks.
Recently, Cassano et al. \cite{cassano2020multi} developed a fully decentralized multi-agent algorithm for policy evaluation by combining off-policy learning and linear function approximation, and provided the linear convergence analysis.
{While the above works mainly focus on the policy evaluation problem, which estimates the value function of a given policy without finding the optimal policy.} The proposed distributed policy gradient algorithm in this paper can obtain the optimal policy for decision-making tasks in MARL.

\subsubsection{Policy gradient}
To reduce the high variance of policy gradient, traditional approaches introduced baseline functions \cite{sutton2000policy} or used function approximation for estimating the value function (see e.g., the actor-critic algorithms \cite{konda2000actor, bhatnagar2009natural}). The idea of variance reduction was first proposed to accelerate stochastic optimization. Variance reduction algorithms such as SAG \cite{schmidt2017minimizing}, SVRG \cite{johnson2013accelerating}, SAGA \cite{defazio2014saga} and SARAH \cite{nguyen2017sarah} achieved superior performance for both convex and non-convex problems \cite{allen2016variance, reddi2016stochastic, li2018simple, ge2019stabilized}. Motivated by the advantages in stochastic optimization, some researchers incorporated the variance reduction into policy gradient. Papini et al. \cite{papini2018stochastic} proposed a stochastic variance-reduced policy gradient (SVRPG) by borrowing the idea of SVRG and provided convergence guarantees. Xu et al. \cite{xu2020improved} established an improved convergence analysis of SVRPG \cite{papini2018stochastic} and showed it can find an $\epsilon$-approximate stationary point. They further proposed a stochastic recursive variance reduced policy gradient (SRVR-PG) \cite{xu2020sample} to reduce the sample complexity. {However, these approaches with theoretical guarantees are only proposed for single-agent. Since multiple agents learn and update in a common environment, the inherent challenges in multi-agent settings would make these methods unsuitable.} This paper proposes a distributed policy gradient with variance reduction and gradient tracking in MARL.

\subsubsection{Distributed stochastic optimization}
Recent breakthrough in stochastic variance reduced methods has made it possible to achieve better performance in distributed stochastic optimization problems. Decentralized double stochastic averaging (DSA) gradient algorithm \cite{mokhtari2016dsa} was the first decentralized variance reduction method by combining EXTRA \cite{ShiLWY15} and SAGA \cite{defazio2014saga}, and linear convergence was shown for the strong convexity of local functions and Lipschitz continuity of local gradients.
Later on, Yuan et al. \cite{yuan2019variance} developed a fully decentralized variance-reduced stochastic gradient algorithm named diffusion-AVRG (amortized variance-reduced gradient), which displayed a linear convergence to the exact solution and was more memory efficient than DSA \cite{mokhtari2016dsa}. Xin et al.\cite{xin2020decentralized} proposed a unified framework for variance-reduced decentralized stochastic methods that utilize gradient tracking \cite{qu2017harnessing, nedic2017achieving, pu2020distributed} to obtain robust performance. In particular, two algorithms GT-SAGA and GT-SVRG \cite{xin2020variance} showed accelerated linear convergence without computing the expensive dual gradients. {However, the above distributed stochastic optimization approaches with variance reduction are only applicable for convex functions.} Since the objective functions of policy gradient are non-concave, the aforementioned  methods cannot be directly used for the distributed policy gradient.

Recently, Zhang et al. \cite{zhang2019decentralized} proposed a decentralized stochastic gradient tracking algorithm for non-convex empirical risk minimization problems, and showed  {that} the convergence can be comparable to the centralized SGD method.
Besides, Xin et al. \cite{xin2021improved} developed a decentralized stochastic gradient descent algorithm with gradient tracking for online stochastic non-convex optimization, and established the convergence with constant and decaying step sizes.
Since {policy gradient suffers from the challenge of distribution shift},  the above methods \cite{zhang2019decentralized, xin2021improved} cannot be directly  applied to the policy gradient in MARL. 

\subsection{Contribution}
The main contributions of this paper are summarized as follows.
We first formulate the policy gradient in MARL as a distributed stochastic optimization problem. We then propose a distributed stochastic policy gradient algorithm with variance reduction and gradient tracking, where the importance weight is incorporated to deal with the distribution shift problem in the sampling process such that it can maintain the unbiased property of the variance-reduced gradient estimators.
We provide the theoretical guarantees that the proposed algorithm can converge to a stationary point of the non-concave performance function, and show the convergence rate in terms of the number of iterations, the mini-batch size, the epoch size, the problem parameters,  and the network topology.
We further establish that the sample and communication complexity are $O\left(\frac{1}{\epsilon^{\frac{3}{2}}}\right)$ and $O\left(\left|\mathcal{E}\right|\frac{1}{\epsilon}\right)$, respectively, for finding an $\epsilon$-approximate stationary point such that $\mathbb{E}[\left\|\nabla J \left({\boldsymbol{\theta}}\right)\right\|^2] \le \epsilon$.

\subsection{Notation and Organization}
Throughout the rest of this paper, we use lowercase bold letters $\boldsymbol{x} \in \mathbb{R}^n$ to denote the vectors and uppercase bold letters $\boldsymbol{Y} \in \mathbb{R}^{m \times n}$ to denote the matrices.
We use $\|\boldsymbol{x}\|$ and $\left\|\boldsymbol{Y}\right\|$
to represent the Euclidean norm of vector $\boldsymbol{x}$ and the induced 2-norm of matrix $\boldsymbol{Y}$, respectively.  {We denote by} $\boldsymbol{I}_d$ the $d \times d$ identity matrix, and  {by} $\boldsymbol{1}_{d}$ the $d$-dimensional all one column vector. The Kronecker product of two matrices $\boldsymbol{A}$, $\boldsymbol{B} \in \mathbb{R}^{d \times d}$ is denoted by $\boldsymbol{A} \otimes \boldsymbol{B}$.
We   use $\mathcal{A}$ to denote a finite set.
Let $\mathbb{E}_{g}\left[ \cdot \right]$ denote  the expected value with respect to distribution $g$. Denote $a_{n}=O\left(b_{n}\right)$ if $a_{n} \le Cb_{n}$ for some constant $C > 0$.

The paper is organized as follows. In Section II, we introduce the problem formulation of policy gradient in MARL. We propose a distributed policy gradient algorithm with variance reduction and gradient tracking in Section III, and explore its convergence properties in Section IV. We provide experimental results in Section V. Section VI concludes this paper.

\section{Problem Formulation}
In this section, we first introduce the multi-agent Markov decision process (MDP), and then formulate the policy gradient in MARL as a distributed stochastic optimization problem over networks.

\subsection{Multi-agent MDP}
A multi-agent MDP is characterized by a tuple $\left( \mathcal{S},\left\{\mathcal{A}_{i}\right\}_{i=1}^{n}, P,\left\{R_{i}\right\}_{i=1}^{n}, \gamma\right)$, where $n$ denotes the number of agents, $\mathcal{S}$ is the environmental state space, and $\mathcal{A}_{i}$ is the action space of agent $i$.
Denote $\mathcal{A}=\mathcal{A}_{1} \times \cdots \mathcal{A}_{n}$ as the joint action space of all agents.
In addition, let $\boldsymbol{s} \in \mathcal{S}$ be the global state shared by all agents, $\boldsymbol{a}_{i} \in \mathcal{A}_{i}$ be the action executed by agent $i$, and $\boldsymbol{a}=\left(\boldsymbol{a}_{1}, \ldots, \boldsymbol{a}_{n}\right) \in \mathcal{A}_{1} \times \cdots \times \mathcal{A}_{n}$ be the joint action, respectively. Then the reward function $R_{i}: \mathcal{S} \times \mathcal{A} \rightarrow \mathbb{R}$ is the local reward of agent $i$, and $P\left(\boldsymbol{s}^{\prime}|\boldsymbol{s}, \boldsymbol{a}\right)$ represents the state transition probability from state $\boldsymbol{s}$ to $\boldsymbol{s}^{\prime}$ after taking a joint action $\boldsymbol{a}$. $\gamma \in \left( 0,1 \right) $ is the discount factor.

At time $t$, each agent $i$ selects its action $\boldsymbol{a}_i^t$ given state $\boldsymbol{s}^t$ following a local policy $\pi_i:\mathcal{S} \times \mathcal{A}_i \rightarrow[0,1]$, where $\pi_i \left(\boldsymbol{a}_i|\boldsymbol{s}\right)$ is the probability that agent $i$ selects action $\boldsymbol{a}_i$ at state $\boldsymbol{s}$. Let $\pi: \mathcal{S} \times \mathcal{A} \rightarrow[0,1]$ be a joint policy, and $\pi (\boldsymbol{a}|\boldsymbol{s})=\prod_{i=1}^{n} \pi_i \left(\boldsymbol{a}_i|\boldsymbol{s}\right)$  {be} the probability to choose a joint action $\boldsymbol{a}$ at state $\boldsymbol{s}$.

In episodic task, a trajectory $\tau$ is a sequence of state-action pairs $\left\{\boldsymbol{s}^0,\boldsymbol{a}^0,\boldsymbol{s}^1,\boldsymbol{a}^1,\ldots \boldsymbol{s}^{H-1},\boldsymbol{a}^{H-1},\boldsymbol{s}^{H}\right\}$ observed by following any stationary policy $\pi$ up to time horizon $H$. In the cooperative settings, the goal is to
maximize the collective return of all agents, that is  $R_c(\tau)=\sum_{i=1}^{n}\sum_{h=0}^{H-1}\gamma^{h}R_i\left(\boldsymbol{s}^h,\boldsymbol{a}^h\right)$  {at} trajectory $\tau$.
We assume that both $\boldsymbol{a}$ and $\boldsymbol{s}$ are available to all agents, whereas the rewards $R_{i}$ are observed only locally by  {agent $i$.}
Therefore, it is essential for agents to cooperate with other agent to solve the multi-agent problem based on local information.

\subsection{Multi-agent Policy Gradient}
Suppose the policy is parameterized by an unknown parameter $\boldsymbol{\theta} \in \mathbb{R}^d$ denoted by $\pi_{\boldsymbol{\theta}}$, which is differentiable with respect to  $\boldsymbol{\theta}$.
We denote the trajectory distribution induced by policy $\pi_{\boldsymbol{\theta}}$ as $p\left(\tau \left|\boldsymbol{\theta}\right.\right)$,
\begin{equation}\label{eq1}
p\left(\tau \left|\boldsymbol{\theta}\right.\right) = d\left(\boldsymbol{s}^{0}\right) \prod_{h=0}^{H-1} \pi_{\boldsymbol{\theta}}\left(\boldsymbol{a}^h | \boldsymbol{s}^h\right) P\left(\boldsymbol{s}^{h+1} |\boldsymbol{s}^{h}, \boldsymbol{a}^{h} \right),
\end{equation}
where $d\left(\boldsymbol{s}^{0}\right)$ is the distribution of initial state.
The performance function under policy $\pi_{\boldsymbol{\theta}}$ is measured by the expected discounted return $J\left(\boldsymbol{\theta}\right)$,
\begin{equation}\label{eq2}
J\left(\boldsymbol{\theta}\right)=\mathbb{E}_{\tau \sim p\left(\cdot \left|\boldsymbol{\theta}\right.\right)}\left[R_c(\tau)\right].
\end{equation}
Maximizing the performance function $J\left(\boldsymbol{\theta}\right)$ can be obtained through gradient ascent algorithm $\boldsymbol{\theta}^{k+1}=\boldsymbol{\theta}^{k}+\alpha {\nabla}_{\boldsymbol{\theta}}J(\boldsymbol{\theta}^{k})$, where $\alpha > 0$ is the step size and the gradient $\nabla_{\boldsymbol{\theta}} J(\boldsymbol{\theta})$ is calculated as follows,
\begin{equation}\label{eq3}
\nabla_{\boldsymbol{\theta}} J(\boldsymbol{\theta})
=\mathbb{E}_{\tau \sim p\left(\cdot \left|\boldsymbol{\theta}\right.\right)}\left[\nabla_{\boldsymbol{\theta}} \log \pi_{\boldsymbol{\theta}}(\tau) R_c(\tau)\right].
\end{equation}
It is difficult to compute the exact gradient in (\ref{eq3}) because the trajectory distribution is unknown.
In practice, an approximate gradient estimator $\hat{\nabla}_{\boldsymbol{\theta}}J(\boldsymbol{\theta})$ using a batch of sampled trajectories $\{\tau_j\}_{j=0}^{M}$ from $\pi_{\boldsymbol{\theta}}$ is applied,
\begin{equation}\label{eq4}
\hat{\nabla}_{\boldsymbol{\theta}}J(\boldsymbol{\theta})
=\frac{1}{M}\sum_{j=1}^{M}g\left(\tau_j|\boldsymbol{\theta}\right),
\end{equation}
where $g\left(\tau_j|\boldsymbol{\theta}\right)$ is an estimator of $\nabla_{\boldsymbol{\theta}} J(\boldsymbol{\theta})$ using trajectory $\tau_j$.
Then the stochastic gradient ascent algorithm below is typically utilized to update the policy parameter $\boldsymbol{\theta}$,
\begin{equation}\label{eq5}
\boldsymbol{\theta}^{k+1}=\boldsymbol{\theta}^{k}+\alpha\hat{\nabla}_{\boldsymbol{\theta}}J(\boldsymbol{\theta}^{k}).
\end{equation}
The widely used unbiased estimators of policy gradient include REINFORCE \cite{williams1992simple} and G(PO)MDP \cite{baxter2001infinite}.
Since G(PO)MDP has lower variance than REINFORCE and the variance is independent of $H$, we use G(PO)MDP estimator as the policy gradient estimator, i.e,
\begin{align}
&g\left(\tau_j| \boldsymbol{\theta}\right) \notag \\
= &\sum_{h=0}^{H-1} \left(\sum_{t=0}^{h}\nabla_{\boldsymbol{\theta}}\log \pi_{\boldsymbol{\theta}}(\boldsymbol{a}^{t}_{(j)} \left| \boldsymbol{s}^{t}_{(j)}\right.)\right)\left(\gamma^{h}R(\boldsymbol{s}^{h}_{(j)},\boldsymbol{a}^{h}_{(j)})-b^{h}\right), \label{eq6}
\end{align}
where $\boldsymbol{a}^{t}_{(j)}$ and $\boldsymbol{s}^{t}_{(j)}$ are the  {joint} action and state at time $t$ in trajectory $j$, respectively, and $b$ is the constant baseline.

Since each agent $i$ can only observe its local reward $R_i$, it is unable to obtain $R_c\left(\tau\right)$.
We define the local discounted cumulative reward of agent $i$ over the trajectory $\tau$ as
$R_i(\tau)=\sum_{h=0}^{H-1}\gamma^{h}R_i\left(\boldsymbol{s}^h,\boldsymbol{a}^h\right)$, which is private to agent $i$.
Maximizing the performance function $J(\boldsymbol{\theta})$ in (\ref{eq2}) is equivalent to solving
\begin{equation}\label{eq7}
\max J(\boldsymbol{\theta})=\max \sum_{i=1}^{n} J_i(\boldsymbol{\theta}),
\end{equation}
where $J_i(\boldsymbol{\theta})=\mathbb{E}_{\tau \sim p\left(\cdot\left|\boldsymbol{\theta}\right.\right)}\left[R_i(\tau)\right]$.

We focus on the fully distributed manner where $n$ agents need to reach a consensus on $\boldsymbol{\theta}$ without sharing the local performance function $J_i(\boldsymbol{\theta})$.
Suppose that agents can exchange information through a communication graph $\mathcal{G}=\left(\mathcal{V}, \mathcal{E}\right)$, where $\mathcal{V}=\{1,\ldots,n\}$ is set of nodes, $\mathcal{E}$ is the set of edges, and $\left(i, j\right) \in \mathcal{E}$ if and only if nodes $i$ and $j$ can communicate with each other. The set $\mathcal{N}_{i}=\{j | \left(i, j\right) \in \mathcal{E}\}$ is the neighbors of node $i$.
The corresponding adjacency matrix is $\boldsymbol{W} = \{w_{ij}\} \in \mathbb{R}^{n \times n}$, where $w_{ij} > 0$ if $j \in \mathcal{N}_{i}$ and $w_{ij} = 0$, otherwise.

\section{Distributed policy gradient with variance reduction and gradient tracking}
 Since the stochastic gradient estimator of the performance function has a high variance, which destabilizes and decelerates the convergence. We will incorporate the variance reduction technique into the gradient estimation, and combine it with the distributed gradient tracking for policy gradient in collaborative MARL.
 In addition, since the distribution of sampled trajectories shifts with the update of policy parameter $\boldsymbol{\theta}$, and this problem introduces a bias to the gradient estimator. We will use the importance weight method to correct the distribution shift.

We now introduce the distributed policy gradient with variance reduction and gradient tracking shown in Algorithm \ref{Alg1}.
The proposed algorithm consists of $S$ epochs. The reference policy of agent $i$ at the $s$-th epoch is denoted by $\tilde{\boldsymbol{\theta}}^s_i$. After initialization, we sample $M$ trajectories $\{\tilde{\tau}_{i,j}\}_{j=1}^{M}$ from $\tilde{\boldsymbol{\theta}}^s_i$ for agent $i$ to compute the gradient estimator $\tilde{\mu}_{i}^{s}$ shown in Line 4 of Algorithm \ref{Alg1}.

At the $k$-th iteration within the $s$-th epoch, the current policy of agent $i$ is denoted by  {${\boldsymbol{\theta}}_{i,k}^{s}$}. Each agent samples $B$ trajectories $\{\tau_{i,b}\}_{b=1}^{B}$ from ${\boldsymbol{\theta}}_{i,k+1}^{s+1}$ and estimates the gradient using SVRG as follows,
\begin{align}
\boldsymbol{v}_{i,k+1}^{s+1}=\tilde{\mu}_{i}^{s}
+\frac{1}{B}&\sum_{b=1}^{B}\left(g_i\left(\tau_{i,b}\left|{\boldsymbol{\theta}}_{i,k+1}^{s+1}\right.\right)\right. \notag \\
&\left.-\omega\left(\tau_{i,b}\left|{\boldsymbol{\theta}}_{i,k+1}^{s+1}\right.,\tilde{\boldsymbol{\theta}}_i^{s}\right)
g_i\left(\tau_{i,b}\left|\tilde{\boldsymbol{\theta}}_i^{s}\right.\right)\right), \label{eq8}
\end{align}
where $\omega\left(\tau\left|{\boldsymbol{\theta}}_{i,k+1}^{s+1}\right.,\tilde{\boldsymbol{\theta}}_{i}^{s}\right)$ is the importance weight from the current policy $\boldsymbol{\theta}^{s+1}_{i,k+1}$ to the reference policy $\tilde{\boldsymbol{\theta}}^{s}_{i}$, and is defined by
\begin{equation}\label{eq9}
\omega\left(\tau\left|{\boldsymbol{\theta}}_{i,k+1}^{s+1}\right.,\tilde{\boldsymbol{\theta}}_{i}^{s}\right)
=\frac{p\left(\tau\left|\tilde{\boldsymbol{\theta}}_{i}^{s}\right.\right)}
{p\left(\tau\left|{\boldsymbol{\theta}}_{i,k+1}^{s+1}\right.\right)}. 
\end{equation}
It has been proved in \cite{papini2018stochastic, xu2020improved, xu2020sample} that
\begin{equation}\label{eq10}
\begin{aligned}
&\mathbb{E}_{\tau \sim p\left(\cdot \left| \boldsymbol{\theta}_{i,k+1}^{s+1} \right.\right)}\left[\omega\left(\cdot\left|{\boldsymbol{\theta}}_{i,k+1}^{s+1}\right.,\tilde{\boldsymbol{\theta}}^{s}\right)g_i\left(\cdot \left|\tilde{\boldsymbol{\theta}}_i^{s}\right.\right)\right]\\
=&\mathbb{E}_{\tau \sim p\left(\cdot \left| \tilde{\boldsymbol{\theta}}_{i}^{s} \right.\right)}\left[g_i\left(\cdot \left|\tilde{\boldsymbol{\theta}}_i^{s}\right.\right)\right].
\end{aligned}
\end{equation}
It can remove the inconsistency caused by the dynamic trajectory distribution and ensure that $\boldsymbol{v}_{i,k+1}^{s+1}$ is an unbiased estimator of  {$\nabla J_{i}(\boldsymbol{ \theta }_{i,k+1}^{s+1})$ }\cite{papini2018stochastic}.
The term $\tilde{\mu}_{i}^{s}-\frac{1}{B}\sum_{b=1}^{B}\left(\omega\left(\tau_{i,b}\left|{\boldsymbol{\theta}}_{i,k+1}^{s+1}\right.,\tilde{\boldsymbol{\theta}}_i^{s}\right)g_i\left(\tau_{i,b}\left|\tilde{\boldsymbol{\theta}}_i^{s}\right.\right)\right)$ can be seen as a correction to the sub-sampled gradient estimator $\frac{1}{B}\sum_{b=1}^{B}g_i\left(\tau_{i,b}\left|{\boldsymbol{\theta}}_{i,k+1}^{s+1}\right.\right)$.

In the gradient tacking part,
we incorporate an auxiliary variable $\boldsymbol{y}_{i}$ to track the average of local SVRG gradient estimator $\boldsymbol{v}_i$ across the agents.
After receiving the parameter ${\boldsymbol{\theta}}_{r,k}^{s+1}$ from its neighbors $r \in \mathcal{N}_i$,   agent $i$ updates ${\boldsymbol{\theta}}_{i,k+1}^{s+1}$ via
\begin{equation}\label{eq11}
{\boldsymbol{\theta}}_{i,k+1}^{s+1} =\sum_{r \in \mathcal{N}_i}w_{ir}{\boldsymbol{\theta}}_{r,k}^{s+1}+\alpha \boldsymbol{y}_{i,k}^{s+1}.
\end{equation}
Then each agent refines the mix of all available gradient trackers $\boldsymbol{y}_{r,k}^{s+1}$ with local gradient estimator $\boldsymbol{v}_{i,k+1}^{s+1}$ by
\begin{equation}\label{eq12}
\boldsymbol{y}_{i,k+1}^{s+1} =\sum_{r \in \mathcal{N}_i}w_{ir}\boldsymbol{y}_{r,k}^{s+1}+\boldsymbol{v}_{i,k+1}^{s+1}-\boldsymbol{v}_{i,k}^{s+1}.
\end{equation}
Thus, SVRG and gradient tracking can jointly learn the global gradient estimator at each agent asymptotically.

Finally, we select the $\boldsymbol{\theta}_{i,\text{out}}$ uniformly at random among all the $\boldsymbol{\theta}_{i,k}^{s}$ instead of setting it to the final value.

\begin{algorithm} [H]
\caption{{Distributed policy gradient with variance reduction and gradient tracking at each agent $i$ }}\label{Alg1}
\textbf{Input}: number of epochs $S$, epoch size $K$, weight matrix $\{w_{ir}\}_{r\in \mathcal{N}_i}$ of node $i$, step size $\alpha$, batch size $M$, mini-batch size $B$, policy gradient estimator $g_{i}$, and initial parameter $\tilde{\boldsymbol{\theta}}^0_i=\boldsymbol{\theta}_{i,K}^{0}=\boldsymbol{\theta}_{i}(0)$, $\boldsymbol{y}_{i,0}^{1}=\boldsymbol{v}_{i,0}^1=\nabla J_i\left(\tilde{\boldsymbol{\theta}}_i^0\right)$.
	\begin{enumerate}
		\item[1:] \textbf{for} $s=0,1,\dots to \quad S-1$, \textbf{do}
		\item[2:] \quad Initialize ${\boldsymbol{\theta}}_{i,0}^{s+1} = \tilde{\boldsymbol{\theta}}_{i}^{s}=\boldsymbol{\theta}_{i,K}^{s}$
		\item[3:] \quad Sample $M$ trajectories $\{\tilde{\tau}_{i,j}\}$ from $p\left(\cdot\left|\tilde{\boldsymbol{\theta}}_i^{s}\right.\right)$
		\item[4:] \quad Compute $ \tilde{\mu}_{i}^{s}=\frac{1}{M}\sum_{j=1}^{M}g_i\left(\tilde{\tau}_{i,j}\left|\tilde{\boldsymbol{\theta}}_i^{s}\right.\right)$
		\item[5:] \quad \textbf{for} $k=0,1,\dots to \quad K-1$, \textbf{do}
		\item[6:] \qquad Update  ${\boldsymbol{\theta}}_{i,k+1}^{s+1}$ according to (\ref{eq11})	
		\item[7:] \qquad Sample $B$ trajectories $\{\tau_{i,b}\}$ from $p\left(\cdot\left|{\boldsymbol{\theta}}_{i,k+1}^{s+1}\right.\right)$
		\item[8:] \qquad Compute the gradient estimator
		$\boldsymbol{v}_{i,k+1}^{s+1}$ via (\ref{eq8})		
		\item[9:] \qquad Update $\boldsymbol{y}_{i,k+1}^{s+1}$ according to (\ref{eq12})
		\item[10:] \quad \textbf{end for}
		\item[11:] \quad Set $\boldsymbol{y}_{i,0}^{s+2}=\boldsymbol{y}_{i,K}^{s+1}$, $\boldsymbol{v}_{i,0}^{s+2}=\boldsymbol{v}_{i,K}^{s+1}$
		\item[12:] \textbf{end for}
		\item[13:] \textbf{return} ${\boldsymbol{\theta}}_{i,\text{out}}$ uniformly chosen from $\{{\boldsymbol{\theta}}_{i,k}^{s+1}\}$ for $k=0,\ldots,K-1$; $s=0,\ldots S-1$
	\end{enumerate}
\end{algorithm}

\section{Theoretical results }
In this section, we provide the theoretical analysis of Algorithm \ref{Alg1} for non-concave performance function of policy gradient in MARL.
\subsection{Assumptions}
The convergence results are established under the following assumptions.
\begin{Asu}\label{Asu1}
Let $\pi_{\boldsymbol{\theta}}\left(\boldsymbol{a} \left| \boldsymbol{s}\right.\right)$ be the policy parameterized by $\boldsymbol{\theta}$. For all $\boldsymbol{a} \in \mathcal{A}$ and $\boldsymbol{s} \in \mathcal{S}$, there exist constants $G, F >0$ such that the gradient and Hessian matrix of $\operatorname{log}\pi_{\boldsymbol{\theta}}\left(\boldsymbol{a} \left|\boldsymbol{s}\right.\right)$ satisfy
\begin{equation*}
\|\nabla_{\boldsymbol{\theta}}\operatorname{log}\pi_{\boldsymbol{\theta}}\left(\boldsymbol{a} \left|\boldsymbol{s}\right.\right)\|\leq G,\qquad \|\nabla^2_{\boldsymbol{\theta}}\operatorname{log}\pi_{\boldsymbol{\theta}}\left(\boldsymbol{a} \left|\boldsymbol{s}\right.\right)\| \leq F.
\end{equation*}
\end{Asu}
Since the policy function is often required to be twice differentiable in practice, Assumption \ref{Asu1} holds.
This assumption is also used in the existing works of policy gradient, see
\cite{papini2018stochastic, xu2020improved, xu2020sample}. The following assumption requires the variance of gradient estimator to be bounded, and is widely made in stochastic policy gradient \cite{papini2018stochastic, xu2020improved, xu2020sample}.
\begin{Asu}\label{Asu2}
Let $g\left(\cdot \left| \boldsymbol{\theta}\right.\right)$ be the gradient estimator given by (\ref{eq6}). There is a constant $V > 0$ such that for any policy $\pi_{\boldsymbol{\theta}}$,
\begin{equation*}
\operatorname{Var}\left[g\left(\cdot \left| \boldsymbol{\theta}\right.\right)\right] \leq V, \quad \forall \boldsymbol{\theta} \in \mathbb{R}^{d}.
\end{equation*}
\end{Asu}
Importance weight is applied to correct the distribution shift in Algorithm \ref{Alg1}.
The next assumption guarantees that the variance of the importance weight  {is} bounded, which is also made in \cite{papini2018stochastic, xu2020improved, xu2020sample}.
\begin{Asu}\label{Asu3}
Let $\omega\left(\cdot \left|\boldsymbol{\theta}_1\right.,\boldsymbol{\theta}_2\right)=\frac{p\left(\cdot \left|\boldsymbol{\theta}_2\right.\right)}{p\left(\cdot \left|\boldsymbol{\theta}_1\right.\right)}$. There is a constant $W < \infty$ such that for each policy pairs
\begin{equation*}
\operatorname{Var}\left[\omega\left(\tau \left| \boldsymbol{\theta}_{1}\right., \boldsymbol{\theta}_{2}\right)\right] \leq W, \quad \forall \boldsymbol{\theta}_{1}, \boldsymbol{\theta}_{2} \in \mathbb{R}^{d}, \tau \sim p\left(\cdot \left| \boldsymbol{\theta}_{1}\right.\right).
\end{equation*}
\end{Asu}
Below is the assumption on the communication graph.
\begin{Asu}\label{Asu4}
The communication graph $\mathcal{G}$ is connected and undirected, and the weight matrix $\boldsymbol{W}=\{w_{ir}\} \in \mathbb{R}^{n \times n}$ is doubly stochastic, i.e., $\boldsymbol{W}{\boldsymbol{1}_{n}}={\boldsymbol{1}_{n}}$ and ${\boldsymbol{1}_{n}^{\rm{T}}}\boldsymbol{W}={\boldsymbol{1}_{n}^{\rm{T}}}$.
\end{Asu}
Under Assumption \ref{Asu4}, let $\sigma$ denote the spectral norm of matrix $\boldsymbol{W}-\frac{1}{n}{\boldsymbol{1}_{n}}{\boldsymbol{1}_{n}^{\rm{T}}}$. Then $0 \leq \sigma < 1$.

\subsection{Preliminaries}
Let ${{\boldsymbol{\theta }}_{k}^{s+1}} = {\left[ {{\boldsymbol{\theta }_{1,k}^{s+1}}^{\rm{T}} ,{\boldsymbol{\theta }_{2,k}^{s+1}}^{\rm{T}}, \ldots, {\boldsymbol{\theta }_{n,k}^{s+1}}^{\rm{T}}} \right]^{\rm{T}}}\in \mathbb{R}^{nd}$, ${{\boldsymbol{y}}_{k}^{s+1}} = {\left[ {{\boldsymbol{y}_{1,k}^{s+1}}^{\rm{T}}, {\boldsymbol{y}_{2,k}^{s+1}}^{\rm{T}}, \ldots, {\boldsymbol{y}_{n,k}^{s+1}}^{\rm{T}}} \right]^{\rm{T}}}\in \mathbb{R}^{nd}$, and ${{\boldsymbol{v}}_{k}^{s+1}} = {\left[ {{\boldsymbol{v}_{1,k}^{s+1}}^{\rm{T}},{\boldsymbol{v}_{2,k}^{s+1}}^{\rm{T}}, \ldots, {\boldsymbol{v}_{n,k}^{s+1}}^{\rm{T}}} \right]^{\rm{T}}}\in \mathbb{R}^{nd}$. Then the distributed policy gradient with variance reduction and gradient tracking in (\ref{eq11}) and (\ref{eq12}) can be written in a compact form as follows,
\begin{subequations}\label{eq13}
\begin{align}
{{\boldsymbol{\theta }}_{k+1}^{s+1}} &= \left({\boldsymbol{W}}\otimes{\boldsymbol{I}_d}\right){\boldsymbol{\theta }}_{k}^{s+1} + \alpha {\boldsymbol{y}_{k}^{s+1}}, \label{eq13a} \\
{{\boldsymbol{y}}_{k+1}^{s+1}} &= \left({\boldsymbol{W}}\otimes{\boldsymbol{I}_d}\right){{\boldsymbol{y}}_{k}^{s+1}} + {{\boldsymbol{v}}_{k+1}^{s+1}} - {{\boldsymbol{v}}_{k}^{s+1}}. \label{eq13b}
\end{align}
\end{subequations}
In the following, we define several auxiliary variables that will help the subsequent convergence analysis.
\begin{subequations}\label{eq14}
\begin{align}
{{\bar{\boldsymbol{ \theta }}}_{k}^{s+1}} &= \frac{1}{n}\sum\limits_{i = 1}^n {{\boldsymbol{\theta }}_{i,k}^{s+1}}  = \frac{1}{n}\left({\boldsymbol{1}_{n}^{\rm{T}}}\otimes{\boldsymbol{I}_d}\right){{\boldsymbol{\theta }}_{k}^{s+1}}, \label{eq14a} \\
{\bar {\boldsymbol{y}}_{k}^{s+1}} &= \frac{1}{n}\sum\limits_{i = 1}^n {{\boldsymbol{y}}_{i,k}^{s+1}}  = \frac{1}{n}\left({\boldsymbol{1}_{n}^{\rm{T}}}\otimes{\boldsymbol{I}_d}\right){{\boldsymbol{y}}_{k}^{s+1}}, \label{eq14b}\\
{\bar {\boldsymbol{v}}_{k}^{s+1}} &= \frac{1}{n}\sum\limits_{i = 1}^n {{\boldsymbol{v}}_{i,k}^{s+1}}  = \frac{1}{n}\left({\boldsymbol{1}_{n}^{\rm{T}}}\otimes{\boldsymbol{I}_d}\right){{\boldsymbol{v}}_{k}^{s+1}}, \label{eq14c} \\
\nabla J\left( {{{\boldsymbol{\theta }}_{k}^{s+1}}} \right) &= {\left[ {\nabla {J_1}{{\left( {{\boldsymbol{\theta }}_{1,k}^{s+1}} \right)}^{\rm{T}}},\ldots, \nabla {J_n}{{\left( {{\boldsymbol{\theta }}_{n,k}^{s+1}} \right)}^{\rm{T}}}} \right]^{\rm{T}}} \in \mathbb{R}^{nd}, \label{eq14d}\\
\overline{\nabla J}\left(\boldsymbol{\theta}_{k}^{s+1} \right)&= \frac{1}{n}\sum\limits_{i = 1}^n {\nabla {J_i}\left( {{\boldsymbol{\theta }}_{i,k}^{s+1}} \right)}  = \frac{1}{n}\left({\boldsymbol{1}_{n}^{\rm{T}}}\otimes{\boldsymbol{I}_d}\right)\nabla J\left( \boldsymbol{\theta }_{k}^{s+1}\right). \label{eq14e}
\end{align}
\end{subequations}
In the rest of the paper, we set $d=1$ for the sake of simplicity.

Next, we present several preliminary results related to the policy gradient methods whose proof can be found in \cite{papini2018stochastic, xu2020improved, xu2020sample}.
The following lemma shows that with Assumption \ref{Asu1}, the Hessian matrix $\nabla^2J(\boldsymbol{\theta})$  of the performance function is bounded. This implies that the performance function $J(\boldsymbol{\theta})$ is $L$-smooth, which is important for analyzing the convergence of non-convex optimization (see e.g., \cite{allen2016variance, reddi2016stochastic, li2018simple, ge2019stabilized}).
\begin{Lem}\label{Lem1}
Suppose Assumption 1 hold. Let $g\left(\tau \left| \boldsymbol{\theta}\right.\right)$ be the G(PD)MDP gradient estimator given by (\ref{eq6}). For any trajectory $\tau \in \mathcal{T}$, we have\\
(1) $J(\boldsymbol{\theta})$ is $L$-smooth, that is $\|\nabla_{\boldsymbol{\theta}}^2 J(\boldsymbol{\theta})\| \leq L$;\\
(2) $g\left(\tau \left| \boldsymbol{\theta}\right.\right)$ is $L_{g}$-Lipschitz continuous, that is $\|g\left(\tau \left|\boldsymbol{\theta}_1\right.\right)-g\left(\tau \left|\boldsymbol{\theta}_2\right.\right)\| \leq L_{g}\|\boldsymbol{\theta}_{1}-\boldsymbol{\theta}_{2}\|, \forall \boldsymbol{\theta}_{1}, \boldsymbol{\theta}_{2} \in \mathbb{R}^{d}$;\\
(3) There is a positive constant $C_g$ such that $\|g\left(\tau \left| \boldsymbol{\theta}\right.\right)\| \leq C_g, \forall \boldsymbol{\theta} \in \mathbb{R}^{d}$. 	
\end{Lem}


The subsequent lemma from \cite{xu2020improved, xu2020sample}  shows that the variance of the importance weight ${\omega \left( {\tau \left| {{\boldsymbol{\theta}_{k+1}^{s+1}}} \right.,{{\tilde{ \boldsymbol{\theta}} }^s}} \right)}$  is proportional to the distance between the current policy ${\boldsymbol{\theta}_{k+1}^{s+1}}$ and the reference policy ${{\tilde{\boldsymbol{\theta}}}^s}$.
\begin{Lem}\label{Lem2}
Under Assumptions 1 and 3, the importance sampling weight ${\omega \left( {\tau \left| {{\boldsymbol{\theta}_{k+1}^{s+1}}} \right.,{{\tilde{ \boldsymbol{\theta}} }^s}} \right)}$ in (\ref{eq9}) satisfies
\begin{equation}\label{eq15}
{\rm{Var}}\left[ {\omega \left( {\tau \left| {{\boldsymbol{\theta}_{k+1}^{s+1}}} \right.,{{\tilde{ \boldsymbol{\theta}}}^s}} \right)} \right] \leq {C_\omega }{\left\| {{\boldsymbol{\theta}_{k+1}^{s+1}} - {{\tilde{\boldsymbol{\theta}}}^s}} \right\|^2},
\end{equation}
where ${C_\omega}=H\left(2HG^2+F\right)\left(W+1\right)$.
\end{Lem}





\subsection{Convergence Results}
The convergence of the proposed distributed policy gradient with variance reduction is shown in the following theorem, for which the detailed proof is given in the next subsection.
\begin{Thm}\label{Thm1}
Suppose Assumptions 1-4 hold. Consider   Algorithm \ref{Alg1} with the step size $\alpha$ and mini-batch size $B$ satisfying
\begin{align}
0 <  \alpha <& \text{min}\left\{  \frac{\left(1-\sigma^2\right)^2}{24 \sqrt[3]{4 \Psi K^2 L \left(\left(1-\sigma^2\right)^2+24\left(1-\sigma^2\right)\right)}},\right. \notag \\
&\left. \qquad \frac{\left(1-\sigma^2\right)^2}{96\sqrt{3 \Psi K}},\frac{1}{2L} \right\}, \notag \\
B \ge& \text{max} \left\{ B_1(\alpha), B_2(\alpha)\right\}, \label{eq_alpha_B}
\end{align}
where
\begin{align}
&\Psi =2(C_{g}^{2}C_{\omega}+L_{g}^2), \label{eq_psi} \\
&B_1(\alpha) = \frac{54 \alpha \Psi K^2 }{n L }, \label{eq_B1}\\
&B_2(\alpha) = \frac{36 \alpha^3 \Psi^2 K^2 \left(\left(1-\sigma^2\right)^2+24\left(1-\sigma^2\right)\right)}{n L\frac{\left(1-\sigma^2\right)^6}{4608}- \tilde{\Psi}(\alpha)}, \label{eq_B2}
\end{align}
with $\tilde{\Psi}(\alpha)=12\alpha^3 \Psi K^2 n L^2 \left(\left(1-\sigma^2\right)^2+24\left(1-\sigma^2\right)\right)$.

Then the output $\bar{\boldsymbol{\theta}}_{\text{out}}=\frac{1}{n}\sum_{i=1}^{n}{\boldsymbol{\theta}}_{\text{out}}$ satisfies
\begin{align}
\mathbb{E} \left[\left\|\nabla J
\left(\bar{\boldsymbol{\theta}}_{\text{out}}\right)\right\|^{2}\right]
\le& \frac{2\left(J\left ( {\boldsymbol{\theta}}^{*} \right ) -  J\left ( \bar{\boldsymbol{\theta}}_{0}^{1} \right )\right)}{\alpha KS} +\frac{2 V}{Mn} \notag \\
&+ \frac{2 v_1}{ nKS} \mathbb{E} \left[\left\|\boldsymbol{\theta}_{0}^{1}
-{\boldsymbol{1}_{n}}\bar{\boldsymbol{\theta}}_{0}^{1}\right\|^2\right]  \notag \\
& + \frac{2 \alpha^2 v_2}{ n KS }  \mathbb{E} \left[\left\|\boldsymbol{y}_{0}^{1}
-{\boldsymbol{1}_{n}}\bar{\boldsymbol{y}}_{0}^{1}\right\|^2\right],
\end{align}
where $\boldsymbol{\theta}^{*}$ is the maximizer of $J\left(\boldsymbol{\theta}\right)$,
\begin{align}
v_1&=\frac{\frac{1536 \alpha^2 \Psi K \left(L^2+\frac{3 \Psi}
{Bn}\right)}{\left(1-\sigma^2\right)^2}\left(\frac{3}{K}+\frac{17-\sigma^2}{4}-\frac{9216 \alpha^2 \Psi}{\left(1-\sigma^2\right)^4}\right)}{\frac{\left(1-\sigma^2\right)^4}{64}-48\alpha^2\Psi K} \notag \\
&\qquad + \frac{\left(L^2+\frac{3 \Psi}{Bn}\right)\left(1-\sigma^2\right)^2+\frac{3 \alpha K}{B n}}{\frac{\left(1-\sigma^2\right)^4}{64}-48\alpha^2\Psi K}, \notag \\
v_2 &= \frac{\frac{3 \Psi K \left(1-\sigma^2\right)}{B n}+16\left(L^2+\frac{3 \Psi}
{Bn}\right)\left(3+\frac{384\alpha^2\Psi K}{\left(1-\sigma^2\right)^3}\right)}{\frac{\left(1-\sigma^2\right)^4}{64}-48\alpha^2\Psi K}. \notag
\end{align}
\end{Thm}

\begin{Rem}\label{Rem1}
Theorem \ref{Thm1} shows that the proposed distributed policy gradient algorithm can converge to an approximate first-order stationary point. Note that $S$ is the number of epochs and $K$ is the epoch length, so $KS$ is the total number of iterations of Algorithm \ref{Alg1}. The first term $O\left(\frac{1}{KS}\right)$ characterizes the linear convergence rate, which is coherent with results on SVRG for non-convex optimization problems \cite{allen2016variance, reddi2016stochastic}.
The second term $O\left(\frac{1}{M}\right)$ comes from the stochastic gradient estimator using $M$ trajectories in the outer loop, which matches the results in \cite{papini2018stochastic, xu2020sample}. It is possible to select $M$ large enough to make the second term insignificant. That is, the variance introduced by G(PO)MDP gradient estimator can be neglected. The last two terms are related to the initial values and are analogous to the results of DSGT in \cite{zhang2019decentralized}.
\end{Rem}

In the following remark, we show the effectiveness of the parameters appearing in Theorem \ref{Thm1}.
\begin{Rem}\label{Rem2}
(1) The denominator of (\ref{eq_B2}) is positive, i.e., $n L\frac{\left(1-\sigma^2\right)^6}{4608}- \tilde{\Psi}(\alpha) > 0$ since the step size satisfy $0 < \alpha < \frac{\left(1-\sigma^2\right)^2}{24 \sqrt[3]{4 \Psi K^2 L \left(\left(1-\sigma^2\right)^2+24\left(1-\sigma^2\right)\right)}}$.
(2) Since $ 0 < \alpha <\frac{\left(1-\sigma^2\right)^2}{96\sqrt{3 \Psi K}}$, we conclude that $v_1, v_2 > 0$.
(3) Set the mini-batch size as $\tilde{B} = \frac{54 \alpha \Psi K^2 }{n L - \frac{55296 \alpha^3 \Psi K^2 L^2 n\left(\left(1-\sigma^2\right)^2+24\left(1-\sigma^2\right)\right)}{\left(1-\sigma^2\right)^6 }}$. It is easily seen that $\tilde{B} > B_1(\alpha)$. If $0 < \alpha \le \frac{\left(1-\sigma^2\right)^3}{32 \sqrt{ \Psi   \left(\left(1-\sigma^2\right)^2+24\left(1-\sigma^2\right)\right)}}$, we have $\tilde{B} \ge B_2(\alpha) $. Therefore, we can obtain the refined mini-batch size $B \ge \tilde{B}$.
\end{Rem}

Based on the convergence established in Theorem \ref{Thm1}, we are able to establish the complexity bounds to find an $\epsilon$-approximate stationary point satisfying $\mathbb{E}[\left\|\nabla J \left(\bar{\boldsymbol{\theta}}_{\text{out}}\right)\right\|^2] \le \epsilon$.

\begin{Col}\label{Col1}
Suppose Assumptions 1-4 hold. Let $\epsilon > 0$.  Set the step size  $\alpha$ as (\ref{eq_alpha_B}), the batch size as $M=O\left(\frac{1}{\epsilon}\right)$,  the mini-batch size as $B=O\left(\frac{1}{\epsilon^{\frac{1}{2}}}\right)$,  {the  epoch size as $S=O\left(\frac{1}{\epsilon^{\frac{1}{2}}}\right)$,} and  the  epoch length as $K=O\left(\frac{1}{\epsilon^{\frac{1}{2}}}\right)$. Then Algorithm \ref{Alg1} requires $O\left(\frac{1}{\epsilon^{\frac{3}{2}}}\right)$ trajectories to achieve $\mathbb{E}\left[\left\|\nabla J \left(\bar{\boldsymbol{\theta}}_{\text{out}}\right)\right\|^2\right] \le \epsilon$.
\end{Col}

\begin{proof}
Based on the convergence results in Theorem \ref{Thm1}, we conclude that $\mathbb{E} \left[\left\|\nabla J\left(\bar{\boldsymbol{\theta}}_{\text{out}}\right)\right\|^{2}\right]  \le \epsilon$ when  $M=O\left(\frac{1}{\epsilon}\right)$,  $S=O\left(\frac{1}{\epsilon^{\frac{1}{2}}}\right)$, and  $K=O\left(\tfrac{1}{\epsilon^{\frac{1}{2}}}\right)$.
Hence $KS= O\left(\frac{1}{\epsilon}\right)$, and  by $B=O\left(\tfrac{1}{\epsilon^{\frac{1}{2}}}\right)$ it follows that
the total number of stochastic gradient evaluations   required  is bounded by
\begin{equation*}
	SM + SKB =   O\left(\frac{1}{\epsilon ^{\frac{3}{2}}}\right).
\end{equation*}
\end{proof}

\begin{Rem}\label{Rem3}
We use the notation $\left|\mathcal{N}_{i}\right|$ to represent the number of neighbors which can directly send message to agent $i$, and use the notation $\left|\mathcal{E}\right|$ to represent the number of edges for communication graph $\mathcal{G}$. For Algorithm \ref{Alg1}, agent $i$ needs $KS\left|\mathcal{N}_{i}\right|$ communication rounds to receive its neighbors' message ${\boldsymbol{\theta}}_{r,k}^{s+1}$ and ${\boldsymbol{y}}_{r,k}^{s+1}$. Therefore, the total number of communication rounds for agent $i$ and all agents over the graph are $O\left(\left|\mathcal{N}_{i}\right|\frac{1}{\epsilon}\right)$ and $O\left(\left|\mathcal{E}\right|\frac{1}{\epsilon}\right)$, respectively.
\end{Rem}

\subsection{Convergence Analysis}
The following lemma shows the upper bound of consensus error in the inner loop.
\begin{Lem}\label{Lem9}
Let Assumptions \ref{Asu1}, \ref{Asu3}, and \ref{Asu4} hold.
Consider Algorithm \ref{Alg1} with $0 < \alpha < \frac{\left(1-\sigma^2\right)^2}{24\sqrt{2\Psi}}$. Define $\lambda=\frac{3+\sigma^2}{4}+\frac{6 \alpha\sqrt{2 \Psi}}{1-\sigma^2}$. Then we have the following inequality for any $k \ge 0,s \ge 0$
\begin{align}
&\sum_{k=0}^{K-1} \mathbb{E}\left[\left\|\boldsymbol{\theta}_{k}^{s+1}-\boldsymbol{1}_{n}\bar{\boldsymbol{\theta}}_{k}^{s+1}\right\|^{2}\right] \notag\\
\le& c_0(s) + \frac{1}{\left(1-\lambda\right)^2}\sum_{k=0}^{K-1}\alpha^2 \left(c_1\mathbb{E}\left[  \left\| \bar{\boldsymbol{ \theta }}_{k+1}^{s+1} - \bar{\boldsymbol{ \theta }}_{0}^{s+1}\right\|^2\right]\right. \notag \\
+ &\left. c_1\mathbb{E}\left[  \left\| \bar{\boldsymbol{ \theta }}_{k}^{s+1} - \bar{\boldsymbol{ \theta }}_{0}^{s+1}\right\|^2\right] + c_2 \mathbb{E}\left[\left\|\boldsymbol{\theta}_{0}^{s+1}-{\boldsymbol{1}_{n}}\bar{\boldsymbol{\theta}}_{0}^{s+1}\right\|^2\right]\right), \label{eq42}
\end{align}
where
\begin{subequations}
\begin{align}		c_0(s)=&\frac{1}{1-\lambda}\mathbb{E}\left[\left\|\boldsymbol{\theta}_{0}^{s+1}-{\boldsymbol{1}_{n}}\bar{\boldsymbol{\theta}}_{0}^{s+1}\right\|^2\right] \notag\\		&+\frac{2\alpha^2}{\left(1-\sigma^2\right)\left(1-\lambda\right)^2}\mathbb{E}\left[\left\|\boldsymbol{y}_{0}^{s+1}-{\boldsymbol{1}_{n}}\bar{\boldsymbol{y}}_{0}^{s+1}\right\|^2\right], \label{addeq1}\\
c_1=&\frac{24\Psi n}{\left(1-\sigma^2\right)^2}, \text{ and } c_2=\frac{48\Psi }{\left(1-\sigma^2\right)^2}. \label{addeq2}
\end{align}
\end{subequations}
\end{Lem}
The proof of Lemma \ref{Lem9} can be found in Appendix B-A.

We provide the accumulated consensus error and gradient tracking error in the outer loop in the subsequent lemma.
\begin{Lem}\label{Lem10}
Suppose the same conditions in Lemma \ref{Lem9} hold. 
If $0 < \alpha < \frac{\left(1-\sigma^2\right)^2}{32\sqrt{3 \Psi K}}$, we have
\begin{align}
& \sum_{s=0}^{S-1}\mathbb{E}\left[\left\| {{\boldsymbol{\theta }}_{0}^{s+1} - \boldsymbol{1}_{n}\bar{\boldsymbol{ \theta }}_{0}^{s+1}} \right\|^{2}\right] \notag \\
\le &
\frac{\left(1-\lambda\right)\left(1-\sigma^2\right)^2}{\Xi}\mathbb{E} \left[\left\|\boldsymbol{\theta}_{0}^{1}-{\boldsymbol{1}_{n}}\bar{\boldsymbol{\theta}}_{0}^{1}\right\|^2\right] \notag \\
&+ \frac{2 \alpha^2 \left(1-\sigma^2\right)}{\Xi}\mathbb{E} \left[\left\|\boldsymbol{y}_{0}^{1}-{\boldsymbol{1}_{n}}\bar{\boldsymbol{y}}_{0}^{1}\right\|^2\right] \notag \\
&+ \frac{24\alpha^2\Psi n}{\Xi}\sum_{s=0}^{S-1}\sum_{k=0}^{K-1}\mathbb{E}\left[  \left\| \bar{\boldsymbol{ \theta }}_{k+1}^{s+1} - \bar{\boldsymbol{ \theta }}_{0}^{s+1}\right\|^2\right] \notag \\
&+ \frac{24\alpha^2\Psi n}{\Xi}\sum_{s=0}^{S-1}\sum_{k=0}^{K-1}\mathbb{E}\left[  \left\| \bar{\boldsymbol{ \theta }}_{k}^{s+1} - \bar{\boldsymbol{ \theta }}_{0}^{s+1}\right\|^2\right], \label{eq48} \\
& \sum_{s=0}^{S-1}\mathbb{E}\left[\left\| {{\boldsymbol{y }}_{0}^{s+1}
- \boldsymbol{1}_{n}\bar{\boldsymbol{y}}_{0}^{s+1}} \right\|^{2}\right] \notag \\
\le & \left(\frac{36\Psi}{\left(1-\lambda\right)^2\left(1-\sigma^2\right)}
+\frac{3 \Psi K \left(1-\sigma^2\right)\Phi}{\left(1-\lambda\right)\Xi}\right)\mathbb{E} \left[\left\|\boldsymbol{\theta}_{0}^{1}-{\boldsymbol{1}_{n}}\bar{\boldsymbol{\theta}}_{0}^{1}\right\|^2\right] \notag \\
&+ \frac{\left(1-\sigma^2\right)^2\Phi}{8\Xi}\mathbb{E}
\left[\left\|\boldsymbol{y}_{0}^{1}-{\boldsymbol{1}_{n}}\bar{\boldsymbol{y}}_{0}^{1}\right\|^2\right] \notag \\
& + \frac{3\Psi n \left(1-\sigma^2\right)\Phi }
{ 2\Xi}\sum_{s=0}^{S-1}\sum_{k=0}^{K-1}\mathbb{E}\left[  \left\| \bar{\boldsymbol{ \theta }}_{k+1}^{s+1} - \bar{\boldsymbol{ \theta }}_{0}^{s+1}\right\|^2\right] \notag \\
& + \frac{3\Psi n \left(1-\sigma^2\right)\Phi }
{ 2\Xi}\sum_{s=0}^{S-1}\sum_{k=0}^{K-1}\mathbb{E}\left[  \left\| \bar{\boldsymbol{ \theta }}_{k}^{s+1} - \bar{\boldsymbol{ \theta }}_{0}^{s+1}\right\|^2\right], \label{eq49}
\end{align}
where
\begin{subequations}
\begin{align}
\Xi &= \left(1-\lambda\right)^2\left(1-\sigma^2\right)^2-48\alpha^2\Psi K , \label{eq50a}\\
\Phi &= 8 \left(1-\lambda\right)+\left(1-\sigma^2\right).  \label{eq50b}
\end{align}
\end{subequations}
\end{Lem}
The proof of Lemma \ref{Lem10} can be found in Appendix B-B.

Next, we provide the bound of the gradient estimator error.
\begin{Lem}\label{Lem11}
Let Assumptions \ref{Asu1} to \ref{Asu4} hold. Considering Algorithm 1, we have the following inequality,
\begin{align}
&\mathbb{E}\left[\left\|\bar{\boldsymbol{v}}_{k}^{s+1}-\overline{\nabla J}\left(\boldsymbol{\theta}_{k}^{s+1}\right)\right\|^{2}\right] \notag \\
\leq& \frac{3 \Psi }{ B n^2 } \mathbb{E}\left[\left\| \boldsymbol{\theta}_{k}^{s+1} - {\boldsymbol{1}_{n}}\bar{\boldsymbol{\theta}}_{k}^{s+1} \right\|^2 \right]  + \frac{3 \Psi }{ B n } \mathbb{E}\left[\left\| \bar{\boldsymbol{\theta}}_{k}^{s+1}-\bar{\boldsymbol{\theta}}_{0}^{s+1}\right\|^2\right] \notag\\
&+\frac{3 \Psi }{ B n^2} \mathbb{E}\left[\left\| \boldsymbol{\theta}_{0}^{s+1} - {\boldsymbol{1}_{n}}\bar{\boldsymbol{\theta}}_{0}^{s+1} \right\|^2\right]+ \frac{V}{Mn}. \label{eq61}
\end{align}		
\end{Lem}
The proof of Lemma \ref{Lem11} can be found in Appendix B-C.

With the above auxiliary results, we are ready to prove Theorem \ref{Thm1}.
\begin{proof}
	Note that for $L$-smooth $f\left (  {\bf{x}}\right )$, the following quadratic upper bound holds
	\begin{equation*}
	\begin{aligned}
	-\frac{L}{2} \left \| {\bf{x}}-{\bf{y}} \right \|^2 \le  f\left (  {\bf{x}}\right ) - f\left ( {\bf{y}} \right )
    -\left \langle \nabla f\left ( {\bf{y}} \right ),{\bf{x}}-{\bf{y}}  \right \rangle.
	\end{aligned}
	\end{equation*}	
	From Lemma \ref{Lem1}, $J(\boldsymbol{\theta})$ is $L$-smooth. Setting ${\bf{x}}=\bar{\boldsymbol{\theta}}_{k+1}^{s+1}$ and ${\bf{y}}=\bar{\boldsymbol{\theta}}_{k}^{s+1}$ leads to
	\begin{equation*}
	\begin{aligned}
	J\left ( \bar{\boldsymbol{\theta}}_{k+1}^{s+1} \right ) \ge& J\left ( \bar{\boldsymbol{\theta}}_{k}^{s+1} \right ) + \left \langle \nabla J\left ( \bar{\boldsymbol{\theta}}_{k}^{s+1} \right ),\bar{\boldsymbol{\theta}}_{k+1}^{s+1}-\bar{\boldsymbol{\theta}}_{k}^{s+1}  \right \rangle \\
	&-\frac{L}{2} \left \| \bar{\boldsymbol{\theta}}_{k+1}^{s+1}-\bar{\boldsymbol{\theta}}_{k}^{s+1} \right \|^2.
	\end{aligned}
	\end{equation*}
	By using ${\bar {\boldsymbol{v}}_{k}^{s+1}}={\bar {\boldsymbol{y}}_{k}^{s+1}}$ from \cite[Lemma 7]{qu2017harnessing}, we obtain that $\bar{\boldsymbol{\theta}}_{k+1}^{s+1}-\bar{\boldsymbol{\theta}}_{k}^{s+1}=\alpha \bar{\boldsymbol{y}}_{k}^{s+1}=\alpha \bar{\boldsymbol{v}}_{k}^{s+1}$. This combined with the above inequality produces that
	\begin{align}
	J\left(\bar{\boldsymbol{\theta}}_{k+1}^{s+1}\right) & \ge J\left(\bar{\boldsymbol{\theta}}_{k}^{s+1}\right)
    + \left \langle \nabla J\left(\bar{\boldsymbol{\theta}}_{k}^{s+1}\right),\alpha \bar{\boldsymbol{v}}_{k}^{s+1} \right \rangle \notag \\
	&\quad -\frac{L}{2}{\left\|\bar{\boldsymbol{\theta}}_{k+1}^{s+1}-\bar{\boldsymbol{\theta}}_{k}^{s+1}\right\|^2} \notag \\
	\stackrel{(a)}{=}&  J\left ( \bar{\boldsymbol{\theta}}_{k}^{s+1} \right )
    + \frac{\alpha}{2} \left\|\nabla J\left( \bar{\boldsymbol{\theta}}_{k}^{s+1}\right)\right\|^2 +\frac{\alpha}{2}\left \|  \bar{\boldsymbol{v}}_{k}^{s+1}\right \|^2 \notag \\
	& - \frac{\alpha}{2} \left\|\nabla J\left( \bar{\boldsymbol{\theta}}_{k}^{s+1} \right )-\bar{\boldsymbol{v}}_{k}^{s+1}\right\|^2 -\frac{L}{2} \left \| \bar{\boldsymbol{\theta}}_{k+1}^{s+1}-\bar{\boldsymbol{\theta}}_{k}^{s+1} \right \|^2 \notag \\
	= & J\left ( \bar{\boldsymbol{\theta}}_{k}^{s+1} \right ) -\frac{\alpha}{2} \left \|  \nabla J\left ( \bar{\boldsymbol{\theta}}_{k}^{s+1} \right )-\bar{\boldsymbol{v}}_{k}^{s+1}\right \|^2
	\notag \\
	& +\frac{\alpha}{2} \left \|  \nabla J\left ( \bar{\boldsymbol{\theta}}_{k}^{s+1} \right) \right \|^2
    +\left(\frac{1}{2\alpha}-\frac{L}{2}\right) \left \| \bar{\boldsymbol{\theta}}_{k+1}^{s+1}-\bar{\boldsymbol{\theta}}_{k}^{s+1} \right \|^2
	\notag \\
	\ge & J\left(\bar{\boldsymbol{\theta}}_{k}^{s+1}\right)-\alpha\left\|\nabla J\left(\bar{\boldsymbol{\theta}}_{k}^{s+1}\right)-\overline{\nabla J}\left(\boldsymbol{\theta}_{k}^{s+1}\right)\right\|^{2} \notag \\
	&- \alpha \left\|\overline{\nabla J}\left(\boldsymbol{\theta}_{k}^{s+1}\right)-\bar{\boldsymbol{v}}_{k}^{s+1}\right\|^{2} +\frac{\alpha}{2}\left\|\nabla J\left(\bar{\boldsymbol{\theta}}_{k}^{s+1}\right)\right\|^{2} \notag \\
	&+ \left(\frac{1}{2 \alpha}-\frac{L}{2}\right)\left\|\bar{\boldsymbol{\theta}}_{k+1}^{s+1}-\bar{\boldsymbol{\theta}}_{k}^{s+1}\right\|^{2} \notag \\
	\stackrel{(b)}{\geq} & J\left(\bar{\boldsymbol{\theta}}_{k}^{s+1}\right)-\frac{ \alpha L^2}{ n }\left\|\boldsymbol{\theta}_{k}^{s+1}-\boldsymbol{1}_{n}\bar{\boldsymbol{\theta}}_{k}^{s+1}\right\|^{2} \notag \\
	&-\alpha \left\|\overline{\nabla J}\left(\boldsymbol{\theta}_{k}^{s+1}\right)-\bar{\boldsymbol{v}}_{k}^{s+1}\right\|^{2} +\frac{\alpha}{2}\left\|\nabla J\left(\bar{\boldsymbol{\theta}}_{k}^{s+1}\right)\right\|^{2} \notag \\
	&+\left(\frac{1}{2 \alpha}-\frac{L}{2}\right)\left\|\bar{\boldsymbol{\theta}}_{k+1}^{s+1}-\bar{\boldsymbol{\theta}}_{k}^{s+1}\right\|^{2}. \label{eq65}
	\end{align}
	where in (a) we use $\left\langle {\bf{a}}, {\bf{b}} \right\rangle = 0.5 \left(\left\|{\bf{a}}\right\|^2 + \left\|{\bf{b}}\right\|^2 - \left\| {\bf{a}}-{\bf{b}}\right\|^2\right)$, and in (b) we use $\left\| \overline{\nabla J}\left(\boldsymbol{\theta}_{k}^{s+1} \right) - \nabla J\left( {{\bar{\boldsymbol{ \theta }}}_{k}^{s+1}} \right) \right\| \leq \frac{L}{\sqrt{n}}\left\| \boldsymbol{\theta}_{k}^{s+1} - {\boldsymbol{1}_{n}} {{\bar{\boldsymbol{ \theta }}}_{k}^{s+1}}\right\|$ related to the distributed gradient tracking methods (see, e.g., \cite[Lemma 8]{qu2017harnessing}).
	
	We apply the basic Young's inequality that $\left\|\boldsymbol{a}+\boldsymbol{b}\right\|^2
    \leq \left(1+\eta\right)\|\boldsymbol{a}\|^2
    + \left(1+\frac{1}{\eta}\right)\|\boldsymbol{b}\|^2$, $\forall \boldsymbol{a}, \boldsymbol{b} \in \mathbb{R}^{ d}, \forall \eta>0$ to obtain
	\begin{equation*}
	\begin{aligned}
	\left\| \bar{\boldsymbol{\theta}}_{k+1}^{s+1}-\bar{\boldsymbol{\theta}}_{0}^{s+1} \right\|^2
	\leq& \left(1+\eta \right)\left\| \bar{\boldsymbol{\theta}}_{k+1}^{s+1}-\bar{\boldsymbol{\theta}}_{k}^{s+1} \right\|^2\\
	& + \left(1+\frac{1}{\eta} \right)\left\| \bar{\boldsymbol{\theta}}_{k}^{s+1}-\bar{\boldsymbol{\theta}}_{0}^{s+1} \right\|^2.
	\end{aligned}
	\end{equation*}
	Regrouping terms of the above inequality and taking expectations, we have
	\begin{align}
	\mathbb{E} \left[ \left\| \bar{\boldsymbol{\theta}}_{k+1}^{s+1}-\bar{\boldsymbol{\theta}}_{k}^{s+1} \right\|^2 \right]
	\geq&  \frac{1}{1+\eta} \mathbb{E} \left[ \left\| \bar{\boldsymbol{\theta}}_{k+1}^{s+1}-\bar{\boldsymbol{\theta}}_{0}^{s+1} \right\|^2  \right] \notag \\
	&- \frac{1}{\eta} \mathbb{E} \left[ \left\| \bar{\boldsymbol{\theta}}_{k}^{s+1}-\bar{\boldsymbol{\theta}}_{0}^{s+1} \right\|^2\right]. \label{eq67}
	\end{align}
	By taking the expectations on both sides of (\ref{eq65}), and using (\ref{eq61}) and (\ref{eq67}), we  obtain that	
	\begin{align}
	\mathbb{E} \left[J\left ( \bar{\boldsymbol{\theta}}_{k+1}^{s+1} \right )\right] &\ge\mathbb{E} \left[J\left ( \bar{\boldsymbol{\theta}}_{k}^{s+1} \right )\right] -\frac{\alpha V}{Mn}+\frac{\alpha}{2}\mathbb{E} \left[\left\|\nabla J\left(\bar{\boldsymbol{\theta}}_{k}^{s+1}\right)\right\|^{2}\right]  \notag\\
	&-\left(\frac{ \alpha L^2}{ n }  + \frac{3 \alpha \Psi }{B n^2} \right)\mathbb{E}\left[\left\|\boldsymbol{\theta}_{k}^{s+1}-\boldsymbol{1}_{n}\bar{\boldsymbol{\theta}}_{k}^{s+1}\right\|^{2} \right] \notag\\
	&+\frac{\frac{1}{\alpha}-L}{2\left(\eta+1\right)}  \mathbb{E} \left[ \left\| \bar{\boldsymbol{\theta}}_{k+1}^{s+1}-\bar{\boldsymbol{\theta}}_{0}^{s+1} \right\|^2\right] \notag\\
	&- \left(\frac{\frac{1}{\alpha}-L}{2 \eta} + \frac{3 \alpha \Psi }{B n} \right){\mathbb{E}\left[\left\|\boldsymbol{\theta}_{k}^{s+1}-\bar{\boldsymbol{\theta}}_{0}^{s+1}\right\|^{2} \right]} \notag \\
	& -\frac{3 \alpha \Psi}{B n^2} {\mathbb{E}\left[\left\|\boldsymbol{\theta}_{0}^{s+1}-\boldsymbol{1}_{n}\bar{\boldsymbol{\theta}}_{0}^{s+1}\right\|^{2} \right]}. \label{eq68}
	\end{align}
	By setting $\eta=2k+1$ and taking the telescoping sum of the above inequality over $k$ from 0 to $K-1$, we have
	\begin{align}
	\mathbb{E} & \left[J\left ( \bar{\boldsymbol{\theta}}_{K}^{s+1} \right )\right] \ge\mathbb{E} \left[J\left ( \bar{\boldsymbol{\theta}}_{0}^{s+1} \right )\right]+\frac{\alpha}{2}\sum_{k=0}^{K-1}\mathbb{E} \left[\left\|\nabla J\left(\bar{\boldsymbol{\theta}}_{k}^{s+1}\right)\right\|^{2}\right] \notag \\
	&-\frac{\alpha V K}{Mn} +\sum_{k=0}^{K-1} \frac{\frac{1}{\alpha}-L}{4(k+1)}  \mathbb{E} \left[ \left\| \bar{\boldsymbol{\theta}}_{k+1}^{s+1}-\bar{\boldsymbol{\theta}}_{0}^{s+1} \right\|^2\right] \notag \\
	&  -\sum_{k=0}^{K-1} \left( \frac{\frac{1}{\alpha }-L}{2(2k+1)} + \frac{3 \alpha \Psi }{B n}\right)\mathbb{E} \left[ \left\| \bar{\boldsymbol{\theta}}_{k}^{s+1}-\bar{\boldsymbol{\theta}}_{0}^{s+1} \right\|^2  \right] \notag \\
	& -\left(\frac{ \alpha L^2}{ n }  + \frac{3 \alpha \Psi }{B n^2} \right) \sum_{k=0}^{K-1}\mathbb{E}\left[\left\|\boldsymbol{\theta}_{k}^{s+1}-\boldsymbol{1}_{n}\bar{\boldsymbol{\theta}}_{k}^{s+1}\right\|^{2} \right] \notag \\
	& -\frac{3 \alpha \Psi K }{B n^2} {\mathbb{E}\left[\left\|\boldsymbol{\theta}_{0}^{s+1}-\boldsymbol{1}_{n}\bar{\boldsymbol{\theta}}_{0}^{s+1}\right\|^{2} \right]}. \label{eq69}
	\end{align}
	This together with (\ref{eq42}) in Lemma \ref{Lem9} produces
	\begin{align}\label{eq70}
	&\mathbb{E} \left[J\left ( \bar{\boldsymbol{\theta}}_{K}^{s+1} \right )\right] \ge \mathbb{E} \left[J\left ( \bar{\boldsymbol{\theta}}_{0}^{s+1} \right )\right]+\frac{\alpha}{2}\mathbb{E} \sum_{k=0}^{K-1}\left[\left\|\nabla J\left(\bar{\boldsymbol{\theta}}_{k}^{s+1}\right)\right\|^{2}\right] \notag \\
	&\quad  -\frac{\alpha VK}{Mn}- \left(\frac{ \alpha L^2}{ n}+\frac{3\alpha\Psi}{B n^2}\right)c_0(s) \notag \\
	&\quad - \left(\frac{\alpha^3 c_2 K \left(L^2+\frac{3\Psi}{Bn}\right) } {n\left(1-\lambda\right)^2}+ \frac{3 \alpha \Psi K}{B n^2 } \right) \mathbb{E}\left[\left\|\boldsymbol{\theta}_{0}^{s+1}-\boldsymbol{1}_{n}\bar{\boldsymbol{\theta}}_{0}^{s+1}\right\|^{2} \right] \notag \\
	& \quad +\sum_{k=0}^{K-1}\left(\frac{\frac{1}{\alpha}-L}{4(k+1)}- {\frac{\alpha^3  c_1 \left(L^2+\frac{3\Psi}{B n}\right)}{ n \left(1-\lambda\right)^2}}\right)\notag\\
	& \qquad \times \mathbb{E}\left[ \left\| \bar{\boldsymbol{\theta}}_{k+1}^{s+1}-\bar{\boldsymbol{\theta}}_{0}^{s+1} \right\|^2\right] \notag\\
	&\quad - \sum_{k=0}^{K-1}\left( {\frac{3\alpha \Psi}{B n}} + \frac{\frac{1}{\alpha}-L}{2\left(2k+1\right)} +  {\frac{\alpha^3  c_1 \left(L^2+\frac{3\Psi}{B n}\right)}{ n \left(1-\lambda\right)^2}}\right)\notag\\
	& \qquad \times \mathbb{E}\left[ \left\| \bar{\boldsymbol{\theta}}_{k}^{s+1}-\bar{\boldsymbol{\theta}}_{0}^{s+1} \right\|^2\right],
	\end{align}
	 By summing the above inequality over $s$ from 0 to $S-1$ and regrouping the terms, we have
	\begin{align}
	&\mathbb{E} \left[J\left ( \bar{\boldsymbol{\theta}}_{K}^{S} \right )\right] \ge \mathbb{E} \left[J\left ( \bar{\boldsymbol{\theta}}_{0}^{1} \right )\right] -\frac{\alpha VKS}{Mn}  \notag\\
	&+\frac{\alpha}{2}\sum_{s=0}^{S-1}\sum_{k=0}^{K-1}\mathbb{E} \left[\left\|\nabla J\left(\bar{\boldsymbol{\theta}}_{k}^{s+1}\right)\right\|^{2}\right]  \notag\\
	&- \left(\frac{ \alpha \left(L^2+\frac{3\Psi}{Bn}\right) \Gamma}{n \left(1-\lambda\right)^2 \left(1-\sigma^2\right)^2} + \frac{3 \alpha \Psi K}{B n^2} \right)\notag \\
	& \quad \times \sum_{s=0}^{S-1}\mathbb{E}\left[\left\| {{\boldsymbol{\theta }}_{0}^{s+1} - \boldsymbol{1}_{n}\bar{\boldsymbol{ \theta }}_{0}^{s+1}} \right\|^{2}\right]  \notag\\
	& - {\frac{2 \alpha^3 \left(L^2+\frac{3\Psi}{Bn}\right)}{ n\left(1-\lambda\right)^2\left(1-\sigma^2\right)} }\sum_{s=0}^{S-1}\mathbb{E}\left[\left\| {{\boldsymbol{y}}_{0}^{s+1} - \boldsymbol{1}_{n}\bar{\boldsymbol{y}}_{0}^{s+1}} \right\|^{2}\right] \notag\\
	&+\sum_{s=0}^{S-1}\sum_{k=0}^{K-1}\left(\frac{\frac{1}{\alpha}-L}{4(k+1)}- {\frac{\alpha^3  c_1 \left(L^2+\frac{3\Psi}{Bn}\right)}{ n \left(1-\lambda\right)^2}  }\right)\notag \\
	& \quad \times\mathbb{E}\left[ \left\| \bar{\boldsymbol{\theta}}_{k+1}^{s+1}-\bar{\boldsymbol{\theta}}_{0}^{s+1} \right\|^2\right] \notag \\
	&- \sum_{s=0}^{S-1}\sum_{k=0}^{K-1}\left( {\frac{3\alpha \Psi}{B n}} + \frac{\frac{1}{\alpha}-L}{2\left(2k+1\right)} +  {\frac{\alpha^3  c_1 \left(L^2+\frac{3\Psi}{Bn}\right) }{ n \left(1-\lambda\right)^2} }\right) \notag \\
	& \quad \times \mathbb{E}\left[ \left\| \bar{\boldsymbol{\theta}}_{k}^{s+1}-\bar{\boldsymbol{\theta}}_{0}^{s+1} \right\|^2\right],\label{eq72}
	\end{align}
	where $\Gamma= \left(1-\lambda\right)\left(1-\sigma^2\right)^2 + 48 \alpha^2 \Psi K $.

	Substituting (\ref{eq48}) and (\ref{eq49}) into (\ref{eq72}), we have
	\begin{align}
	&\mathbb{E} \left[J\left ( \bar{\boldsymbol{\theta}}_{K}^{S} \right )\right] \ge \mathbb{E} \left[J\left ( \bar{\boldsymbol{\theta}}_{0}^{1} \right )\right]-\frac{\alpha VKS}{Mn} \notag \\
	&+\frac{\alpha}{2}\sum_{s=0}^{S-1}\sum_{k=0}^{K-1}\mathbb{E} \left[\left\|\nabla J\left(\bar{\boldsymbol{\theta}}_{k}^{s+1}\right)\right\|^{2}\right] \notag \\
	&- \left( \frac{\alpha \left(L^2+\frac{3\Psi}{Bn}\right) \Gamma}{ n \left(1-\lambda\right) \Xi} + \frac{3 \alpha \Psi K \left(1-\lambda\right)\left(1-\sigma^2\right)^2}{B n^2 \Xi }  \right)  \notag \\
	& \quad \times \mathbb{E}\left[\left\|\boldsymbol{\theta}_{0}^{1}
-{\boldsymbol{1}_{n}}\bar{\boldsymbol{\theta}}_{0}^{1}\right\|^2\right] \notag \\
	& - \left(\frac{2 \alpha^3   \left(L^2+\frac{3\Psi}{Bn}\right) \Gamma}{ n \left(1-\sigma^2\right) \left(1-\lambda\right)^2 \Xi } + \frac{6 \alpha^3 \Psi K \left(1-\sigma^2\right)}{B n^2 \Xi } \right) \notag \\
	& \quad \times \mathbb{E} \left[\left\|\boldsymbol{y}_{0}^{1}
-{\boldsymbol{1}_{n}}\bar{\boldsymbol{y}}_{0}^{1}\right\|^2\right] \notag \\
	& -  \left(\frac{24 \alpha^3 \left(L^2+\frac{3\Psi}{Bn}\right) \Psi  \Gamma}{  \left(1-\sigma^2\right)^2 \left(1-\lambda\right)^2 \Xi}  + \frac{72 \alpha^3 \Psi^2 K}{B n \Xi}\right) \notag \\
	& \quad \times \sum_{s=0}^{S-1}\sum_{k=0}^{K-1}\mathbb{E}\left[  \left\| \bar{\boldsymbol{ \theta }}_{k+1}^{s+1} - \bar{\boldsymbol{ \theta }}_{0}^{s+1}\right\|^2\right] \notag\\
	&  -   \left(\frac{24 \alpha^3 \left(L^2+\frac{3\Psi}{Bn}\right) \Psi  \Gamma}{  \left(1-\sigma^2\right)^2 \left(1-\lambda\right)^2 \Xi}  + \frac{72 \alpha^3 \Psi^2 K}{B n \Xi}\right) \notag \\
	& \quad \times   \sum_{s=0}^{S-1}\sum_{k=0}^{K-1}\mathbb{E}\left[  \left\| \bar{\boldsymbol{ \theta }}_{k}^{s+1} - \bar{\boldsymbol{ \theta }}_{0}^{s+1}\right\|^2\right] \notag\\
	& -  \frac{2\alpha^3 \left(L^2+\frac{3\Psi}{Bn}\right)}{ n \left(1-\lambda\right)^4\left(1-\sigma^2\right)^2}
	\left(36\Psi+\frac{3 \Psi K \left(1-\lambda\right)\left(1-\sigma^2\right)^2\Phi}{\Xi}\right)\notag\\
	& \quad \times \mathbb{E} \left[\left\|\boldsymbol{\theta}_{0}^{1}
-{\boldsymbol{1}_{n}}\bar{\boldsymbol{\theta}}_{0}^{1}\right\|^2\right]\notag\\
	& - {\frac{\alpha^3 \left(1-\sigma^2\right)  \Phi \left(L^2+\frac{3\Psi}{Bn}\right)}{ 4 n\left(1-\lambda\right)^2 \Xi}}\mathbb{E} \left[\left\|\boldsymbol{y}_{0}^{1}
-{\boldsymbol{1}_{n}}\bar{\boldsymbol{y}}_{0}^{1}\right\|^2\right]\notag\\
	& - \frac{3 \alpha^3  \Psi \Phi \left(L^2+\frac{3\Psi}{Bn}\right) }{\left(1-\lambda\right)^2 \Xi} \sum_{s=0}^{S-1}\sum_{k=0}^{K-1}\mathbb{E}\left[  \left\| \bar{\boldsymbol{ \theta }}_{k+1}^{s+1} - \bar{\boldsymbol{ \theta }}_{0}^{s+1}\right\|^2\right]\notag\\
	&  - \frac{3 \alpha^3  \Psi \Phi \left(L^2+\frac{3\Psi}{Bn}\right) }{ \left(1-\lambda\right)^2 \Xi}\sum_{s=0}^{S-1}\sum_{k=0}^{K-1}\mathbb{E}\left[  \left\| \bar{\boldsymbol{ \theta }}_{k}^{s+1} - \bar{\boldsymbol{ \theta }}_{0}^{s+1}\right\|^2\right]\notag\\
	&+\sum_{s=0}^{S-1}\sum_{k=0}^{K-1}\left(\frac{\frac{1}{\alpha}-L}{4(k+1)}- {\frac{ \alpha^3  c_1 \left(L^2+\frac{3\Psi}{Bn}\right) }{ n \left(1-\lambda\right)^2}  }\right)\notag \\
	& \quad \times\mathbb{E}\left[ \left\| \bar{\boldsymbol{\theta}}_{k+1}^{s+1}-\bar{\boldsymbol{\theta}}_{0}^{s+1} \right\|^2\right] \notag \\
	&- \sum_{s=0}^{S-1}\sum_{k=0}^{K-1}\left({\frac{3\alpha \Psi}{B n}} + \frac{\frac{1}{ \alpha}-L}{2\left(2k+1\right)} +  {\frac{\alpha^3  c_1 \left(L^2+\frac{3\Psi}{Bn}\right) }{ n \left(1-\lambda\right)^2}}\right) \notag \\
	& \quad \times \mathbb{E}\left[ \left\| \bar{\boldsymbol{\theta}}_{k}^{s+1}-\bar{\boldsymbol{\theta}}_{0}^{s+1} \right\|^2\right].\notag
	\end{align}
	Regrouping the terms of the above inequality leads to
	\begin{align}
	&\mathbb{E} \left[J\left ( \bar{\boldsymbol{\theta}}_{K}^{S} \right )\right] \ge \mathbb{E} \left[J\left ( \bar{\boldsymbol{\theta}}_{0}^{1} \right )\right]-\frac{\alpha VKS}{Mn}\notag\\
	& +\frac{\alpha}{2}\sum_{s=0}^{S-1}\sum_{k=0}^{K-1}\mathbb{E} \left[\left\|\nabla J\left(\bar{\boldsymbol{\theta}}_{k}^{s+1}\right)\right\|^{2}\right]\notag\\
	& -\frac{\alpha}{n}{\rho}_4\mathbb{E} \left[\left\|\boldsymbol{\theta}_{0}^{1}-{\boldsymbol{1}_{n}}\bar{\boldsymbol{\theta}}_{0}^{1}\right\|^2\right] -\frac{ \alpha^3}{n} {\rho}_5 \mathbb{E} \left[\left\|\boldsymbol{y}_{0}^{1}-{\boldsymbol{1}_{n}}\bar{\boldsymbol{y}}_{0}^{1}\right\|^2\right]\notag\\
	& +\sum_{s=0}^{S-1}\sum_{k=1}^{K-1}\left(\frac{\frac{1}{2\alpha}-\frac{L}{2}}{2k(2k+1)}- \frac{ 2 \alpha^3 \left(L^2+\frac{3\Psi}{Bn}\right) c_1}{ n \left(1-\lambda\right)^2}\right.\notag\\
&\left. \quad -\frac{3 \alpha \Psi}{B n} - 2 {\rho}_2-2 {\rho}_3\right) \mathbb{E}\left[ \left\| \bar{\boldsymbol{\theta}}_{k}^{s+1}-\bar{\boldsymbol{\theta}}_{0}^{s+1} \right\|^2\right]\notag\\
	&  + \left(\frac{\frac{1}{\alpha}-L}{4K} - \frac{ \alpha^3 \left(L^2+\frac{3\Psi}{Bn}\right) c_1}{ n \left(1-\lambda\right)^2} - {\rho}_2-{\rho}_3\right)\notag\\
	& \quad \times
	\sum_{s=0}^{S-1}\mathbb{E}\left[ \left\| \bar{\boldsymbol{\theta}}_{K}^{s+1}-\bar{\boldsymbol{\theta}}_{0}^{s+1} \right\|^2\right],\label{eq75}
	\end{align}
	where
	\begin{subequations}\label{eq76}
	\begin{align}
	{\rho}_1=&36\Psi+\cfrac{3\Psi K \left(1-\lambda\right)\left(1-\sigma^2\right)^2 \Phi}{\Xi},\\
	{\rho}_2=&\cfrac{  24 \alpha^3  \left(L^2+\frac{3\Psi}{Bn}\right) \Psi \Gamma}{ \left(1-\sigma^2\right)^2\left(1-\lambda\right)^2 \Xi}+ \frac{72 \alpha^3 \Psi^2 K}{B n \Xi}, \label{eqy32} \\
{\rho}_3=&\cfrac{ 3 \alpha^3  \left(L^2+\frac{3\Psi}{Bn}\right) \Psi \Phi}{\left(1-\lambda\right)^2 \Xi}, \label{eqy3} \\
	{\rho}_4=&\frac{\Gamma \left(L^2+\frac{3\Psi}{Bn}\right)} { \left(1-\lambda\right) \Xi}
+\frac{ {2\alpha^2 \left(L^2+\frac{3\Psi}{Bn}\right)} {\rho}_1}{\left(1-\lambda\right)^4\left(1-\sigma^2\right)^2}+\frac{3 \Psi K}{B n \Xi}, \label{eq76a}\\
	{\rho}_5=&\frac{ 2 \Gamma\left(L^2+\frac{3\Psi}{Bn}\right) }{\left(1-\lambda\right)^{2}\left(1-\sigma^2\right) \Xi } +\frac{\left(1-\sigma^2\right) \left(L^2+\frac{3\Psi}{Bn}\right)\Phi}{ {4}\left(1-\lambda\right)^2  \Xi}\notag\\
& +\frac{3 \Psi K \left(1-\sigma^2\right)}{B n \Xi}. \label{eq76b}
	\end{align}
	\end{subequations}

	Then by the definitions of $\rho_2, \rho_3$ in \eqref{eqy32}  and \eqref{eqy3}, the coefficients of the last two terms of  \eqref{eq75} satisfy
	\begin{align}
	&\frac{\frac{1}{2\alpha}-\frac{L}{2}}{2k(2k+1)}
	-\frac{3 \alpha \Psi}{B n} -\frac{ 2 \alpha^3 \left(L^2+\frac{3\Psi}{Bn}\right) c_1}{ n \left(1-\lambda\right)^2}-2 {\rho}_2-2 {\rho}_3 \notag \\
	\ge& \frac{\frac{1}{2\alpha}-\frac{L}{2}}{4K^{2}} - \frac{3 \alpha \Psi}{B n} -\frac{144 \alpha^3 \Psi^2 K}{B n \left(\left(1-\lambda\right)^2\left(1-\sigma^2\right)^2-48\alpha^2\Psi K \right)}
\notag \\
	 &-\frac{6 \alpha^3  \Psi \left(L^2+\frac{3\Psi}{Bn}\right) \left( 8\left(1-\lambda\right)\left(3-\lambda\right)+ \left(1-\sigma^2\right)\right)}{  \left(1-\lambda\right)^2\left(\left(1-\lambda\right)^2\left(1-\sigma^2\right)^2-48\alpha^2\Psi K\right)}, \label{eq77} \\
&\frac{\frac{1}{\alpha}-L}{4K}-\frac{\alpha^3  {\left(L^2+\frac{3\Psi}{Bn}\right)} c_1}{ { n } \left(1-\lambda\right)^2}  - {\rho}_2-{\rho}_3 \notag \\
	=& \frac{\frac{1}{\alpha}-L}{4K} -\frac{72 \alpha^3 \Psi^2 K}{B n \left(\left(1-\lambda\right)^2\left(1-\sigma^2\right)^2-48\alpha^2\Psi K\right)}\notag \\
&-\frac{ 3 \alpha^3  \Psi \left(L^2+\frac{3\Psi}{Bn}\right) \left( 8\left(1-\lambda\right)\left(3-\lambda\right)+ \left(1-\sigma^2\right)\right)}{  \left(1-\lambda\right)^2\left(\left(1-\lambda\right)^2\left(1-\sigma^2\right)^2-48\alpha^2\Psi K\right)}.\label{eq78}
	\end{align}

    If $0 < \alpha \le \frac{1}{2L}$, we have
	\begin{align}\label{eq79}
	\frac{1}{ \alpha}-L \ge L.
	\end{align}
   By recalling that $\lambda \le \frac{7+\sigma^2}{8}$  and   $0 < \alpha \le \frac{\left(1-\sigma^2\right)^2}{96\sqrt{3 \Psi K}}$, we get
    \begin{align}\label{eq79b}
    \left(1-\lambda\right)^2\left(1-\sigma^2\right)^2-48\alpha^2\Psi K > \frac{\left(1-\sigma^2\right)^4}{72}.
    \end{align}
    By utilizing (\ref{eq79b}) and the definition of $B_1(\alpha)$ in (\ref{eq_B1}), we derive
    \begin{align}\label{B_1}
    B_1(\alpha) \ge \frac{48 \alpha \Psi K^2 \left(1-\lambda\right)^2\left(1-\sigma^2\right)^2}{n L \left(\left(1-\lambda\right)^2\left(1-\sigma^2\right)^2 - 48 \alpha^2\Psi K\right)}.
    \end{align}
    Since $B \ge B_1(\alpha)$, we have
    \begin{equation}\label{part_77_1}
    \frac{3 \alpha \Psi}{B n} + \frac{144 \alpha^3 \Psi^2 K}{B n \left(\left(1-\lambda\right)^2\left(1-\sigma^2\right)^2-48\alpha^2\Psi K\right)} \le \frac{L}{16 K^2}.
    \end{equation}
    We use (\ref{eq79b}) and $B_2(\alpha)$ in (\ref{eq_B2}) to obtain,
    \begin{equation}\label{B_2}
    B_2(\alpha) \ge \frac{288 \alpha^3 \Psi^2 K^2 \left( 8\left(1-\lambda\right)\left(3-\lambda\right)+ \left(1-\sigma^2\right)\right)}{b_1(\alpha)},
    \end{equation}
    where $b_1(\alpha)=n L (1-\lambda)^2\left((1-\lambda)^2(1-\sigma^2)^2-48\alpha^2\Psi K\right)-96 \alpha^3 \Psi K^2 L^2 n \left( 8\left(1-\lambda\right)\left(3-\lambda\right)+ \left(1-\sigma^2\right)\right)>0$ with $0 < \alpha < \frac{\left(1-\sigma^2\right)^2}{24 \sqrt[3]{4 \Psi K^2 L \left(\left(1-\sigma^2\right)^2+24\left(1-\sigma^2\right)\right)}}$.
    Recalling $B \ge B_2(\alpha)$, we have
    \begin{equation}\label{part_77_2}
    \frac{ 6 \alpha^3  \Psi \left(L^2+\frac{3\Psi}{B n}\right) \left( 8\left(1-\lambda\right)\left(3-\lambda\right)+ \left(1-\sigma^2\right)\right)}{  \left(1-\lambda\right)^2\left(\left(1-\lambda\right)^2\left(1-\sigma^2\right)^2-48\alpha^2\Psi K\right)} \le \frac{L}{16 K^2}.
    \end{equation}
    Combining (\ref{part_77_1}), (\ref{part_77_2}) and (\ref{eq79}) yields (\ref{eq77}) $ \ge 0$.

    Similarly, applying (\ref{eq79b}) and $B_1(\alpha)$ in (\ref{eq_B1}) leads to
    \begin{equation}\label{B_3}
    B_1(\alpha) \ge \frac{576 \alpha^3 \Psi^2 K^2}{n L \left(\left(1-\lambda\right)^2\left(1-\sigma^2\right)^2-48\alpha^2\Psi K\right)},
    \end{equation}
    with $0 < \alpha < \frac{\left(1-\sigma^2\right)^2}{16 \sqrt{3 \Psi}}$.
    Since $B \ge B_1(\alpha)$, we have
    \begin{equation}\label{part_78_1}
    \frac{72 \alpha^3 \Psi^2 K}{B n \left(\left(1-\lambda\right)^2\left(1-\sigma^2\right)^2-48\alpha^2\Psi K\right)} \le \frac{L}{8K}.
    \end{equation}
    By using (\ref{eq79b}) and $B_2(\alpha)$ in (\ref{eq_B2}), we obtain that
    \begin{equation}\label{B_4}
    B_2(\alpha) \ge \frac{72 \alpha^3 \Psi^2 K \left( 8\left(1-\lambda\right)\left(3-\lambda\right)+ \left(1-\sigma^2\right)\right)}{b_2(\alpha)},
    \end{equation}
    where $b_2(\alpha) =n L (1-\lambda)^2\left((1-\lambda)^2(1-\sigma^2)^2-48\alpha^2\Psi K\right)-24 \alpha^3 \Psi K L^2 n \left( 8\left(1-\lambda\right)\left(3-\lambda\right)+ \left(1-\sigma^2\right)\right) >0 $ with $0 < \alpha < \frac{\left(1-\sigma^2\right)^2}{24 \sqrt[3]{ \Psi K L \left(\left(1-\sigma^2\right)^2+24\left(1-\sigma^2\right)\right)}}$. Since $B \ge B_2(\alpha)$, we have
    \begin{equation}\label{part_78_2}
    \frac{  3 \alpha^3  \Psi \left(L^2+\frac{3\Psi}{Bn}\right) \left( 8\left(1-\lambda\right)\left(3-\lambda\right)+ \left(1-\sigma^2\right)\right)}{  \left(1-\lambda\right)^2\left(\left(1-\lambda\right)^2\left(1-\sigma^2\right)^2-48\alpha^2\Psi K\right)} \le \frac{L}{8K}.
    \end{equation}
    By combining (\ref{part_78_1}), (\ref{part_78_2}) and (\ref{eq79}), we observe that (\ref{eq78}) $ \ge 0$.

    Therefore, we can drop the last two terms of (\ref{eq75}) to obtain
	\begin{equation*}
	\begin{aligned}
	\frac{\alpha}{2}\sum_{s=0}^{S-1}\sum_{k=0}^{K-1} & \mathbb{E} \left[\left\|\nabla J\left(\bar{\boldsymbol{\theta}}_{k}^{s+1}\right)\right\|^{2}\right]\le \mathbb{E} \left[J\left ( \bar{\boldsymbol{\theta}}_{K}^{S} \right )\right] - \mathbb{E} \left[J\left ( \bar{\boldsymbol{\theta}}_{0}^{1} \right )\right]\\
	&+\frac{ \alpha VKS}{Mn} +\frac{\alpha}{n}{\rho}_4\mathbb{E} \left[\left\|\boldsymbol{\theta}_{0}^{1}-{\boldsymbol{1}_{n}}\bar{\boldsymbol{\theta}}_{0}^{1}\right\|^2\right]\\
	&+\frac{\alpha^3}{n}{\rho}_5\mathbb{E} \left[\left\|\boldsymbol{y}_{0}^{1}-{\boldsymbol{1}_{n}}\bar{\boldsymbol{y}}_{0}^{1}\right\|^2\right].
	\end{aligned}
	\end{equation*}
	Multiplying both sides of the above inequality with $\frac{2}{\alpha}$ yields
	\begin{equation}\label{eq82}
	\begin{aligned}
	\mathbb{E} \left[\left\|\nabla J\left(\bar{\boldsymbol{\theta}}_{\text{out}}\right)\right\|^{2}\right] & \le \frac{2\left(J\left ( {\boldsymbol{\theta}}^{*} \right ) -  J\left ( \bar{\boldsymbol{\theta}}_{0}^{1} \right )\right)}{\alpha KS} +\frac{2 V}{Mn}\\
	& + \frac{2}{n}{\rho}_4\mathbb{E} \left[\left\|\boldsymbol{\theta}_{0}^{1}-{\boldsymbol{1}_{n}}\bar{\boldsymbol{\theta}}_{0}^{1}\right\|^2\right]\\
	& + \frac{2 \alpha^2}{n}{\rho}_5 \mathbb{E} \left[\left\|\boldsymbol{y}_{0}^{1}-{\boldsymbol{1}_{n}}\bar{\boldsymbol{y}}_{0}^{1}\right\|^2\right].
	\end{aligned}
	\end{equation}
	
	Using ${\rho}_4$, ${\rho}_5$ in (\ref{eq76}) and $\lambda \le \frac{7+\sigma^2}{8}$ gives
	\begin{equation}\label{eq83}
	\begin{aligned}
	{\rho}_4
	&= \frac{6 \alpha^2 \Psi K \left(L^2+\frac{3\Psi}{Bn}\right) \left(\frac{12}{K}+\frac{1-\sigma^2}{1-\lambda}+8-\frac{576\alpha^2 \Psi}{\left(1-\lambda\right)^2\left(1-\sigma^2\right)^2}\right)}{\left(1-\lambda\right)^2\left(\left(1-\lambda\right)^2\left(1-\sigma^2\right)^2-48\alpha^2\Psi K\right)}\\
	& + \frac{\left(L^2 B n +3\Psi\right)\left(\left(1-\sigma^2\right)^2+\frac{48 \alpha^2 \Psi K}{1-\lambda}\right)+3\Psi K}{B n\left(\left(1-\lambda\right)^2\left(1-\sigma^2\right)^2-48\alpha^2\Psi K\right)} \le v_{1},
	\end{aligned}
	\end{equation}
	\begin{equation}\label{eq84}
	\begin{aligned}
	{\rho}_5
	&=\frac{\left(L^2+\frac{3\Psi}{Bn}\right)\left(16 \left(1-\lambda\right)+\frac{384\alpha^2 \Psi K}{\left(1-\sigma^2\right)^2}+ \left(1-\sigma^2\right)\right)}{\frac{4\left(1-\lambda\right)^2}{1-\sigma^2}\left(\left(1-\lambda\right)^2\left(1-\sigma^2\right)^2-48\alpha^2\Psi K\right)}\\
    & +\frac{3 \Psi K \left(1-\sigma^2\right)}{B n \left(\left(1-\lambda\right)^2\left(1-\sigma^2\right)^2-48\alpha^2\Psi K\right)} \le v_{2}.
	\end{aligned}
	\end{equation}
	Substituting the above two inequalities into (\ref{eq82}) proves Theorem \ref{Thm1}.
\end{proof}

\section{Numerical results}
In this section, we evaluate the performance of Algorithm \ref{Alg1} on the classical reinforcement learning environment in OpenAI Gym. 
Consider a network of $N$ agents ($N$ cars) distributed at the bottom of a valley, where agents want to reach their target horizontal positions (flags) on top of the mountain. Agents are able to observe the global state, which is the position of target flag, the position and velocity of other agents. Here we consider the continuous action space, and the possible actions for each agent is the power coefficient ranging from $\left[-1.0, 1.0\right]$. A reward of 100 is awarded when each agent reaches to the target flags. Otherwise the reward is plus $-1$ based on the amount of energy consumed each step and the conflict with other agents.
Agents are connected through a communication network with other agents. 

We compare the proposed Algorithm \ref{Alg1}, abbreviated as DGT-SVRPG, with distributed G(PO)MDP in the multi-agent settings (D-GPOMDP), and distributed G(PO)MDP with gradient tracking (DGT-GPOMDP).
Our implementation is based on the implementation of SVRPG provided by rllab library \cite{papini2018stochastic}.
For the parameterized policy function of each agent $i$, we utilize the following Gaussian policy with a fixed standard deviation $\tilde{\sigma}$.
\begin{equation*}
\pi_{i,\boldsymbol{\theta}}(\boldsymbol{a}_{i}|\boldsymbol{s}) = \frac{1}{\sqrt{2\pi}\tilde{\sigma}}\exp\left(-\frac{\left(\mu_{i,\boldsymbol{\theta }}(\boldsymbol{s})-\boldsymbol{a}_{i}\right)^{2}}{2\tilde{\sigma}^2}\right),
\end{equation*}
where the mean $\mu_{i, \boldsymbol{\theta}}$ is modeled by a neural network with Tanh as the activation function. Under the Gaussian policy, it is easy to verify that Assumption \ref{Asu1} is satisfied.

According to the practical suggestions in SVRPG \cite{papini2018stochastic}, we use the adaptive step size by incorporating Adam to improve the performance. Recall that the proposed DGT-SVRPG algorithm performs the update of policy parameter ${\boldsymbol{\theta}}$ through the auxiliary variable ${\boldsymbol{y}}$ in the sub-iteration.
For each agent $i$, the update using Adam is  given by
\begin{equation*}
{\boldsymbol{\theta}}_{i,k+1}^{s+1} = \sum_{r \in \mathcal{N}_i}w_{ir}{\boldsymbol{\theta}}_{r,k}^{s+1} +\alpha_{k}^{s+1}\left({\boldsymbol{y}}_{i,k}^{s+1} \right),
\end{equation*}
where $\alpha_{k}^{s+1}$ is the adaptive step size associated with the sub-iterations. We also use Adam in the practical implementations of D-GPOMDP and DGT-GPOMDP for a fair comparison.

For the parameters used in the experiments, the task horizon is set as $1000$ and the number of total trajectories is 4000. We set $M=10$, $B=5$, $K=2$, $\gamma = 0.999$ by considering the values suggested in SVRPG \cite{papini2018stochastic} and theoretical analysis in Corollary \ref{Col1}.
We utilize the star graph as the communication network for multiple agents. We run each experiment repeatedly for 20 times with a random policy initialization, and present the averaged global returns with standard deviation. Fig. \ref{fig:different algorithms} demonstrates the performance comparison of different algorithms. It is observed that the proposed DGT-SVRPG algorithm converges to a global return relatively higher than the other two algorithms without variance reduction.

\begin{figure}[htbp]
\centering
\subfloat[3-agent MountainCar environment]{ \includegraphics[width=0.45\textwidth]{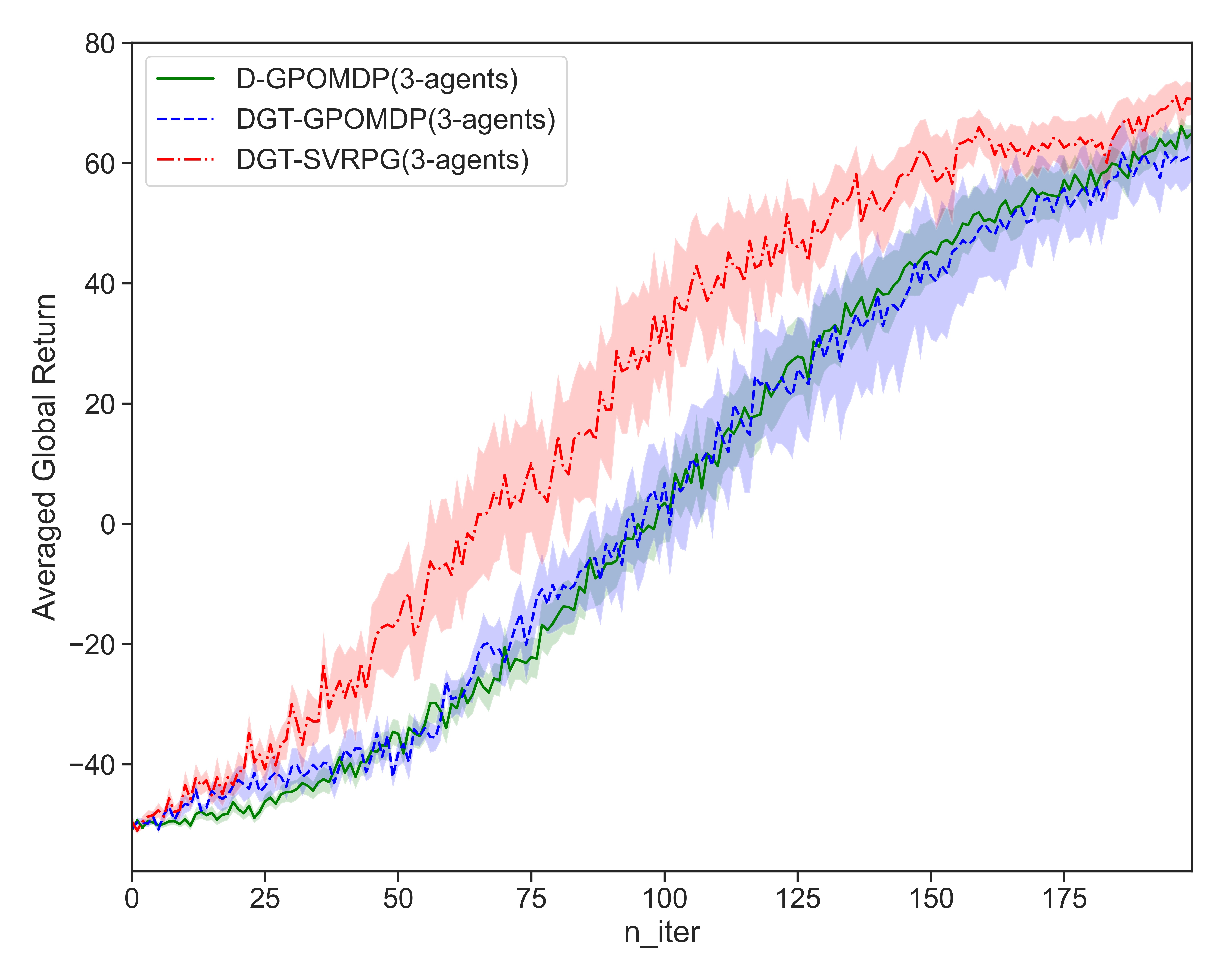}
\label{fig:1}
}
\hfill
\subfloat[5-agent MountainCar environment]{ \includegraphics[width=0.45\textwidth]{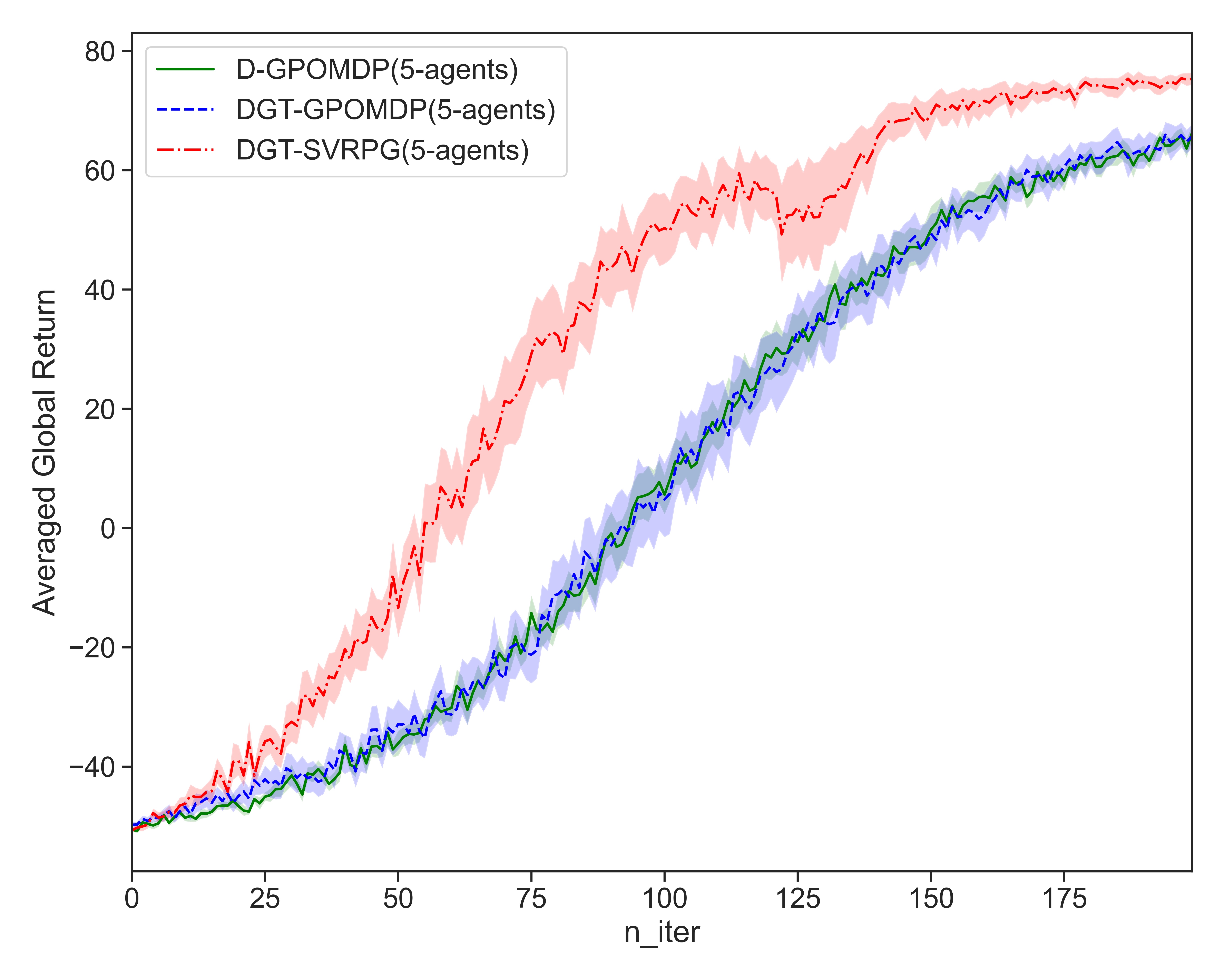}
\label{fig:2}
}
\caption{The averaged global returns of different algorithms}
\label{fig:different algorithms}
\end{figure}

{\textbf{Influence of mini-batch size $B$ on performance}}. We run DGT-SVRPG algorithm with different mini-batch sizes within each epoch to analyze its influence on the performance. Batch size (number of sampled trajectories in the outer loop) is fixed as $M=10$ and mini-batch size (number of sampled trajectories in the inner loop) is selected as $B=[3, 5, 7]$. As the mini-batch size increases, we also scale the step size proportionally such that $\alpha = [0.0015,0.0025,0.0035]$.
The effect of mini-batch size on the performance of DGT-SVRPG is displayed in Fig. \ref{fig:different batchsizes}, which clearly show that the mini-batch size $B=7$ achieves the best performance. While larger mini-batch size requires more sampled trajectories and increases the sampling and computational expenses.
Therefore, mini-batch size $B$ should be properly chosen to balance the performance and computation cost in practice.

\begin{figure}[htbp]
\centering
\subfloat[3-agent MountainCar environment]{ \includegraphics[width=0.45\textwidth]{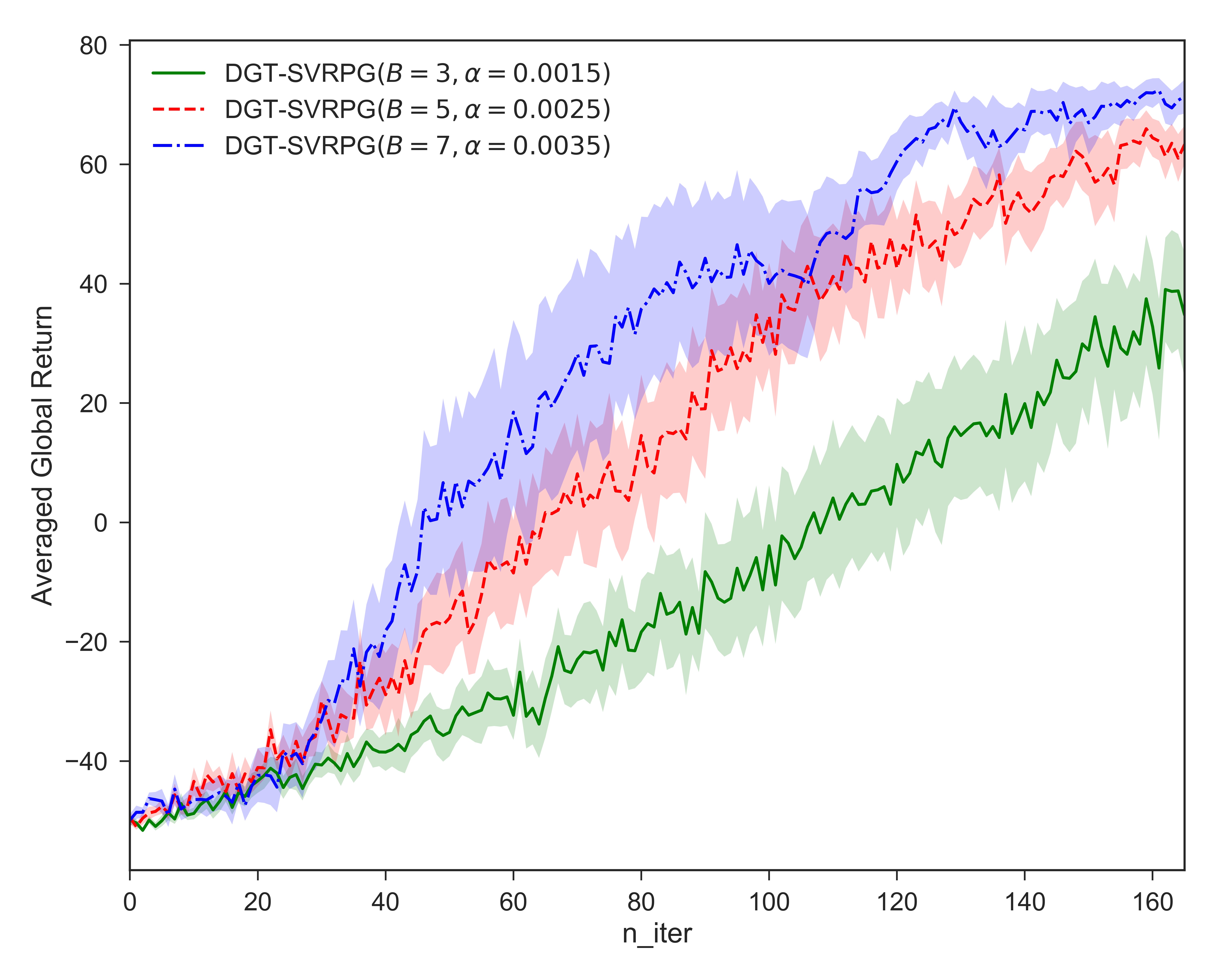}
\label{fig:3}
}
\hfill
\subfloat[5-agent MountainCar environment]{ \includegraphics[width=0.45\textwidth]{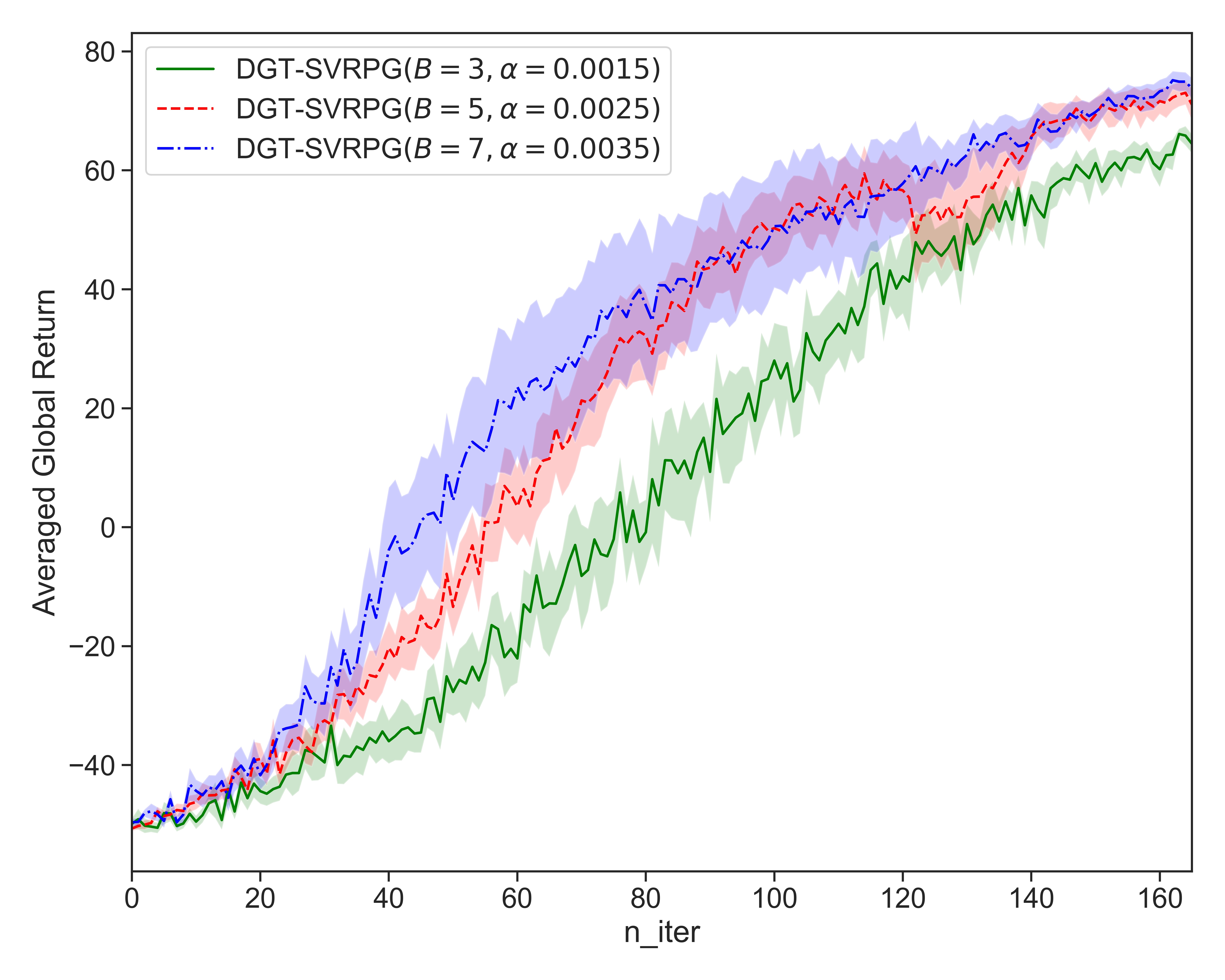}
\label{fig:4}
}
\caption{Performance of DGT-SVRPG with different mini-batch sizes}
\label{fig:different batchsizes}
\end{figure}

{\textbf{Influence of graph on performance}}. Theorem \ref{Thm1} shows the convergence of Algorithm \ref{Alg1} is related to the spectral norm $\sigma$ of matrix $\boldsymbol{W}-\frac{1}{n}{\boldsymbol{1}_{n}}{\boldsymbol{1}_{n}^{\rm{T}}}$, which implies the influence of communication network topology on the convergence. We compare the performance of DGT-SVRPG algorithm with different communication network graphs in multiple-agent MountainCar environment.  The averaged global returns of different communication graphs are shown in Fig. \ref{fig:different graphs}.
As can be seen, DGT-SVRPG in the complete graph is more superior than ring graph and star graph.

\begin{figure}[htbp]
\centering
\subfloat[3-agent MountainCar environment]{ \includegraphics[width=0.45\textwidth]{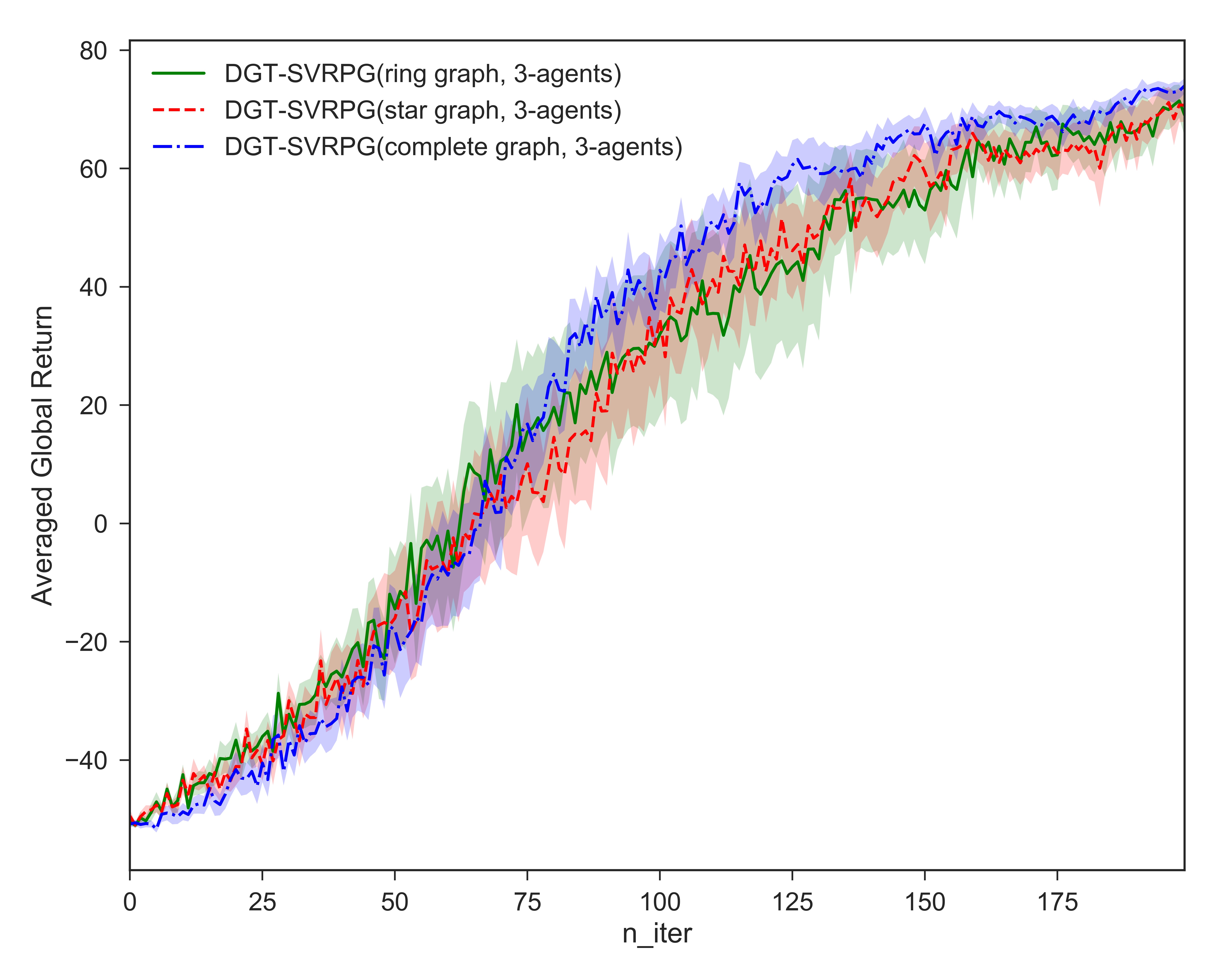}
\label{fig:5}
}
\hfill
\subfloat[5-agent MountainCar environment]{ \includegraphics[width=0.45\textwidth]{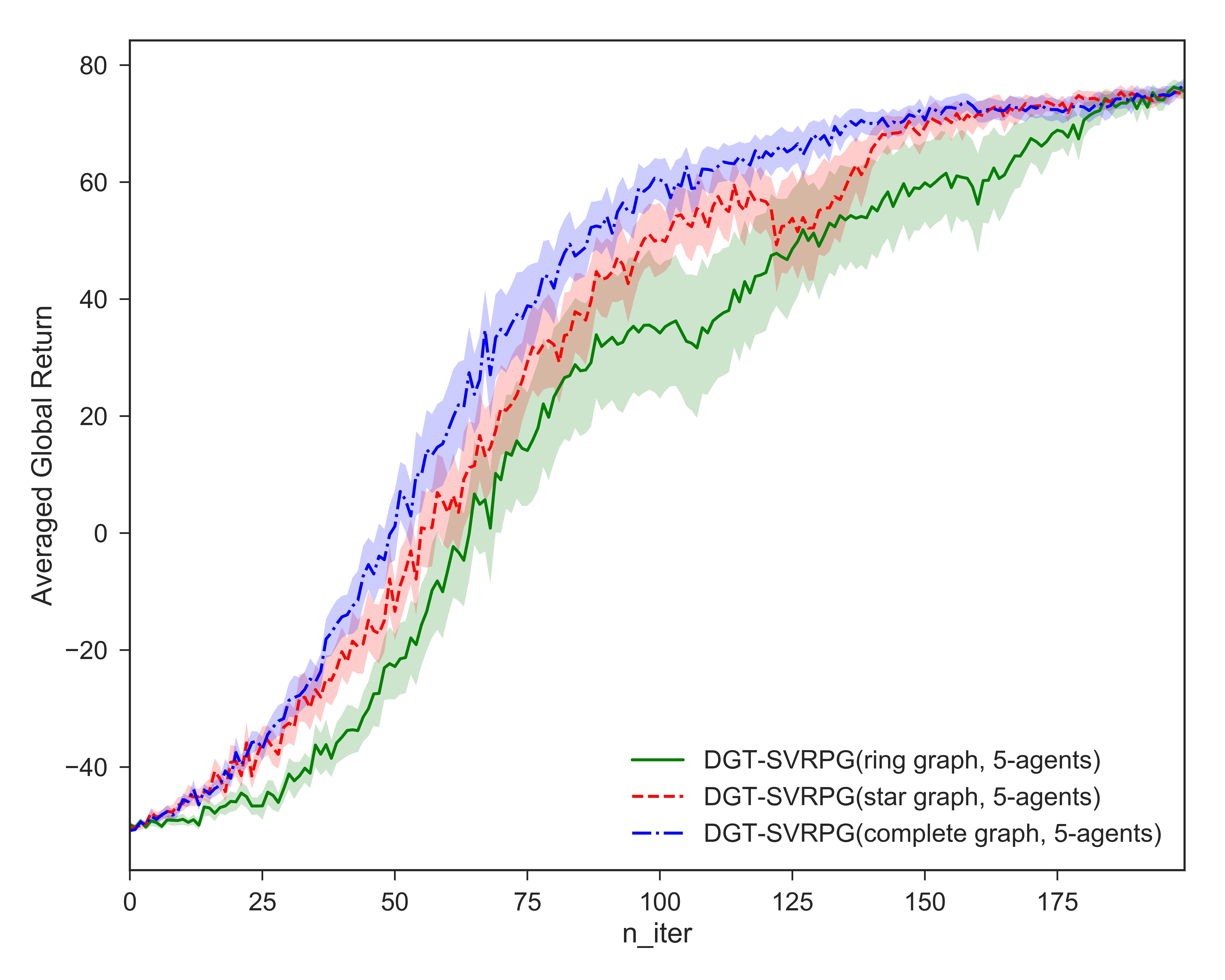}
\label{fig:6}
}
\caption{Performance of DGT-SVRPG with different communication graphs.}
\label{fig:different graphs}
\end{figure}   

\section{Conclusion}
In this paper, we propose a distributed stochastic policy gradient algorithm with variance reduction and gradient tracking for cooperative MARL over networks. The importance weight is also incorporated to address the distribution shift problem in the sampling. We prove that with appropriate step size and mini-batch size, the proposed algorithm converges to an $\epsilon$-approximate stationary point of the non-concave performance function. Experiments validate the advantages of our proposed algorithm. It is of interest to establish the finite sample analysis and global convergence of distributed policy gradient in MARL.


%

\appendices

\section{Auxiliary Results}
In this section, we provide some auxiliary results for analyzing the convergence of Algorithm \ref{Alg1}.
We first establish the interrelationships between the consensus error $\mathbb{E}\left[\left\|\boldsymbol{\theta}_{k}^{s+1}-\boldsymbol{1}_{n}\bar{\boldsymbol{\theta}}_{k}^{s+1}\right\|^2\right]$ and gradient tracking error $\mathbb{E}\left[\left\|\boldsymbol{y}_{k}^{s+1}-\boldsymbol{1}_{n}\bar{\boldsymbol{y}}_{k}^{s+1}\right\|^2\right]$.

\begin{Lem}\label{Lem6}
	Let Assumption 4 hold. Then we have the following inequalities  for any $ k,s \ge 0$.
	\begin{equation}\label{eq22}
	\begin{aligned}
	&\mathbb{E}\left[\left\| {{{\boldsymbol{\theta }}_{k+1}^{s+1}} - \boldsymbol{1}_{n}{{ \bar{\boldsymbol{ \theta }}_{k+1}^{s+1} }}} \right\|^2\right]
	\leq  2\mathbb{E}\left[  \left\| {{\boldsymbol{\theta }}_{k}^{s+1} - \boldsymbol{1}_{n}\bar{\boldsymbol{ \theta }}_{k}^{s+1}} \right\|^2\right] \\
	&\qquad+ 2{\alpha}^2\mathbb{E}\left[\left\| {{{\boldsymbol{y}}_{k}^{s+1}} -\boldsymbol{1}_{n}\bar{\boldsymbol{y}}_{k}^{s+1}} \right\|^2\right],
	\end{aligned}
	\end{equation}
	\begin{equation}\label{eq23}
	\begin{aligned}
	&\mathbb{E}\left[\left\| {{{\boldsymbol{\theta }}_{k+1}^{s+1}} - \boldsymbol{1}_{n}{{ \bar{\boldsymbol{ \theta }}_{k+1}^{s+1} }}} \right\|^2\right]
	\leq \frac{1+\sigma^2}{2}\mathbb{E}\left[  \left\| {{\boldsymbol{\theta }}_{k}^{s+1} - \boldsymbol{1}_{n}\bar{\boldsymbol{ \theta }}_{k}^{s+1}} \right\|^2\right]\\
	&\qquad+ \frac{2{\alpha}^2}{1-\sigma^2}\mathbb{E}\left[\left\| {{{\boldsymbol{y}}_{k}^{s+1}} -\boldsymbol{1}_{n}\bar{\boldsymbol{y}}_{k}^{s+1}} \right\|^2\right].
	\end{aligned}
	\end{equation}
\end{Lem}

\begin{proof}
	With the definition in  (\ref{eq13a}), we have
	\begin{align}
	&\left\| {{{\boldsymbol{\theta }}_{k+1}^{s+1}} - \boldsymbol{1}_{n}{ \bar{\boldsymbol{ \theta }}_{k+1}^{s+1} }} \right\|^2 \notag \\
	\stackrel{(a)}{=}& \left\| {\boldsymbol{W}{{\boldsymbol{\theta }}_{k}^{s+1}} - \boldsymbol{1}_{n}{{\bar{\boldsymbol{ \theta }}}_{k}^{s+1}} + \alpha {{\boldsymbol{y}}_{k}^{s+1}}  - \alpha \boldsymbol{1}_{n}{\bar{\boldsymbol{ y}}_{k}^{s+1}}} \right\|^2 \notag \\
	\stackrel{(b)}{\leq}& \left(1+\eta\right)\left\| \boldsymbol{W}{{\boldsymbol{\theta }}_{k}^{s+1}} - \boldsymbol{1}_{n}{{\bar{\boldsymbol{ \theta }}}_{k}^{s+1}}\right\|^2 \notag \\
	&+ \left(1+\frac{1}{\eta}\right) {\alpha^{2}} \left\|{{{\boldsymbol{y}}_{k}^{s+1}} -\boldsymbol{1}_{n} \bar{\boldsymbol{y}}_{k}^{s+1}}\right\|^2  \notag \\
	\stackrel{(c)}{\leq}& \left(1+\eta\right){\sigma}^2\left\|{{\boldsymbol{\theta }}_{k}^{s+1} - \boldsymbol{1}_{n}\bar{\boldsymbol{ \theta }}_{k}^{s+1}}\right\|^2  \notag\\
	&+ \left(1+\frac{1}{\eta}\right) {\alpha^{2}} \left\|{{{\boldsymbol{y}}_{k}^{s+1}} -\boldsymbol{1}_{n} \bar{\boldsymbol{y}}_{k}^{s+1}}\right\|^2, \label{eq24}
	\end{align}
	where in (a) we use  (\ref{eq14a}), (\ref{eq14b})  and the doubly stochastic $\boldsymbol{W}$  {from} Assumption \ref{Asu4}, in (b) we apply the basic Young's inequality and in (c) we use the following lemma from \cite{xin2020variance}
    \begin{align}
    \left\|\boldsymbol{W}\boldsymbol{\theta}_{k}^{s+1}
    -\boldsymbol{1}_{n}\bar{\boldsymbol{\theta}}_{k}^{s+1}\right\|
    \le \sigma \left\|\boldsymbol{\theta}_{k}^{s+1}
    -\boldsymbol{1}_{n}\bar{\boldsymbol{\theta}}_{k}^{s+1}\right\|. \label{ineq_sigma}
    \end{align}	
	By setting $\eta=$1 and $\eta=\frac{1-\sigma^2}{2\sigma^2}$ in (\ref{eq24}) respectively, we obtain
	\begin{align}
	\left\| {{{\boldsymbol{\theta }}_{k+1}^{s+1}} - \boldsymbol{1}_{n}{{ \bar{\boldsymbol{ \theta }}_{k+1}^{s+1} }}} \right\|^2
	&\leq  2 \left\| {{\boldsymbol{\theta }}_{k}^{s+1} - \boldsymbol{1}_{n}\bar{\boldsymbol{ \theta }}_{k}^{s+1}} \right\|^2 \notag \\
	& + 2{\alpha}^2\left\| {{{\boldsymbol{y}}_{k}^{s+1}} -\boldsymbol{1}_{n}\bar{\boldsymbol{y}}_{k}^{s+1}} \right\|^2, \notag \\
	\left\| {{{\boldsymbol{\theta }}_{k+1}^{s+1}} - \boldsymbol{1}_{n}{{ \bar{\boldsymbol{ \theta }}_{k+1}^{s+1} }}} \right\|^2 &\leq \frac{1+\sigma^2}{2} \left\| {{\boldsymbol{\theta }}_{k}^{s+1} - \boldsymbol{1}_{n}\bar{\boldsymbol{ \theta }}_{k}^{s+1}} \right\|^2 \notag \\
	& + \frac{2{\alpha}^2}{1-\sigma^2}\left\| {{{\boldsymbol{y}}_{k}^{s+1}} -\boldsymbol{1}_{n}\bar{\boldsymbol{y}}_{k}^{s+1}} \right\|^2. \notag
	\end{align}
	Taking the expectations of the above two inequalities leads to (\ref{eq22}) and (\ref{eq23}).
\end{proof}

\begin{Lem}\label{Lem7}
	Suppose Assumptions \ref{Asu1}, \ref{Asu3}, and \ref{Asu4} hold. Define $\Psi=2(C_g^2C_\omega+L_g^2)$. If $\alpha$ satisfies $0 < \alpha < \frac{1-\sigma^2}{4\sqrt{6 \Psi}}$, then
	\begin{align}
	\mathbb{E}\left[\left\| {\boldsymbol{y}_{k+1}^{s+1}} - \boldsymbol{1}_{n}{ \bar{\boldsymbol{y}}_{k+1}^{s+1} } \right\|^2\right]
	\leq& \frac{36\Psi }{1-\sigma^2}\mathbb{E}\left[  \left\| {{\boldsymbol{\theta }}_{k}^{s+1} - \boldsymbol{1}_{n}\bar{\boldsymbol{ \theta }}_{k}^{s+1}} \right\|^2\right] \notag\\
	&+ \frac{3+\sigma^2}{4}\mathbb{E}\left[\left\| {\boldsymbol{y}}_{k}^{s+1} -\boldsymbol{1}_{n}\bar{\boldsymbol{y}}_{k}^{s+1} \right\|^2\right] \notag\\
	&+ \frac{12\Psi n}{1-\sigma^2}\mathbb{E}\left[  \left\| \bar{\boldsymbol{ \theta }}_{k+1}^{s+1} - \bar{\boldsymbol{ \theta }}_{0}^{s+1}\right\|^2\right]\notag\\
	&+ \frac{12\Psi n}{1-\sigma^2}\mathbb{E}\left[  \left\| \bar{\boldsymbol{ \theta }}_{k}^{s+1} - \bar{\boldsymbol{ \theta }}_{0}^{s+1}\right\|^2\right] \notag\\
	&+ \frac{24\Psi }{1-\sigma^2}\mathbb{E}\left[  \left\| {{\boldsymbol{\theta }}_{0}^{s+1} - \boldsymbol{1}_{n}\bar{\boldsymbol{ \theta }}_{0}^{s+1}} \right\|^2\right]. \label{eq26}
	\end{align}
\end{Lem}

\begin{proof}
	Using $\boldsymbol{y}_{k+1}^{s+1}$ in (\ref{eq13b}) and $\bar{\boldsymbol{y}}_{k+1}^{s+1}$ in (\ref{eq14b}), we have
	\begin{equation}\label{eq27}
	\begin{aligned}
	&{\left\| {\boldsymbol{y}}_{k+1}^{s+1} - {\boldsymbol{1}_{n}}{ {\bar{\boldsymbol{y}}}_{k+1}^{s+1} } \right\|^2}\\
	=& {\left\|  \left( {{\boldsymbol{I}_n} - \frac{1}{n}{\boldsymbol{1}_{n}}{\boldsymbol{1}_{n}^{\rm{T}}}} \right)\left({\boldsymbol{W}}{{\boldsymbol{y}}_{k}^{s+1}} + {{\boldsymbol{v}}_{k+1}^{s+1}}-{{\boldsymbol{v}}_{k}^{s+1}}\right)\right\|^2}\\
	\stackrel{(a)}{=}& {\left\| {{\boldsymbol{W}}{{\boldsymbol{y}}_{k}^{s+1}} - {\boldsymbol{1}_{n}} {{{\bar{\boldsymbol{ y}}}_{k}^{s+1}}}  + \left( {{\boldsymbol{I}_n} - \frac{1}{n}{\boldsymbol{1}_{n}}{\boldsymbol{1}_{n}^{\rm{T}}}} \right)\left( {{{\boldsymbol{v}}_{k+1}^{s+1}} - {{\boldsymbol{v}}_{k}^{s+1}}} \right)} \right\|^2}\\
	\stackrel{(b)}{\le}& \left(1+\frac{1-\sigma^2}{2\sigma^2}\right){\left\| {{\boldsymbol{W}}{{\boldsymbol{y}}_{k}^{s+1}} - {\boldsymbol{1}_{n}}\bar{\boldsymbol{ y}}_{k}^{s+1}} \right\|^2}\\
	&  + \left(1+\frac{2\sigma^2}{1-\sigma^2}\right){\left\| {{{\boldsymbol{v}}_{k+1}^{s+1}} - {{\boldsymbol{v}}_{k}^{s+1}}} \right\|^2}\\
	\stackrel{(c)}{\le}& \frac{1+\sigma^2}{2}{\left\| {{{\boldsymbol{y}}_{k}^{s+1}} - {\boldsymbol{1}_{n}}\bar{\boldsymbol{ y}}_{k}^{s+1}} \right\|^2}  + \frac{2}{1-\sigma^2}{\left\| {{{\boldsymbol{v}}_{k+1}^{s+1}} - {{\boldsymbol{v}}_{k}^{s+1}}} \right\|^2},\\
	\end{aligned}
	\end{equation}
	where in (a) we use  the doubly stochastic ${\boldsymbol{W}}$ from Assumption \ref{Asu4}, in (b) we use $\left\|\boldsymbol{I}_n-\frac{\boldsymbol{1}_{n}{\boldsymbol{1}_{n}^{\rm{T}}}}{n}\right\| =1$ and the Young's inequality with $\eta=\frac{1-\sigma^2}{2\sigma^2}$, and in (c) we use (\ref{ineq_sigma}).
	Taking the expectations of both sides of (\ref{eq27}) yields
	\begin{equation}\label{eq28}
	\begin{aligned}
	&\mathbb{E}\left[ {{{\left\| {{{\boldsymbol{y}}_{k+1}^{s+1}} - {\boldsymbol{1}_{n}}{\bar{\boldsymbol{y}}}_{k+1}^{s+1}} \right\|}^2}} \right]
	\leq \frac{1+\sigma^2}{2}\mathbb{E}\left[ {{{\left\| {{{\boldsymbol{y}}_{k}^{s+1}} - {\boldsymbol{1}_{n}}\bar{\boldsymbol{ y}}_{k}^{s+1}} \right\|}^2}} \right]\\
	&\qquad +\frac{2}{1-\sigma^2}\mathbb{E}\left[ {{{\left\| {{{\boldsymbol{v}}_{k+1}^{s+1}} - {{\boldsymbol{v}}_{k}^{s+1}}} \right\|}^2}} \right].
	\end{aligned}
	\end{equation}
	
	Now, we derive the upper bound for $\mathbb{E}\left[ {{{\left\| {{{\boldsymbol{v}}_{k+1}^{s+1}} - {{\boldsymbol{v}}_{k}^{s+1}}} \right\|}^2}} \right]$. Using $\boldsymbol{v}_{i,k}^{s+1}$ in (\ref{eq8}) leads to
	\begin{align}
	&\left\|\boldsymbol{v}_{i,k+1}^{s+1}-\boldsymbol{v}_{i,k}^{s+1}\right\|^2 \notag\\
	=&\left\|\frac{1}{B}\sum_{b=1}^{B}\left(g_i(\tau_{i,b} \left|{\boldsymbol{\theta}}_{i,k+1}^{s+1}\right.)-\omega(\tau_{i,b} \left|{\boldsymbol{\theta}}_{i,k+1}^{s+1}\right.,\tilde{\boldsymbol{\theta}}_i^s)g_i(\tau_{i,b}\left|\tilde{\boldsymbol{\theta}}_i^s\right.)\right)\right.\notag\\
	-&\left.\frac{1}{B}\sum_{b^{\prime}=1}^{B}\left(g_i(\tau_{i,b^{\prime}}\left|{\boldsymbol{\theta}}_{i,k}^{s+1}\right.)-\omega(\tau_{i,b^{\prime}}\left|{\boldsymbol{\theta}}_{i,k}^{s+1}\right.,\tilde{\boldsymbol{\theta}}_i^s)g_i(\tau_{i,b^{\prime}}\left|\tilde{\boldsymbol{\theta}}_i^s\right.)\right)\right\|^2  \notag\\
	\leq & \frac{2}{B^2}\left\|\sum_{b=1}^{B}\left(g_i(\tau_{i,b}\left|{\boldsymbol{\theta}}_{i,k+1}^{s+1}\right.)-\omega(\tau_{i,b}\left|{\boldsymbol{\theta}}_{i,k+1}^{s+1}\right.,\tilde{\boldsymbol{\theta}}_i^s)g_i(\tau_{i,b}\left|\tilde{\boldsymbol{\theta}}_i^s\right.)\right)\right\|^2
	\notag \\
	+&\frac{2}{B^2}\left\|\sum_{b^{\prime}=1}^{B}\left(g_i(\tau_{i,b^{\prime}}\left|{\boldsymbol{\theta}}_{i,k}^{s+1}\right.)-\omega(\tau_{i,b^{\prime}}\left|{\boldsymbol{\theta}}_{i,k}^{s+1}\right.,\tilde{\boldsymbol{\theta}}_i^s)g_i(\tau_{i,b^{\prime}}\left|\tilde{\boldsymbol{\theta}}_i^s\right.)\right)\right\|^2.  \notag
	\end{align}
	Let $\mathbb{E}_{M,B}$ denote the expectations over the randomness of the sampling trajectories $\{\tilde{\tau}_{i,j}\}_{j=1}^{M}$ in outer loop and $\{\tau_{i,b}\}_{b=1}^{B}$ in inner loop.
    Taking the expectations of the above inequality, we obtain that
	\begin{align} &\mathbb{E}_{M,B}\left[\left\|\boldsymbol{v}_{i,k+1}^{s+1}-\boldsymbol{v}_{i,k}^{s+1}\right\|^2\right] \notag \\
	\leq& \frac{2}{B}\sum_{b=1}^{B}\mathbb{E}_{M,B}\left[\left\|\omega\left(\tau_{i,b}\left|{\boldsymbol{\theta}}_{i,k+1}^{s+1}\right.,\tilde{\boldsymbol{\theta}}_i^s\right)g_i\left(\tau_{i,b}\left|\tilde{\boldsymbol{\theta}}_i^s\right.\right)\right.\right. \notag \\	&\left.\left.-g_i\left(\tau_{i,b}\left|{\boldsymbol{\theta}}_{i,k+1}^{s+1}\right.\right)\right\|^2\right] \notag \\	&+\frac{2}{B}\sum_{b^{\prime}=1}^{B}\mathbb{E}_{M,B}\left[\left\|\omega\left(\tau_{i,b^{\prime}}\left|{\boldsymbol{\theta}}_{i,k}^{s+1}\right.,\tilde{\boldsymbol{\theta}}_i^s\right)g_i\left(\tau_{i,b^{\prime}}\left|\tilde{\boldsymbol{\theta}}_i^s\right.\right)\right.\right. \notag \\	&\left.\left.-g_i\left(\tau_{i,b^{\prime}}\left|{\boldsymbol{\theta}}_{i,k}^{s+1}\right.\right)\right\|^2\right],\label{eq29}
	\end{align}
	where we have used the inequality that $\left \|  {\bf{x}}_1 + {\bf{x}}_2 \cdots   {\bf{x}}_B\right \|^2 \le B \left(\left \| {\bf{x}}_1\right \|^2+\left \| {\bf{x}}_2\right \|^2+\cdots \left \| {\bf{x}}_B\right \|^2\right)$. Note that for any $k, s \ge 0$, we have
	\begin{align}
	&\mathbb{E}_{M,B}\left[ \left\| \omega \left( {\tau_{i,b}\left| {\boldsymbol{\theta}_{i,k+1}^{s+1}} \right.,\tilde{\boldsymbol{\theta}}_i^s} \right){g_i}\left( {\tau_{i,b}\left| {\tilde{\boldsymbol{\theta}}_i^s} \right.} \right) \right.\right. \notag \\
&\left.\left.- {g_i}\left( {\tau_{i,b}\left| {\boldsymbol{\theta}_{i,k+1}^{s+1}} \right.} \right) \right\|^2 \right] \notag \\
	=& \mathbb{E}_{M,B}\left[ \left\| \omega \left( {\tau_{i,b}\left| {\boldsymbol{\theta}_{i,k+1}^{s+1}} \right.,\tilde{\boldsymbol{\theta}}_i^s} \right){g_i}\left( {\tau_{i,b}\left| {\tilde{\boldsymbol{\theta}}_i^s} \right.} \right) - {g_i}\left( {\tau_{i,b}\left| {\tilde{\boldsymbol{\theta}}_i^s} \right.} \right)\right.\right. \notag \\
	&\left.\left.+ {g_i}\left( {\tau_{i,b}\left| {\tilde{\boldsymbol{\theta}}_i^s} \right.} \right) - {g_i}\left( {\tau_{i,b}\left| {\boldsymbol{\theta}_{i,k+1}^{s+1}} \right.} \right) \right\|^2 \right] \notag \\
	\leq& 2\mathbb{E}_{M,B}\left[ {{{\left\| {\left( {\omega \left( {\tau_{i,b}\left| {\boldsymbol{\theta}_{i,k+1}^{s+1}} \right.,\tilde{\boldsymbol{\theta}}_i^s} \right) - 1} \right){g_i}\left( {\tau_{i,b}\left| {\tilde{\boldsymbol{\theta}}_i^s} \right.} \right)} \right\|}^2}} \right] \notag \\
	&+ 2\mathbb{E}_{M,B}\left[ {{{\left\| {{g_i}\left( {\tau_{i,b}\left| {\tilde{\boldsymbol{\theta}}_i^s} \right.} \right) - {g_i}\left( {\tau_{i,b}\left| {\boldsymbol{\theta}_{i,k+1}^{s+1}} \right.} \right)} \right\|}^2}} \right] \notag \\
	\stackrel{(a)}{\le}& 2C_g^2\mathbb{E}_{M,B}\left[ {{{\left\| { {\omega \left( {\tau_{i,b}\left| {\boldsymbol{\theta}_{i,k+1}^{s+1}} \right.,\tilde{\boldsymbol{\theta}}_i^s} \right) - 1} } \right\|}^2}} \right] \notag \\
	&+ 2L_g^2{{{\left\| {\tilde{\boldsymbol{\theta}}_i^s - \boldsymbol{\theta}_{i,k+1}^{s+1}} \right\|}^2}}, \label{eq30}
	\end{align}
	where in (a) we use $\left\|g\left(\tau \left| \boldsymbol{\theta}\right.\right)\right\| \leq C_g$ and $g\left(\tau \left|\boldsymbol{\theta}\right.\right)$ is $L_g$-Lipschitz continuous in Lemma \ref{Lem1}.
	
	By using $\mathbb{E}\left[\omega (\cdot)\right]=1$ (see Lemma C.1 in \cite{xu2020sample}) and Lemma \ref{Lem2}, we obtain that
	\begin{equation}\label{eq31}
	\begin{aligned}
	&\mathbb{E}_{M,B}\left[ {{{\left\| { {\omega \left( {\tau_{i,b}\left| {\boldsymbol{\theta}_{i,k+1}^{s+1}} \right.,\tilde{\boldsymbol{\theta}}_i^s} \right) - 1} } \right\|}^2}} \right]\\
	{=}&{\rm{Var}}\left[ {\omega \left( {\tau_{i,b}\left| {\boldsymbol{\theta}_{i,k+1}^{s+1}} \right.,\tilde{\boldsymbol{\theta}}_i^s} \right)} \right]
	{\leq}{C_\omega }{\left\| {\boldsymbol{\theta}_{i,k+1}^{s+1} - \tilde{\boldsymbol{\theta}}_i^s} \right\|^2}.
	\end{aligned}
	\end{equation}
	Substituting (\ref{eq31}) into (\ref{eq30}) leads to
	\begin{equation}\label{eq32}
	\begin{aligned}
	&\mathbb{E}_{M,B}\left[ \left\| \omega \left( {\tau_{i,b}\left| {\boldsymbol{\theta}_{i,k+1}^{s+1}} \right.,\tilde{\boldsymbol{\theta}}_i^s} \right){g_i}\left( {\tau_{i,b}\left| {\tilde{\boldsymbol{\theta}}_i^s} \right.} \right) \right.\right. \notag \\
&\left.\left.- {g_i}\left( {\tau_{i,b}\left| {\boldsymbol{\theta}_{i,k+1}^{s+1}} \right.} \right) \right\|^2 \right] \\
	\leq& 2C_g^2{C_\omega }{\left\| {\boldsymbol{\theta}_{i,k+1}^{s+1} - \tilde{\boldsymbol{\theta}}_i^s} \right\|^2} + 2L_g^2 {{{\left\| {\tilde{\boldsymbol{\theta}}_i^s - \boldsymbol{\theta}_{i,k+1}^{s+1}} \right\|}^2}}\\
	\leq& 2\left(C_g^2{C_\omega }+{L_g^2}\right)\left\|\boldsymbol{\theta}_{i,k+1}^{s+1}-\tilde{\boldsymbol{\theta}}_i^s\right\|^2.
	\end{aligned}
	\end{equation}
	This together with (\ref{eq29}) and the definition  $2\left(C_g^2C_\omega+L_g^2\right)=\Psi$ gives that
    \begin{equation*}
    \begin{aligned}
	&\mathbb{E}_{M,B}\left[ \left\|\boldsymbol{v}_{i,k+1}^{s+1}-\boldsymbol{v}_{i,k}^{s+1}\right\|^2 \right] \\
	\leq& 2\Psi{\left\| {\boldsymbol{\theta}_{i,k+1}^{s+1} - \tilde{\boldsymbol{\theta}}^s} \right\|^2}
	+ 2\Psi{\left\| {\boldsymbol{\theta}_{i,k}^{s+1} - \tilde{\boldsymbol{\theta}}^s} \right\|^2}.
	\end{aligned}
    \end{equation*}
    Summing up the above inequality over $i$ from $1$ to $n$ and taking the expectations, we have
	\begin{equation}\label{eq33}
	\begin{aligned}
	&\mathbb{E}\left[ \left\|\boldsymbol{v}_{k+1}^{s+1}-\boldsymbol{v}_{k}^{s+1}\right\|^2 \right] \\
	\leq& 2\Psi \mathbb{E}\left[ \left\| {\boldsymbol{\theta}_{k+1}^{s+1} - \tilde{\boldsymbol{\theta}}^s} \right\|^2 \right]
	+ 2\Psi \mathbb{E}\left[ \left\| {\boldsymbol{\theta}_{k}^{s+1} - \tilde{\boldsymbol{\theta}}^s} \right\|^2\right].
	\end{aligned}
	\end{equation}
	
	Recalling that $\tilde{\boldsymbol{\theta}}^s=\boldsymbol{\theta}_{0}^{s+1}$ in Line 2 of Algorithm \ref{Alg1},  {for any} $k, s \ge 0$ we have
	\begin{equation}\label{eq34}
	\begin{aligned}
	&\left\| \boldsymbol{\theta}_{k+1}^{s+1} - \tilde{\boldsymbol{\theta}}^s\right\|^2\\
	=&  \left\| \boldsymbol{\theta}_{k+1}^{s+1} - {\boldsymbol{1}_{n}}\bar{\boldsymbol{\theta}}_{k+1}^{s+1}+ {\boldsymbol{1}_{n}}\bar{\boldsymbol{\theta}}_{k+1}^{s+1}- {\boldsymbol{1}_{n}}\bar{\boldsymbol{\theta}}_{0}^{s+1}+ {\boldsymbol{1}_{n}}\bar{\boldsymbol{\theta}}_{0}^{s+1} - \boldsymbol{\theta}_{0}^{s+1}\right\|^2\\
	\leq &
	3\left\| \boldsymbol{\theta}_{k+1}^{s+1} - {\boldsymbol{1}_{n}}\bar{\boldsymbol{\theta}}_{k+1}^{s+1} \right\|^2 + 3n\left\| \bar{\boldsymbol{\theta}}_{k+1}^{s+1}-\bar{\boldsymbol{\theta}}_{0}^{s+1}\right\|^2\\
	&+3\left\| \boldsymbol{\theta}_{0}^{s+1} - {\boldsymbol{1}_{n}}\bar{\boldsymbol{\theta}}_{0}^{s+1} \right\|^2.
	\end{aligned}
	\end{equation}
    This combined with (\ref{eq22}) and (\ref{eq33}) produces
	\begin{equation}\label{eq35-2}
	\begin{aligned}
	&\mathbb{E}\left[ \left\|\boldsymbol{v}_{k+1}^{s+1}-\boldsymbol{v}_{k}^{s+1}\right\|^2 \right] \leq 18\Psi \mathbb{E}\left[\left\| \boldsymbol{\theta}_{k}^{s+1} - {\boldsymbol{1}_{n}}\bar{\boldsymbol{\theta}}_{k}^{s+1} \right\|^2 \right]\\
    &+ 12\alpha^2 \Psi \mathbb{E}\left[ \left\| \boldsymbol{y}_{k}^{s+1} - {\boldsymbol{1}_{n}}\bar{\boldsymbol{y}}_{k}^{s+1} \right\|^2 \right] \\
	&+ 6\Psi n \mathbb{E}\left[\left\|\bar{\boldsymbol{\theta}}_{k+1}^{s+1}-\bar{\boldsymbol{\theta}}_{0}^{s+1} \right\|^2 \right]
	+ 6\Psi n \mathbb{E}\left[\left\|\bar{\boldsymbol{\theta}}_{k}^{s+1}-\bar{\boldsymbol{\theta}}_{0}^{s+1} \right\|^2 \right] \\
	&+ 12\Psi \mathbb{E}\left[\left\| \boldsymbol{\theta}_{0}^{s+1} - {\boldsymbol{1}_{n}}\bar{\boldsymbol{\theta}}_{0}^{s+1} \right\|^2 \right].
	\end{aligned}
	\end{equation}
	By applying (\ref{eq35-2}) to (\ref{eq28}), we obtain the refined bound of gradient tracking error as follows	
	\begin{equation}\label{eq36}
	\begin{aligned}
	&\mathbb{E}\left[\left\| {\boldsymbol{y}_{k+1}^{s+1}} - \boldsymbol{1}_{n}{ \bar{\boldsymbol{y}}_{k+1}^{s+1} } \right\|^2\right]
	\le \frac{36\Psi }{1-\sigma^2}\mathbb{E}\left[  \left\| {{\boldsymbol{\theta }}_{k}^{s+1} - \boldsymbol{1}_{n}\bar{\boldsymbol{ \theta }}_{k}^{s+1}} \right\|^2\right]\\
	&+ \left(\frac{1+\sigma^2}{2} + \frac{24\alpha ^2 \Psi }{1-\sigma^2}\right)\mathbb{E}\left[\left\| {\boldsymbol{y}}_{k}^{s+1} -\boldsymbol{1}_{n}\bar{\boldsymbol{y}}_{k}^{s+1} \right\|^2\right]\\
	&+ \frac{12\Psi n}{1-\sigma^2}\mathbb{E}\left[  \left\| \bar{\boldsymbol{ \theta }}_{k+1}^{s+1} - \bar{\boldsymbol{ \theta }}_{0}^{s+1}\right\|^2\right]
    + \frac{12 \Psi n}{1-\sigma^2}\mathbb{E}\left[  \left\| \bar{\boldsymbol{ \theta }}_{k}^{s+1} - \bar{\boldsymbol{ \theta }}_{0}^{s+1}\right\|^2\right]\\
	&+ \frac{24\Psi }{1-\sigma^2}\mathbb{E}\left[  \left\| {{\boldsymbol{\theta }}_{0}^{s+1} - \boldsymbol{1}_{n}\bar{\boldsymbol{ \theta }}_{0}^{s+1}} \right\|^2\right].
	\end{aligned}
	\end{equation}
	If $ 0 < \alpha \le \frac{1-\sigma^2}{4\sqrt{6\Psi}}$, then we have $\frac{1+\sigma^2}{2} + \frac{\alpha ^2 \Psi }{1-\sigma^2} \le \frac{3+\sigma^2}{4}$. This combined with (\ref{eq36}) proves the lemma.
\end{proof}

In the following proposition, we write (\ref{eq23}) and (\ref{eq26}) jointly as a linear matrix inequality.
\begin{Pro}\label{Pro1}
	Suppose Assumptions \ref{Asu1}, \ref{Asu3}, and \ref{Asu4} hold. Consider Algorithm \ref{Alg1} where the step size $\alpha$ follows $ 0< \alpha < \frac{1-\sigma^2}{4\sqrt{6 \Psi}}$, then the following linear matrix inequality holds  {for any} $k,s \ge 0 $
	\begin{equation}\label{eq37}
	\boldsymbol{u}_{k+1}^{s+1} \leq \boldsymbol{G}\boldsymbol{u}_{k}^{s+1}+\boldsymbol{b}_{k}^{s+1},
	\end{equation}
	where $\boldsymbol{u}_{k}^{s+1}$, $\ \boldsymbol{b}_{k}^{s+1}  \in \mathbb{R}^2$, and $\boldsymbol{G} $ are defined as
	\begin{equation*}
	\boldsymbol{u}_{k}^{s+1}=\begin{bmatrix}
	\mathbb{E}\left[  \left\| {{\boldsymbol{\theta }}_{k}^{s+1} - \boldsymbol{1}_{n}\bar{\boldsymbol{ \theta }}_{k}^{s+1}} \right\|^2\right]\\
	\mathbb{E}\left[  \left\| {{\boldsymbol{y}}_{k}^{s+1} - \boldsymbol{1}_{n}\bar{\boldsymbol{y}}_{k}^{s+1}} \right\|^2\right]
	\end{bmatrix}, 	
	\end{equation*}
	\begin{equation}\label{eq_G_b}
	\boldsymbol{G} =\begin{bmatrix}
	\frac{1+\sigma^2}{2}       & \frac{2\alpha^2}{1-\sigma^2}\\
	\frac{36 \Psi }{1-\sigma^2} & \frac{3+\sigma^2}{4}
	\end{bmatrix},     { {\rm~and ~}
		\boldsymbol{b}_{k}^{s+1} =\begin{bmatrix}
		0\\
		\tilde{\boldsymbol{b}}_{k,2}^{s+1}
		\end{bmatrix} },
	\end{equation}
	with
	\begin{align}
	\tilde{\boldsymbol{b}}_{k,2}^{s+1}&=
	\frac{12\Psi n}{1-\sigma^2}\mathbb{E}\left[  \left\| \bar{\boldsymbol{ \theta }}_{k+1}^{s+1} - \bar{\boldsymbol{ \theta }}_{0}^{s+1}\right\|^2\right] \notag \\
	&+ \frac{12\Psi n}{1-\sigma^2}\mathbb{E}\left[  \left\| \bar{\boldsymbol{ \theta }}_{k}^{s+1} - \bar{\boldsymbol{ \theta }}_{0}^{s+1}\right\|^2\right] \notag \\
	&+ \frac{24\Psi }{1-\sigma^2}\mathbb{E}\left[  \left\| {{\boldsymbol{\theta }}_{0}^{s+1} - \boldsymbol{1}_{n}\bar{\boldsymbol{ \theta }}_{0}^{s+1}} \right\|^2\right].
    \label{eq_bk2}
	\end{align}

\end{Pro}

The inequality (\ref{eq37}) in Proposition \ref{Pro1} recursively leads to
\begin{equation}\label{eq39}
\begin{aligned}
\boldsymbol{u}_{k}^{s+1} \leq \boldsymbol{G}^{k}\boldsymbol{u}_{0}^{s+1}+\sum_{r=0}^{k-1}\boldsymbol{G}^{k-r-1}\boldsymbol{b}_{r}^{s+1}.
\end{aligned}
\end{equation}

The subsequent lemma gives an upper bound for $\boldsymbol{G}^{k}$.

\begin{Lem}\label{Lem8}
	If $0 < \alpha < \frac{\left(1-\sigma^2\right)^2}{24\sqrt{2\Psi}}$, we have
	\begin{equation}\label{eq40}
	\begin{aligned}
	\boldsymbol{G}^{k}
	\leq& \begin{bmatrix}
	\lambda^{k} & \boldsymbol{G}_{12} k \lambda^{k-1}\\
	\boldsymbol{G}_{21} k \lambda^{k-1} & \lambda^{k}+\frac{\left(\boldsymbol{G}_{22}-\boldsymbol{G}_{11}\right) k \lambda^{k-1} }{2}
	\end{bmatrix},
	\end{aligned}
	\end{equation}
	where $\lambda=\frac{3+\sigma^2}{4}+\frac{6\alpha\sqrt{2\Psi}}{1-\sigma^2}$.
\end{Lem}

\begin{proof}
	Consider the eigen-decomposition $\boldsymbol{G}=\boldsymbol{T}\Lambda\boldsymbol{T}^{-1}$ with $\Lambda=\text{diag}\left(\lambda_{1}, \lambda_{2}\right)$, where $\lambda_{1}$ and $\lambda_{2}$ are the two eigenvalues of $\boldsymbol{G}$ and $\lambda_{1} < \lambda_{2}$. Define $\Omega=\sqrt{\left(\boldsymbol{G}_{11}- \boldsymbol{G}_{22}\right)^2+4\boldsymbol{G}_{12}\boldsymbol{G}_{21}}$. With some tedious calculation, we obtain that
	\begin{align*}
    &\lambda_{1}=\frac{5+3\sigma^2}{8}-\frac{\sqrt{\left(1-\sigma^2\right)^4
    +4608\alpha^2\Psi}}{8\left(1-\sigma^2\right)},\\	&\lambda_{2}=\frac{5+3\sigma^2}{8}+\frac{\sqrt{\left(1-\sigma^2\right)^4 + 4608\alpha^2\Psi}}{8\left(1-\sigma^2\right)},\\
	&\boldsymbol{T}=\begin{bmatrix}
	\frac{\boldsymbol{G}_{11}-\boldsymbol{G}_{22}-\Omega}{2\boldsymbol{G}_{21}}
    & \frac{\boldsymbol{G}_{11}-\boldsymbol{G}_{22}+\Omega}{2\boldsymbol{G}_{21}}\\
	1 & 1
	\end{bmatrix},  \\
	&\boldsymbol{T}^{-1}=\begin{bmatrix}
	-\frac{\boldsymbol{G}_{21}}{\Omega}
    & \frac{\boldsymbol{G}_{11}-\boldsymbol{G}_{22}+\Omega}{2\Omega}\\
	\frac{\boldsymbol{G}_{21}}{\Omega}
    & \frac{\boldsymbol{G}_{22}-\boldsymbol{G}_{11}+\Omega}{2\Omega}
	\end{bmatrix}.
	\end{align*}
	Hence, for any $k \geq 0$, we have
	\begin{equation}\label{eq41}
	\begin{aligned}
	&\boldsymbol{G}^{k}=\boldsymbol{T}\Lambda^{k}\boldsymbol{T}^{-1} \\
	\leq& \begin{bmatrix}	\frac{\lambda_{1}^{k}+\lambda_{2}^{k}}{2}+\frac{\left(\boldsymbol{G}_{11}-\boldsymbol{G}_{22}\right)\left(\lambda_{2}^{k}-\lambda_{1}^{k}\right)}{2\Omega}       & \frac{\boldsymbol{G}_{12}\left(\lambda_{2}^{k}-\lambda_{1}^{k}\right)}{\Omega}\\
	\frac{\boldsymbol{G}_{21}\left(\lambda_{2}^{k}-\lambda_{1}^{k}\right)}{\Omega} & \frac{\lambda_{1}^{k}+\lambda_{2}^{k}}{2}+\frac{\left(\boldsymbol{G}_{11}-\boldsymbol{G}_{22}\right)\left(\lambda_{1}^{k}-\lambda_{2}^{k}\right)}{2\Omega}
	\end{bmatrix}
	\end{aligned}
	\end{equation}
	By the basic inequality $\sqrt{a^2+b^2} \le a +b, \forall a,b > 0$, we have $\lambda_{2} \le \frac{5+3\sigma^2}{8}+\frac{1-\sigma^2}{8}+\frac{6\alpha\sqrt{2\Psi}}{1-\sigma^2} = \lambda$.
	Since $0 < \alpha < \frac{\left(1-\sigma^2\right)^2}{24\sqrt{2\Psi}}$, we have $\lambda < 1 $. Then we can obtain $0 < \lambda_{1} < \lambda_{2} \le \lambda < 1$. By substituting $\lambda_{2}^{k}-\lambda_{1}^{k} = \left(\lambda_{2}-\lambda_{1}\right)\sum_{l=0}^{k-1}\lambda_{2}^{l}\lambda_{1}^{k-1-l}=\Omega k \lambda^{k-1}$ into (\ref{eq41}) leads to (\ref{eq40}).
\end{proof}

\section{Proofs of Section IV-D}
\subsection{Proof of Lemma \ref{Lem9}}
The proof builds upon the following lemma from \cite{zhang2019decentralized}.
\begin{Lem}\label{Lem5}
	Let $\{s_k\}$ be a non-negative sequence, $\theta$ be a constant in $(0,1)$, and $a_{k}=\sum_{r=1}^{k}s_r (k-r) \theta^{k-r-1}$. Then
	\begin{equation*}
	\sum_{k=1}^{K}a_{k} \le \frac{1}{\left(1-\theta\right)^2}\sum_{k=1}^{K}s_{k}, \quad \forall \theta \in (0,1), k \in \mathbb{N}.
	\end{equation*}
\end{Lem}

\begin{proof}
	Let $\boldsymbol{G}^{k} \boldsymbol{u}^{s+1}_{0}\left[1,;\right]$ be the first row of $\boldsymbol{G}^{k} \boldsymbol{u}^{s+1}_{0}$.
	By using (\ref{eq_G_b}) and (\ref{eq40}), we obtain that
	\begin{equation}\label{eq43}
	\begin{aligned}
	\boldsymbol{G}^{k} \boldsymbol{u}^{s+1}_{0}\left[1,;\right]
	&\le \lambda^{k}\mathbb{E}\left[\left\|\boldsymbol{\theta}_{0}^{s+1}-{\boldsymbol{1}_{n}}\bar{\boldsymbol{\theta}}_{0}^{s+1}\right\|^2\right] \\
	+& \boldsymbol{G}_{12}k\lambda^{k-1} \mathbb{E}\left[\left\|\boldsymbol{y}_{0}^{s+1}-{\boldsymbol{1}_{n}}\bar{\boldsymbol{y}}_{0}^{s+1}\right\|^2\right].
	\end{aligned}
	\end{equation}
	Let $\boldsymbol{G}^{k-r-1}\boldsymbol{b}_{r}^{s+1}[1]$ be the first element of $\boldsymbol{G}^{k-r-1}\boldsymbol{b}_{r}^{s+1}$. Similarly, it follows from (\ref{eq_G_b}) and (\ref{eq40}) that
	\begin{align} \label{eq44}
	&\boldsymbol{G}^{k-r-1}\boldsymbol{b}_{r}^{s+1}[1] \notag \\
	\leq& \boldsymbol{G}_{12} (k-r-1)\lambda^{k-r-2} {\tilde{\boldsymbol{b}}_{r,2}^{s+1}}.
	\end{align}
	This combined with (\ref{eq39}),  { (\ref{eq43}), the definition of $\boldsymbol{G}$ in (\ref{eq_G_b}), and (\ref{eq_bk2})} produces that
	\begin{align} &\mathbb{E}\left[\left\|\boldsymbol{\theta}_{k}^{s+1}-\boldsymbol{1}_{n}\bar{\boldsymbol{\theta}}_{k}^{s+1}\right\|^{2}\right] \le \lambda^{k}\mathbb{E}\left[\left\|\boldsymbol{\theta}_{0}^{s+1}-{\boldsymbol{1}_{n}}\bar{\boldsymbol{\theta}}_{0}^{s+1}\right\|^2\right] \notag \\
	&+\frac{2 \alpha^2 k\lambda^{k-1} }{1-\sigma^2} \mathbb{E}\left[\left\|\boldsymbol{y}_{0}^{s+1}-{\boldsymbol{1}_{n}}\bar{\boldsymbol{y}}_{0}^{s+1}\right\|^2\right] \notag \\
	&+ \sum_{r=0}^{k-1}\frac{2 \alpha^2 (k-r-1)\lambda^{k-r-2}}{1-\sigma^2}   \left(\frac{12\Psi n}{1-\sigma^2}\mathbb{E}\left[  \left\| \bar{\boldsymbol{ \theta }}_{r+1}^{s+1} - \bar{\boldsymbol{ \theta }}_{0}^{s+1}\right\|^2\right] \right.& \notag \\
	&\left. \qquad + \frac{12\Psi n}{1-\sigma^2}\mathbb{E}\left[  \left\| \bar{\boldsymbol{ \theta }}_{r}^{s+1} - \bar{\boldsymbol{ \theta }}_{0}^{s+1}\right\|^2\right] \right.& \notag \\
	&\left. \qquad + \frac{24\Psi }{1-\sigma^2}\mathbb{E}\left[  \left\| {{\boldsymbol{\theta }}_{0}^{s+1} - \boldsymbol{1}_{n}\bar{\boldsymbol{ \theta }}_{0}^{s+1}} \right\|^2\right]\right). \label{eq45}
	\end{align}
	By using $c_1$ and $c_2$ in (\ref{addeq2}) and  defining
		\begin{equation}\label{def-tdc}
		\begin{aligned}
		 {\tilde{c}_0(k)}=&\lambda^{k}\mathbb{E}[\left\|\boldsymbol{\theta}_{0}^{s+1}-{\boldsymbol{1}_{n}}\bar{\boldsymbol{\theta}}_{0}^{s+1}\right\|^2]
		\\& +\frac{2 \alpha^2k\lambda^{k-1}}{1-\sigma^2} \mathbb{E}[\left\|\boldsymbol{y}_{0}^{s+1}-{\boldsymbol{1}_{n}}\bar{\boldsymbol{y}}_{0}^{s+1}\right\|^2],
		\end{aligned}
		\end{equation}
	we have
	\begin{equation}\label{eq46}
	\begin{aligned}
	&\mathbb{E}\left[\left\|\boldsymbol{\theta}_{k}^{s+1}-\boldsymbol{1}_{n}\bar{\boldsymbol{\theta}}_{k}^{s+1}\right\|^{2}\right]
	\le  {\tilde{c}_0(k)}\\
	&+ \sum_{r=0}^{k-1}(k-r-1)\lambda^{k-r-2} \alpha^2\left(c_1\mathbb{E}\left[  \left\| \bar{\boldsymbol{ \theta }}_{r+1}^{s+1} - \bar{\boldsymbol{ \theta }}_{0}^{s+1}\right\|^2\right] \right.\\
	&\left.+  c_1\mathbb{E}\left[  \left\| \bar{\boldsymbol{ \theta }}_{r}^{s+1} - \bar{\boldsymbol{ \theta }}_{0}^{s+1}\right\|^2\right]+ c_2\mathbb{E}\left[  \left\| {{\boldsymbol{\theta }}_{0}^{s+1} - \boldsymbol{1}_{n}\bar{\boldsymbol{ \theta }}_{0}^{s+1}} \right\|^2\right]\right).
	\end{aligned}
	\end{equation}	
	Define $a_{k}= \sum_{r=0}^{k-1} s_{r} \left(k-r-1\right)\lambda^{k-r-2}$ with
	\begin{equation}\label{def-sk}
	\begin{aligned}
	s_{r} =& \alpha^2 \left(c_1\mathbb{E}\left[  \left\| \bar{\boldsymbol{ \theta }}_{r+1}^{s+1} - \bar{\boldsymbol{ \theta }}_{0}^{s+1}\right\|^2\right]
	+ c_1\mathbb{E}\left[  \left\| \bar{\boldsymbol{ \theta }}_{r}^{s+1} - \bar{\boldsymbol{ \theta }}_{0}^{s+1}\right\|^2\right]\right.\\
	&\left.+ c_2 \mathbb{E}\left[\left\|\boldsymbol{\theta}_{0}^{s+1}-{\boldsymbol{1}_{n}}\bar{\boldsymbol{\theta}}_{0}^{s+1}\right\|^2\right]\right).
	\end{aligned}
	\end{equation}
	Thus, we can write (\ref{eq46}) as follows
	\begin{equation}\label{recu-theta}
	\mathbb{E}\left[\left\|\boldsymbol{\theta}_{k}^{s+1}-\boldsymbol{1}_{n}
	\bar{\boldsymbol{\theta}}_{k}^{s+1}\right\|^{2}\right] \le  {\tilde{c}_0(k)} + a_k.
	\end{equation}
	
	 By using $\sum_{r=0}^{\infty} \lambda ^r = \frac{1}{1-\lambda}$, $\sum_{r=1}^{\infty} r \lambda^{r-1} = \frac{1}{\left(1-\lambda\right)^2}$, we obtain from  (\ref{addeq1})  and \eqref{def-tdc} that
	\begin{equation}\label{bd-c0k}
	\begin{aligned}
	&\sum_{k=0}^{K-1}  \tilde{c}_0(k) \le    \frac{1}{1-\lambda}\mathbb{E}\left[\left\|\boldsymbol{\theta}_{0}^{s+1}-{\boldsymbol{1}_{n}}\bar{\boldsymbol{\theta}}_{0}^{s+1}\right\|^2\right]\\
	&+\frac{2\alpha^2}{\left(1-\sigma^2\right)\left(1-\lambda\right)^2}		\mathbb{E}\left[\left\|\boldsymbol{y}_{0}^{s+1}-{\boldsymbol{1}_{n}}\bar{\boldsymbol{y}}_{0}^{s+1}\right\|^2\right]
		 =c_0(s).
	\end{aligned}
	\end{equation}
	 By summing up the inequality  \eqref{recu-theta} over $k=0, \ldots, K-1$,  using   \eqref{bd-c0k} and Lemma \ref{Lem5}, we get
	\begin{equation}\label{eq47}
	\begin{aligned}
	&\sum_{k=0}^{K-1} \mathbb{E}\left[\left\|\boldsymbol{\theta}_{k}^{s+1}-\boldsymbol{1}_{n}\bar{\boldsymbol{\theta}}_{k}^{s+1}\right\|^{2}\right] \\\le& \sum_{k=0}^{K-1}  {\tilde{c}_0(k)} + \sum_{k=0}^{K-1} a_k
	\le   c_0(s)+  \frac{1}{\left(1-\lambda\right)^2}\sum_{k=0}^{K-1} s_k .
	\end{aligned}
	\end{equation}
	This combined with \eqref{def-sk} proves the lemma.
\end{proof}

\subsection{Proof of Lemma \ref{Lem10}}
\begin{proof}
	Since $\boldsymbol{y}_{0}^{s+2} = \boldsymbol{y}_{K}^{s+1}$ in Line 11 of Algorithm \ref{Alg1}, we have $\boldsymbol{u}_{0}^{s+2} = \boldsymbol{u}_{K}^{s+1}$. Then by (\ref{eq39}), we can obtain
	\begin{align}\label{eq51}
	\boldsymbol{u}_{0}^{s+2} = \boldsymbol{u}_{K}^{s+1} \leq \boldsymbol{G}^{K}\boldsymbol{u}_{0}^{s+1}+\sum_{r=0}^{K-1}\boldsymbol{G}^{K-1-r}\boldsymbol{b}_{r}^{s+1}.
	\end{align}
	Applying the above inequality over $s$ leads to
	\begin{align}\label{eq52}
	\boldsymbol{u}_{0}^{s+1} \le \boldsymbol{G}^{sK}\boldsymbol{u}_{0}^{1}+\sum_{l=0}^{s-1}\sum_{r=0}^{K-1}\boldsymbol{G}^{(s-l)K-1-r}\boldsymbol{b}_{r}^{l+1}.
	\end{align}
	
	Similarly to (\ref{eq43}) in the proof of Lemma \ref{Lem9}, let $\boldsymbol{G}^{sK}\boldsymbol{u}_{0}^{1}\left[1,;\right]$ be the first row of $\boldsymbol{G}^{sK}\boldsymbol{u}_{0}^{1}$. By using (\ref{eq40}), we have
	\begin{align}\label{eq53}
	\boldsymbol{G}^{sK}\boldsymbol{u}_{0}^{1}\left[1,;\right]
    \le& \lambda^{sK}\mathbb{E}\left[\left\|\boldsymbol{\theta}_{0}^{1}
    -{\boldsymbol{1}_{n}}\bar{\boldsymbol{\theta}}_{0}^{1}\right\|^2\right] \notag\\
	&+  {\boldsymbol{G}_{12}sK \lambda^{sK-1}} \mathbb{E}\left[\left\|\boldsymbol{y}_{0}^{1}
    -{\boldsymbol{1}_{n}}\bar{\boldsymbol{y}}_{0}^{1}\right\|^2\right].
	\end{align}
	
	Similarly to (\ref{eq44}), let $\boldsymbol{G}^{(s-l)K-1-r}\boldsymbol{b}_{r}^{l+1}[1]$ be the first element of $\boldsymbol{G}^{(s-l)K-1-r}\boldsymbol{b}_{r}^{l+1}$, then we get
	\begin{equation}\label{eq54}
	\begin{aligned}
	&\boldsymbol{G}^{(s-l)K-1-r}\boldsymbol{b}_{r}^{l+1}[1] \\
	\le&  {\boldsymbol{G}_{12} \left((s-l)K-r-1\right)\lambda ^{(s-l)K-r-2} \tilde{\boldsymbol{b}}_{r,2}^{l+1}}.
	\end{aligned}
	\end{equation}
	This together with (\ref{eq52}),  (\ref{eq53}), $\boldsymbol{G}$ in (\ref{eq_G_b}), $c_1$ and $c_2$ in (\ref{addeq2}), and  (\ref{eq_bk2}) yields
	\begin{align}
	&\mathbb{E}\left[\left\| {{\boldsymbol{\theta }}_{0}^{s+1} - \boldsymbol{1}_{n}\bar{\boldsymbol{ \theta }}_{0}^{s+1}} \right\|^{2}\right] \notag\\
	\le & \lambda^{sK}\mathbb{E}\left[\left\|\boldsymbol{\theta}_{0}^{1}-{\boldsymbol{1}_{n}}\bar{\boldsymbol{\theta}}_{0}^{1}\right\|^2\right] + \frac{2 \alpha^2 sK \lambda^{sK-1}}{1-\sigma^2} \mathbb{E}\left[\left\|\boldsymbol{y}_{0}^{1}-{\boldsymbol{1}_{n}}\bar{\boldsymbol{y}}_{0}^{1}\right\|^2\right] \notag\\
	+&  \sum_{l=0}^{s-1}\sum_{r=0}^{K-1}\left((s-l)K-r-1\right)\lambda ^{(s-l)K-r-2}\alpha^2 \notag\\
	& \times \left(c_1\mathbb{E}\left[  \left\| \bar{\boldsymbol{ \theta }}_{r+1}^{l+1} - \bar{\boldsymbol{ \theta }}_{0}^{l+1}\right\|^2\right]
	+ c_1\mathbb{E}\left[  \left\| \bar{\boldsymbol{ \theta }}_{r}^{l+1} - \bar{\boldsymbol{ \theta }}_{0}^{l+1}\right\|^2\right] \right. & \notag\\
	& \left. \qquad
	+ c_2\mathbb{E}\left[  \left\| {{\boldsymbol{\theta }}_{0}^{l+1} - \boldsymbol{1}_{n}\bar{\boldsymbol{ \theta }}_{0}^{l+1}} \right\|^2\right]\right). \label{eq55}
	\end{align}
	Similarly to (\ref{bd-c0k}) and (\ref{eq47}), we sum up the inequality (\ref{eq55}) over $s=0, \cdots, S-1$. By using Lemma \ref{Lem5}, we obtain the upper bound of $\sum_{s=0}^{S-1}\mathbb{E}\left[\left\| {{\boldsymbol{\theta }}_{0}^{s+1} - \boldsymbol{1}_{n}\bar{\boldsymbol{ \theta }}_{0}^{s+1}} \right\|^{2}\right]$ as follows

	\begin{align}
	& \sum_{s=0}^{S-1}\mathbb{E}\left[\left\| {{\boldsymbol{\theta }}_{0}^{s+1}
    - \boldsymbol{1}_{n}\bar{\boldsymbol{ \theta }}_{0}^{s+1}} \right\|^{2}\right]
	\le  \frac{1}{1-\lambda}\mathbb{E} \left[\left\|\boldsymbol{\theta}_{0}^{1}
    -{\boldsymbol{1}_{n}}\bar{\boldsymbol{\theta}}_{0}^{1}\right\|^2\right] \notag \\
	&+ \frac{2 \alpha^2}{\left(1-\sigma^2\right)\left(1-\lambda\right)^2}
    \mathbb{E}\left[\left\|\boldsymbol{y}_{0}^{1}
    -{\boldsymbol{1}_{n}}\bar{\boldsymbol{y}}_{0}^{1}\right\|^2\right]  \notag\\
	&  + \frac{\alpha^2}{\left(1-\lambda\right)^2}\sum_{s=0}^{S-1}\sum_{k=0}^{K-1} c_1\mathbb{E}
    \left[  \left\| \bar{\boldsymbol{ \theta }}_{k+1}^{s+1}
    - \bar{\boldsymbol{ \theta }}_{0}^{s+1}\right\|^2\right]   \notag\\
	&+ \frac{\alpha^2}{\left(1-\lambda\right)^2}\sum_{s=0}^{S-1}\sum_{k=0}^{K-1}
    c_1 \mathbb{E} \left[  \left\| \bar{\boldsymbol{ \theta }}_{k}^{s+1} - \bar{\boldsymbol{ \theta }}_{0}^{s+1}\right\|^2\right]  \notag\\
	& + \frac{\alpha^2 K}{\left(1-\lambda\right)^2}\sum_{s=0}^{S-1}
    c_2 \mathbb{E}\left[\left\| {{\boldsymbol{\theta }}_{0}^{s+1} - \boldsymbol{1}_{n}\bar{\boldsymbol{ \theta }}_{0}^{s+1}} \right\|^{2}\right]. \label{eq56}
	\end{align}
	Suppose the upper bound of $\lambda$ is $\lambda \le \frac{7+\sigma^2}{8}$. If $0 < \alpha < \frac{\left(1-\sigma^2\right)^2}{32\sqrt{3 \Psi K}}$, we have $\left(1-\lambda\right)^2\left(1-\sigma^2\right)^2-48\alpha^2\Psi K > 0$.
	We regroup the terms of (\ref{eq56}) and use the definition of $\Xi$ in (\ref{eq50a}) to obtain (\ref{eq48}).
	
	
	Next, we bound the accumulated gradient tracking error in the outer loop.
	Similarly to (\ref{eq43}), let $\boldsymbol{G}^{sK}\boldsymbol{u}_{0}^{1}\left[2,;\right]$ be the second row of $\boldsymbol{G}^{sK}\boldsymbol{u}_{0}^{1}$. By Lemma \ref{Lem8}, we have
	\begin{equation}\label{eq57}
	\begin{aligned}
	&\boldsymbol{G}^{sK}\boldsymbol{u}_{0}^{1}\left[2,;\right] \le \frac{36\Psi sK\lambda^{sK-1}}{1-\sigma^2 }\mathbb{E}\left[\left\|\boldsymbol{\theta}_{0}^{1}-{\boldsymbol{1}_{n}}\bar{\boldsymbol{\theta}}_{0}^{1}\right\|^2\right]\\
	&\quad +\left(\lambda^{sK}+\frac{\left(1-\sigma^2\right)sK \lambda^{sK-1}}{8}\right)\mathbb{E}\left[\left\|\boldsymbol{y}_{0}^{1}-{\boldsymbol{1}_{n}}\bar{\boldsymbol{y}}_{0}^{1}\right\|^2\right].
	\end{aligned}
	\end{equation}
	Similarly to (\ref{eq44}), let $\boldsymbol{G}^{(s-l)K-1-r}\boldsymbol{b}_{r}^{l+1}[2]$ be the second element of $\boldsymbol{G}^{(s-l)K-1-r}\boldsymbol{b}_{r}^{l+1}$, we get
	\begin{equation}\label{eq58}
	\begin{aligned}
	&\boldsymbol{G}^{(s-l)K-1-r}\boldsymbol{b}_{r}^{l+1}[2] \\
	\le&
\lambda^{(s-l)K-r-1}\left(1+\frac{\left(1-\sigma^2\right)\left((s-l)K-r-1\right)}{8\lambda}\right)\tilde{\boldsymbol{b}}_{r,2}^{l+1}.
	\end{aligned}
	\end{equation}
	This together with (\ref{eq52}) and (\ref{eq57}) produces that
	
	\begin{align}\label{eq59}
	&\mathbb{E}\left[\left\| {{\boldsymbol{y }}_{0}^{s+1} - \boldsymbol{1}_{n}\bar{\boldsymbol{y}}_{0}^{s+1}} \right\|^{2}\right]
	\le \frac{36\Psi sK\lambda^{sK-1}}{1-\sigma^2 }\mathbb{E}\left[\left\|\boldsymbol{\theta}_{0}^{1}-{\boldsymbol{1}_{n}}\bar{\boldsymbol{\theta}}_{0}^{1}\right\|^2\right] \notag \\
	& +\left(\lambda^{sK}+\frac{\left(1-\sigma^2\right)sK \lambda^{sK-1}}{8}\right)\mathbb{E}\left[\left\|\boldsymbol{y}_{0}^{1}-{\boldsymbol{1}_{n}}\bar{\boldsymbol{y}}_{0}^{1}\right\|^2\right] \notag\\
	& + \sum_{l=0}^{s-1}\sum_{r=0}^{K-1}\lambda ^{(s-l)K-r-1}\left(1+\frac{\left(1-\sigma^2\right)\left((s-l)K-r-1\right)}{8\lambda}\right) \notag\\
	&\times \left(\frac{12 \Psi n}{1-\sigma^2}\mathbb{E}\left[  \left\| \bar{\boldsymbol{ \theta }}_{r+1}^{l+1} - \bar{\boldsymbol{ \theta }}_{0}^{l+1}\right\|^2\right]
	+ \frac{12 \Psi n}{1-\sigma^2}\mathbb{E}\left[  \left\| \bar{\boldsymbol{ \theta }}_{r}^{l+1} - \bar{\boldsymbol{ \theta }}_{0}^{l+1}\right\|^2\right] \right.& \notag \\
	&\left. \qquad+\frac{24 \Psi}{1-\sigma^2}\mathbb{E}\left[  \left\| {{\boldsymbol{\theta }}_{0}^{l+1} - \boldsymbol{1}_{n}\bar{\boldsymbol{ \theta }}_{0}^{l+1}} \right\|^2\right]\right).
	\end{align}
	
	Similarly to (\ref{bd-c0k}) and (\ref{eq47}), we sum up the inequality (\ref{eq59}) over $s=0, \cdots, S-1$. Then using Lemma \ref{Lem5} yields
	\begin{align}\label{eq60}
	& \sum_{s=0}^{S-1}\mathbb{E}\left[\left\| {{\boldsymbol{y }}_{0}^{s+1} - \boldsymbol{1}_{n}\bar{\boldsymbol{y}}_{0}^{s+1}} \right\|^{2}\right] \notag \\
	\le &
	\frac{36\Psi}{\left(1-\lambda\right)^2\left(1-\sigma^2\right)}\mathbb{E} \left[\left\|\boldsymbol{\theta}_{0}^{1}-{\boldsymbol{1}_{n}}\bar{\boldsymbol{\theta}}_{0}^{1}\right\|^2\right] \notag \\
	+& \frac{8 \left(1-\lambda\right)+\left(1-\sigma^2\right)}{8\left(1-\lambda\right)^2}\mathbb{E} \left[\left\|\boldsymbol{y}_{0}^{1}-{\boldsymbol{1}_{n}}\bar{\boldsymbol{y}}_{0}^{1}\right\|^2\right]  \notag \\
	+& \frac{3 \Psi K \left(8 \left(1-\lambda\right)+\left(1-\sigma^2\right)\right)}{\left(1-\lambda\right)^2\left(1-\sigma^2\right)}\sum_{s=0}^{S-1}\mathbb{E}\left[\left\| {{\boldsymbol{\theta }}_{0}^{s+1} - \boldsymbol{1}_{n}\bar{\boldsymbol{ \theta }}_{0}^{s+1}} \right\|^{2}\right] \notag \\
	+&\frac{3 \Psi n \left(8 \left(1-\lambda\right)+\left(1-\sigma^2\right)\right)}{2 \left(1-\lambda\right)^2\left(1-\sigma^2\right)}\sum_{s=0}^{S-1}\sum_{k=0}^{K-1}\mathbb{E}\left[  \left\| \bar{\boldsymbol{ \theta }}_{k+1}^{s+1} - \bar{\boldsymbol{ \theta }}_{0}^{s+1}\right\|^2\right] \notag\\
	+& \frac{3 \Psi n \left(8 \left(1-\lambda\right)+\left(1-\sigma^2\right)\right)}{2 \left(1-\lambda\right)^2\left(1-\sigma^2\right)}\sum_{s=0}^{S-1}\sum_{k=0}^{K-1}\mathbb{E}\left[  \left\| \bar{\boldsymbol{ \theta }}_{k}^{s+1} - \bar{\boldsymbol{ \theta }}_{0}^{s+1}\right\|^2\right].
	\end{align}
	Substituting the upper bound of accumulated consensus error in (\ref{eq48}) into (\ref{eq60}) and using $\Phi$ in (\ref{eq50b}) yield (\ref{eq49}).
\end{proof}

\subsection{Proof of Lemma \ref{Lem11}}	

\begin{proof}
	Using $\boldsymbol{v}_{i,k}^{s+1}$ in (\ref{eq8}) and Line 4 of Algorithm 1, we have
	\begin{align}
	&\mathbb{E}_{M,B}\left[ {{{\left\| {\boldsymbol{v}_{i,k}^{s+1} - \nabla J_i\left( {\boldsymbol{\theta}_{i,k}^{s+1}} \right)} \right\|}^2}} \right] \notag \\
	=&\mathbb{E}_{M,B}\left[\left\| {\frac{1}{M}\sum\limits_{j = 1}^M {{g_i}\left( {{\tilde{\tau}_{i,j}}\left| {\tilde{\boldsymbol{\theta}}_i^s} \right.} \right)} }+\frac{1}{B}\sum_{b=1}^{B}\left(g_i\left(\tau_{i,b}\left|{\boldsymbol{\theta}}_{i,k}^{s+1}\right.\right)\right.\right.\right. \notag \\
	&\left.\left.\left.-\omega\left(\tau_{i,b}\left|{\boldsymbol{\theta}}_{i,k}^{s+1}\right.,\tilde{\boldsymbol{\theta}}_i^s\right)g_i\left(\tau_{i,b}\left|\tilde{\boldsymbol{\theta}}_i^s\right.\right)\right)
	-\nabla J_i\left( {\boldsymbol{\theta}_{i,k}^{s+1}} \right) \right\|^2\right] \notag \\
	=& {\mathbb{E}_{M,B}\left[\left\| \mathbf{e}_1-\mathbf{e}_2  \right\|^2\right],}\label{eq62}
	\end{align}
	 where
	\begin{align*}
	{\bf{e}_1} &= {\nabla}J_i\left(\tilde{\boldsymbol{\theta}}_i^s\right)-\nabla J_i\left( {\boldsymbol{\theta}_{i,k}^{s+1}} \right)	+\frac{1}{B}\sum_{b=1}^{B}\left(g_i\left(\tau_{i,b}\left|{\boldsymbol{\theta}}_{i,k}^{s+1}\right.\right) \right.\\
	&\left.\quad-\omega\left(\tau_{i,b}\left|{\boldsymbol{\theta}}_{i,k}^{s+1}\right.,\tilde{\boldsymbol{\theta}}_i^s\right)g_i\left(\tau_{i,b}\left|\tilde{\boldsymbol{\theta}}_i^s\right.\right)\right),\\
	{\bf{e}_2} &= {\nabla J_i\left( {\tilde{\boldsymbol{\theta}}_i^s} \right) - \frac{1}{M}\sum\limits_{j = 1}^M {{g_i}\left( {{\tilde{\tau}_{i,j}}\left| {\tilde{\boldsymbol{\theta}}_i^s} \right.} \right)} }.
	\end{align*}	
	Recall the unbiased estimator of $g\left(\tau_{j} | \boldsymbol{ \theta }\right)$ and the property in (\ref{eq10}), we have $\mathbb{E}\left[ {\bf{e}_1}\right]=0$ and $\mathbb{E}\left[ {\bf{e}_2}\right]=0$.
	Since the trajectories $\{{\tilde{\tau}_{i,j}}\}_{j=1}^{M}$ in the outer loop and $\{\tau_{i,b}\}_{b=1}^{B}$ in the inner loop are independent, $\bf{e}_1$ and $\bf{e}_2$ are independent. Therefore, we have $\mathbb{E} \left[\left \|  \bf{e}_1 -\bf{e}_2  \right \|^2\right] = \mathbb{E} \left[\left \|  \bf{e}_1 \right \|^2\right] + \mathbb{E} \left[\left \| \bf{e}_2 \right \|^2\right]$.
	 Then from \eqref{eq62} it follows that
	\begin{align} \label{bd-vik}
	&\mathbb{E}_{M,B}\left[ {{{\left\| {\boldsymbol{v}_{i,k}^{s+1} - \nabla J_i\left( {\boldsymbol{\theta}_{i,k}^{s+1}} \right)} \right\|}^2}} \right] \notag \\
	 {\stackrel{(a)}{=}}&\frac{1}{B^2}\sum_{b=1}^{B}\mathbb{E}_{M,B}\left[\left\|{\nabla}J_i\left(\tilde{\boldsymbol{\theta}}_i^s\right)-\nabla J_i\left( {\boldsymbol{\theta}_{i,k}^{s+1}} \right)\right.\right.  \notag \\
	&\left.\left.+g_i\left(\tau_{i,b}\left|{\boldsymbol{\theta}}_{i,k}^{s+1}\right.\right)-\omega\left(\tau_{i,b}\left|{\boldsymbol{\theta}}_{i,k}^{s+1}\right.,\tilde{\boldsymbol{\theta}}_i^s\right)g_i\left(\tau_{i,b}\left|\tilde{\boldsymbol{\theta}}_i^s\right.\right)\right\|^2\right] \notag \\
	+& \frac{1}{M^2}\sum\limits_{j = 1}^M\mathbb{E}_{M,B}\left[\left\|{\nabla J_i\left( {\tilde{\boldsymbol{\theta}}_i^s} \right) -  {{g_i}\left( { {\tilde{\tau}_{i,j}}\left| {\tilde{\boldsymbol{\theta}}_i^s} \right.} \right)} }\right\|^2\right] \notag \\
	\stackrel{(b)}{\le}& \frac{1}{B^2}\sum_{b=1}^{B}\mathbb{E}_{M,B}\left[\left\|\omega\left(\tau_{i,b}\left|{\boldsymbol{\theta}}_{i,k}^{s+1}\right.,\tilde{\boldsymbol{\theta}}_i^s\right)g_i\left(\tau_{i,b}\left|\tilde{\boldsymbol{\theta}}_i^s\right.\right)\right.\right.  \notag \\	&\left.\left.-g_i\left(\tau_{i,b}\left|{\boldsymbol{\theta}}_{i,k}^{s+1}\right.\right)\right\|^2\right]+\frac{1}{M^2}\sum\limits_{j = 1}^M\mathbb{E}_{M,B}\left[\left\|{  {{g_i}\left( { {\tilde{\tau}_{i,j}}\left| {\tilde{\boldsymbol{\theta}}_i^s} \right.} \right)} }\right\|^2\right] \notag \\
	\stackrel{(c)}{\le}& \frac{1}{B^2}\sum_{b=1}^{B}\mathbb{E}_{M,B}\left[\left\|\omega\left(\tau_{i,b}\left|{\boldsymbol{\theta}}_{i,k}^{s+1}\right.,\tilde{\boldsymbol{\theta}}_i^s\right)g_i\left(\tau_{i,b}\left|\tilde{\boldsymbol{\theta}}_i^s\right.\right)\right.\right. \notag \\	&\left.\left.-g_i\left(\tau_{i,b}\left|{\boldsymbol{\theta}}_{i,k}^{s+1}\right.\right)\right\|^2\right]+\frac{V}{M}, \end{align}
	 where in (a) we use
	\begin{align*}
	{\bf{e}_2}=\frac{1}{M}\sum_{j = 1}^M {\bf{x}}_j, \text{ with } {\bf{x}}_j = {\nabla J_i\left( {\tilde{\boldsymbol{\theta}}_i^s} \right) -  {{g_i}\left( { {\tilde{\tau}_{i,j}}\left| {\tilde{\boldsymbol{\theta}}_i^s} \right.} \right)} },
	\end{align*}
	since the trajectories $\{\tilde{\tau}_{i,j}\}_{j=1}^{M}$ are independent from each other and therefore $\mathbb{E}\left[\left\|\frac{1}{M}\sum_{j = 1}^M {\bf{x}}_j\right\|^2\right] =  \frac{1}{M^2}\sum_{j = 1}^M\mathbb{E}\left[\left\| {\bf{x}}_j\right\|^2\right]$ for independent and zero mean variables $\{{\bf{x}}_j\}_{j=1}^{M}$. The same to ${\bf{e}_1}$.
	In (b) we use the standard conditional variance decomposition that $\mathbb{E} \left[\left \| \bf{x}-\mathbb{E}\left[\bf{x}\right]  \right \|^2\right]\leq  \mathbb{E} \left[\left \|\bf{x}\right \|^2\right]$, and in (c) we use Assumption \ref{Asu2}.	
	
	Substituting (\ref{eq32}) into \eqref{bd-vik} yields
	\begin{align}\label{eq63}
	&\mathbb{E}_{M,B}\left[ {{{\left\| {\boldsymbol{v}_{i,k}^{s+1} - \nabla J_i\left( {\boldsymbol{\theta}_{i,k}^{s+1}} \right)} \right\|}^2}} \right] \notag \\
	&\leq \frac{2\left(C_g^2{C_\omega }+{L_g^2}\right)}{B}\left\|\boldsymbol{\theta}_{i,k}^{s+1}-\tilde{\boldsymbol{\theta}}_i^s \right\|^2 + \frac{V}{M}.
	\end{align}
	For the gradient estimation error, we have
	\begin{align}
	&\mathbb{E}_{M,B}\left[\left\|\bar{\boldsymbol{v}}_{k}^{s+1}-\overline{\nabla J}\left(\boldsymbol{\theta}_{k}^{s+1}\right)\right\|^{2}\right] \notag \\
	=& \frac{1}{n^2}\mathbb{E}_{M,B}\left[\left\|\sum_{i=1}^n \left( \boldsymbol{v}_{i,k}^{s+1}-\nabla J_{i}\left(\boldsymbol{\theta}_{i,k}^{s+1}\right) \right)\right\|^{2}\right] \notag \\
	=& \frac{1}{n^2}\mathbb{E}_{M,B} \left[ \sum_{i=1}^n \left\| \boldsymbol{v}_{i,k}^{s+1}-\nabla J_{i}\left(\boldsymbol{\theta}_{i,k}^{s+1}\right)\right\|^{2} \right. \notag\\
	&\left. + \sum_{i \neq r} \left<\boldsymbol{v}_{i,k}^{s+1}-\nabla J_{i}\left(\boldsymbol{\theta}_{i,k}^{s+1}\right), \boldsymbol{v}_{r,k}^{s+1}-\nabla J_{r}\left(\boldsymbol{\theta}_{r,k}^{s+1}\right)\right>\right] \notag \\
	=& \frac{1}{n^2}\sum_{i=1}^n\mathbb{E}_{M,B} \left[  \left\| \boldsymbol{v}_{i,k}^{s+1}-\nabla J_{i}\left(\boldsymbol{\theta}_{i,k}^{s+1}\right)\right\|^{2} \right], \label{eq64}
	\end{align}
	where in the last equality we use that $\{\boldsymbol{v}_{i,k}^{s+1}\}$ are independent with each other and therefore $\mathbb{E}_{M,B}\left[ \sum\limits_{i \neq r} \left<\boldsymbol{v}_{i,k}^{s+1}-\nabla J_{i}(\boldsymbol{\theta}_{i,k}^{s+1}), \boldsymbol{v}_{r,k}^{s+1}-\nabla J_{r}(\boldsymbol{\theta}_{r,k}^{s+1})\right>\right] =0$.
	By substituting (\ref{eq63}) into (\ref{eq64}) and using the definition of $\Psi$ in (\ref{eq_psi}), we obtain
	\begin{align} &\mathbb{E}_{M,B}\left[\left\|\bar{\boldsymbol{v}}_{k}^{s+1}-\overline{\nabla J}\left(\boldsymbol{\theta}_{k}^{s+1}\right)\right\|^{2}\right] \notag \\
	\leq& \frac{2\left(C_g^2{C_\omega }+{L_g^2}\right)}{B}\frac{1}{n^2}\sum_{i=1}^n \left\|\boldsymbol{\theta}_{i,k}^{s+1}-\tilde{\boldsymbol{\theta}}_i^s \right\|^2 + \frac{V}{Mn} \notag \\
	\leq& \frac{\Psi}{B n^2}\left\|\boldsymbol{\theta}_{k}^{s+1}-\tilde{\boldsymbol{\theta}}^s \right\|^2 + \frac{V}{Mn}. \label{eq64-1}
	\end{align}	
	Applying the upper bound in (\ref{eq34}) to (\ref{eq64-1}) and taking the expectations of the resulting inequality lead to (\ref{eq61}).
\end{proof}



\ifCLASSOPTIONcaptionsoff
  \newpage
\fi



\bibliographystyle{IEEEtran}
\bibliography{IEEEabrv,myIEEE}

\end{document}